\documentclass[11pt]{article}
\usepackage{amssymb,amsmath,amsthm,mathtools}
\usepackage{array,geometry}  
\usepackage{setspace}
\usepackage{bm,paralist,dsfont,nicefrac}
\usepackage{etoolbox}
\usepackage[normalem]{ulem}
\usepackage{xfrac}

\usepackage{tikz,float,booktabs,pifont,multirow,subcaption}
\usepackage{makecell}
\usetikzlibrary{fit,calc,math,shapes,shapes.multipart,decorations.text,arrows,decorations.markings,decorations.pathmorphing,shapes.geometric,positioning,decorations.pathreplacing, patterns}

\usepackage[T1]{fontenc}
\usepackage{natbib}

\usepackage[pagebackref,colorlinks]{hyperref}
\hypersetup{linkcolor=[rgb]{.5,0,0}}
\hypersetup{citecolor=[rgb]{0,0,0.5}}
\hypersetup{urlcolor=[rgb]{.7,0,.7}}
\renewcommand*{\backref}[1]{}
\renewcommand*{\backrefalt}[4]{%
    \ifcase #1 (Not cited.)%
    \or        (Cited on page~#2)%
    \else      (Cited on pages~#2)%
    \fi}

\usepackage[capitalize, nameinlink, noabbrev]{cleveref} 

\let\oldcite\cite
\renewcommand{\cite}{\oldcite*}

\AddToHook{cmd/appendix/before}{%
	\crefalias{section}{appendix}%
	\crefalias{subsection}{appendix}
}

\usepackage{xcolor}

\usepackage{thmtools,thm-restate}
\usepackage{enumitem}
\usepackage{nicematrix}

\allowdisplaybreaks

\newcommand{\APXH}{\textrm{\textup{APX-hard}}}

\newcommand{\eps}{\varepsilon}
\newcommand{\EPTAS}{\textrm{\textup{EPTAS}}}

\newcommand{\FPTAS}{\textrm{\textup{FPTAS}}}

\newcommand{\NP}{\textrm{\textup{NP}}}
\newcommand{\NPH}{\textrm{\textup{NP-hard}}}

\newcommand{\NSW}{\textrm{\textup{NSW}}}
\renewcommand{\O}{\mathcal{O}}

\newcommand{\poly}{\textrm{\textup{poly}}}

\newcommand{\PTAS}{\textrm{\textup{PTAS}}}

\definecolor{DarkGreen}{rgb}{0.1,0.5,0.1}

\newcommand{\tail}{\texttt{tail}}

\newcommand{\V}{\mathcal{V}}
\newcommand{\VV}{\mathrm{V}}
\newcommand{\BB}{\mathrm{B}}
\newcommand{\W}{\mathcal{W}}
\newcommand{\WH}{\mathcal{W}^{h}}
\newcommand{\WNash}{\W^\texttt{\textup{Nash}}}

\newcommand{\abs}[1]{\lvert #1 \rvert}

\usepackage[linesnumbered,ruled,vlined]{algorithm2e}

\newtheorem{theorem}{Theorem}
\newtheorem*{theorem*}{Theorem}

\newtheorem{definition}{Definition}
\newtheorem*{definition*}{Definition}

\newtheorem{lemma}{Lemma}
\newtheorem*{lemma*}{Lemma}

\newtheorem{claim}{Claim}
\newtheorem*{claim*}{Claim}

\newtheorem*{fact*}{Fact}

\newtheorem*{observation*}{Observation}

\newtheorem*{conjecture*}{Conjecture}

\newtheorem{corollary}{Corollary}
\newtheorem*{corollary*}{Corollary}

\newtheorem{remark}{Remark}
\newtheorem*{remark*}{Remark}

\newtheorem{proposition}{Proposition}
\newtheorem*{proposition*}{Proposition}

\newtheorem{example}{Example}
\newtheorem*{example*}{Example}

\Crefname{algocf}{Algorithm}{Algorithms}

\usepackage{bbm}

\DeclareMathOperator*{\sign}{sign}

\title{Easier, but Not Easy: \\Nash Welfare under Lexicographic Valuations}

\author{
\begin{tabular}{m{0.12\linewidth}m{0.12\linewidth}m{0.12\linewidth}m{0.12\linewidth}m{0.12\linewidth}m{0.12\linewidth}}
\multicolumn{2}{c}{\textbf{Soumil Aggarwal}} & \multicolumn{2}{c}{\textbf{Rohit Vaish}}  & \multicolumn{2}{c}{\textbf{Jatin Yadav}} \\
         \multicolumn{2}{c}{\small{Quantbox Research}} & \multicolumn{2}{c}{\small{IIT Delhi}} & \multicolumn{2}{c}{\small{IIT Delhi}}\\
         \multicolumn{2}{c}{\href{mailto:soumilaggarwal95@gmail.com}{\small{\texttt{soumilaggarwal95@gmail.com}}}} & \multicolumn{2}{c}{\href{mailto:rvaish@iitd.ac.in}{\small{\texttt{rvaish@iitd.ac.in}}}} & \multicolumn{2}{c}{\href{mailto:jatin.yadav@cse.iitd.ac.in}{\small{\texttt{jatin.yadav@cse.iitd.ac.in}}}}\\
  \end{tabular}
}

\date{}

\usepackage{colortbl}
\colorlet{myred}{red!25}
\colorlet{myblue}{blue!25}
\colorlet{mygreen}{green!25}

\usepackage{amsmath,amssymb,amsthm,mathtools}
\usepackage{enumitem}
\usepackage{xcolor}
\usepackage{tikz}
\usetikzlibrary{arrows.meta,calc,decorations.pathreplacing,positioning}

\newcommand{\val}{\operatorname{val}}

\definecolor{proofblue}{RGB}{54,105,170}
\definecolor{proofgreen}{RGB}{70,135,95}
\definecolor{prooforange}{RGB}{204,126,45}
\definecolor{proofred}{RGB}{180,75,75}

\tikzset{
        proofstage/.style={
                draw=black!55,
                rounded corners=2pt,
                fill=black!4,
                align=center,
                font=\scriptsize,
                inner sep=4pt,
                minimum height=1.05cm
        },
        proofarrow/.style={-{Latex[length=2.2mm]}, thick, draw=black!70},
        proofblock/.style={
                draw=black!55,
                rounded corners=1pt,
                align=center,
                font=\scriptsize,
                inner sep=3pt,
                minimum height=.8cm
        },
        proofcell/.style={
                draw=black!55,
                rounded corners=.5pt,
                align=center,
                font=\tiny,
                inner sep=1pt,
                minimum width=.72cm,
                minimum height=.46cm
        },
        proofcopy/.style={
                draw=prooforange!85!black,
                fill=prooforange!22,
                rounded corners=1pt,
                font=\tiny,
                inner sep=1.2pt
        }
}

\begin{document}

\maketitle
\thispagestyle{empty}
\addtocounter{page}{-1}

\begin{abstract}
Maximizing Nash welfare over indivisible goods is a central problem in resource allocation. For additive valuations, the best-known approximation factor is roughly $e^{-1/e} \approx 0.692$, and the problem is APX-hard. We study Nash welfare maximization under \emph{lexicographic valuations}, equivalently superincreasing additive valuations, where every good is worth more than the total value of all lower-ranked goods. This large-gap structure makes preferences almost ordinal, which might suggest that the problem becomes easy. We show, however, that the picture is more nuanced: although lexicographic valuations enable stronger algorithmic guarantees, they retain significant computational hardness.

Our first main result is a $\left( \sfrac{1}{\sqrt{2}}-\eps \right) \approx (0.707-\eps)$-approximation algorithm for weighted Nash welfare under general lexicographic valuations, improving over the guarantee inherited from additive valuations. The algorithm rounds the configuration LP for Nash welfare, but its analysis relies on structural properties specific to the lexicographic domain, including domination phenomena induced by large value gaps and a careful grouping of agents. We complement this algorithm by showing that the weighted configuration LP has an integrality gap of exactly $\sqrt{2}$ for lexicographic valuations, confirming that the analysis is tight for this relaxation.

Our second main contribution is an exact algorithmic framework for structured lexicographic instances. We introduce a \emph{domination-based branch-and-prune method} that allocates goods in priority order while discarding provably suboptimal branches. Although the naive search tree is exponential, we prove a mutual-exclusion property between sibling subtrees and use a matrix-based leaf-counting argument to bound the pruned recursion tree by a polynomial for every constant number of agents. This yields polynomial-time exact algorithms for ordered lexicographic valuations and for doubling lexicographic valuations. In the ordered case, we also obtain an efficient polynomial-time approximation scheme (EPTAS) for arbitrarily many agents.

Finally, we show that large gaps do not eliminate hardness, as Nash welfare maximization is NP-hard even for ordered lexicographic valuations, and it is NP-hard to obtain a $0.9996$-approximation even for doubling lexicographic valuations. Thus, lexicographic valuations make Nash welfare maximization easier, but not easy: they admit tighter approximation and exact algorithms in important cases, yet still require intricate techniques and preserve some of the hardness of the general additive setting.\end{abstract}

\clearpage

\section{Introduction}
\label{sec:Introduction}

Maximizing Nash welfare is a central algorithmic problem in the allocation of indivisible goods. Given a set of agents and goods, the goal is to allocate the goods so as to maximize the geometric mean of the agents' utilities~\citep*{N50bargaining,KN79nash}. This objective strikes a remarkable balance between efficiency and fairness: it rewards allocations that create high total value, but penalizes imbalance across agents. For divisible goods, a Nash-optimal allocation can be computed via the Eisenberg-Gale convex program and coincides with a competitive equilibrium allocation~\citep*{EG59consensus,V74equity}. For indivisible goods, however, the problem is much more challenging. Even for additive valuations, maximizing Nash welfare is APX-hard~\citep*{Lee17APX}, and the best known polynomial-time approximation factor is nearly $e^{-1/e} \approx 0.692$~\citep*{V26better}.

In this paper, we study Nash welfare maximization under \emph{lexicographic valuations}. These are additive valuations with a strong ``large gap'' structure, where each good is worth more than the total value of all lower-ranked goods. Equivalently, they are superincreasing additive valuations. Thus, for any agent, the comparison between two bundles is determined by the most preferred good on which the bundles differ. In this sense, lexicographic valuations are close to ordinal preferences, while still retaining the cardinal structure needed to define and optimize Nash welfare.

At first glance, this structure might suggest that Nash welfare maximization should become easy. After all, arbitrary cardinal tradeoffs---one of the main sources of difficulty in additive valuations---are largely absent. Lexicographic valuations impose a sharp priority ordering on the goods, as a high-ranked good cannot be compensated for by any collection of lower-ranked goods. This raises a natural question:

\begin{quote}
\emph{Does the large-gap structure of lexicographic valuations make Nash welfare maximization tractable?}
\end{quote}

Our work shows that lexicographic valuations make Nash welfare computation \emph{easier, but not easy}. Specifically, we demonstrate that lexicographic valuations admit stronger algorithms than are currently known for general additive valuations. At the same time, the problem remains computationally hard even under severe restrictions. Moreover, obtaining the improved algorithms requires intricate techniques that exploit the large-gap structure in a nontrivial way: configuration-LP rounding, domination arguments, careful grouping of agents, and a branch-and-prune framework with a polynomially bounded recursion tree.

\subsection{Our results}

We provide three types of results: an improved approximation algorithm for general lexicographic valuations, exact and approximation algorithms for structured subclasses, and hardness results showing that the lexicographic domain retains significant computational complexity. \Cref{tab:Summary}
summarizes our results.

\paragraph{A $\left( \sfrac{1}{\sqrt{2}}-\eps \right)$-approximation for general lexicographic valuations (\Cref{sec:Sq-Root-2-Approx-Algo}).}
Our first main result is a polynomial-time $\left( \sfrac{1}{\sqrt{2}}-\eps \right) \approx (0.707 - \eps)$-approximation algorithm for Nash welfare under general lexicographic valuations. This improves over the nearly $e^{-1/e}$ ($\approx 0.692$)-approximation that follows from the best known algorithms for the broader class of additive valuations.

Our algorithm is based on rounding the \emph{configuration LP} for Nash welfare. Configuration LPs have been a central tool in recent approximation algorithms for weighted Nash welfare and related welfare objectives under additive and more general valuations~\citep*{FL25note,DLR+24constant,FHL+25constant,BFH+26nash,FHL26new,V26better}. In our setting, however, the analysis departs from the additive case by using structural features specific to lexicographic valuations. The large gaps in values induce strong \emph{domination} relations: in many situations, assigning a good to one agent can be certified to be better than assigning it to another, independently of how lower-ranked goods are allocated. Our rounding algorithm combines these domination properties with a careful grouping of agents to control the loss in Nash welfare.

We complement the algorithm with a matching lower bound for the relaxation: the configuration LP for weighted Nash welfare has integrality gap exactly $\sqrt{2}$ for lexicographic valuations. Thus, our $\left( \sfrac{1}{\sqrt{2}}-\eps \right)$ guarantee is tight for this LP. This gives a complete picture of what can be achieved through this relaxation on the lexicographic domain.

\begin{table*}[t]
    \centering
    \small
    \begin{center}
    \begin{tabular}{llll}
    \toprule
    Domain & Constant $n$ & General $n$: approx. & General $n$: hardness \\
    \midrule
    General lexicographic & open & $(\sfrac{1}{\sqrt{2}}-\varepsilon)$ (Th.~\ref{thm:Root-2-Approximation}) & APX-hard via doubling subclass\\
    Ordered lexicographic & Poly-time (Th.~\ref{thm:Ordered_Constant_n_Weighted}) & EPTAS (Th.~\ref{thm:EPTAS}) & NP-hard (Th.~\ref{thm:NP-hard_Ordered_Nash})\\
    Doubling lexicographic & Poly-time (Th.~\ref{thm:Doubling_Constant_n}) & $(\sfrac{1}{\sqrt{2}}-\varepsilon)$ (Th.~\ref{thm:Root-2-Approximation}) & APX-hard within $0.9996$ (Th.~\ref{thm:APX-hard_Nash_General_Lexicographic})\\
    \bottomrule
    \end{tabular}
    \end{center}
    \caption{Summary of our computational results for Nash welfare. Each entry of the table denotes the computational complexity of finding an exact or approximately Nash optimal allocation when the number of agents $n$ is constant or general (specified by the columns) and when the given instance has ordered, doubling, or general lexicographic valuations (specified by the rows). All our algorithmic results, except for the EPTAS, can be generalized to weighted Nash welfare.}
    \label{tab:Summary}
\end{table*}

\paragraph{Exact and approximation algorithms for structured subclasses~(\Cref{sec:Results_Subclasses}).}

Our second main contribution is an exact algorithmic framework for structured lexicographic instances. We introduce a \emph{domination-based branch-and-prune} method that recursively allocates goods in priority order. At each step, the algorithm considers possible recipients for the current good, but prunes any branch in which assigning the good to a particular agent is provably dominated by assigning it to another agent.

Naively, such a recursion has exponential size. The key technical challenge is to prove that the pruning is strong enough to make the search tree polynomially bounded. We do this through a \emph{mutual-exclusion} argument. Informally, if two sibling branches allocate a good to two different agents, then the future branching possibilities in the two subtrees cannot overlap too much: choices available in one subtree are excluded in the other. We encode this phenomenon through a binary matrix whose rows correspond to goods and columns to agents, and use a matrix-based leaf-counting argument to bound the number of leaves in the recursion tree.

This framework yields polynomial-time exact algorithms for every constant number of agents in two important settings. The first is \emph{ordered lexicographic valuations}, where all agents agree on the ranking of goods, although they may assign different numerical values. The second is \emph{doubling valuations}, where agents may have different rankings, but for each agent every good is worth at least twice the next lower-ranked good. In the ordered case, we also obtain an efficient polynomial-time approximation scheme (EPTAS) for an arbitrary number of agents.

Our exact algorithms also extend, in a natural comparison-oracle model, beyond Nash welfare to \emph{separable concave welfare} objectives. In particular, the ordered-valuation algorithm extends to weighted $p$-mean welfare for all $p<1$, while the doubling-valuation algorithm extends to weighted Nash welfare and to weighted $p$-mean welfare for $p\in(0,1)$; see Appendix~\ref{appendix:Extensions_p_Mean} for the formal statement.

\paragraph{Hardness despite large gaps (\Cref{sec:Hardness_Results}).}

Finally, we show that the lexicographic structure does not eliminate hardness. Specifically, we show that Nash welfare maximization is NP-hard even for ordered lexicographic valuations. Thus, even when all agents agree on the priority order of the goods, exact optimization remains computationally intractable for an arbitrary number of agents. This result also shows that our aforementioned branch-and-prune algorithm for a constant number of agents cannot be extended to an arbitrary number of agents unless P = NP.

We further prove that the problem is APX-hard even for doubling lexicographic valuations; specifically, it is NP-hard to obtain a $0.9996$-approximation. We believe this result provides the clearest sense in which lexicographic valuations are ``not easy.'' Even when every agent's values decrease by a factor of at least two from one rank to the next, the problem retains constant-factor hardness.

Together, these results provide a detailed computational landscape for Nash welfare maximization under lexicographic valuations. The domain is structured enough to allow improved approximation and exact algorithms in important cases, but rich enough to preserve nontrivial hardness.

\subsection{Why lexicographic valuations?}

Lexicographic valuations are a natural large-gap regime of additive preferences. They model settings in which goods are evaluated according to priority features that are not meaningfully substitutable. For example, a decision maker may prefer any good satisfying a high-priority criterion to all goods that only satisfy lower-priority criteria. Such preferences arise in priority-based choice, matching markets, machine learning models of decision lists, and behavioral models such as elimination by aspects \citep*{Hausner54,T72,GG96reasoning}. 

From an algorithmic perspective, lexicographic valuations are appealing for a different reason: they form a clear test case for understanding which aspects of additive Nash welfare are responsible for computational difficulty. General additive valuations allow arbitrary tradeoffs among goods, while lexicographic valuations eliminate most of these tradeoffs by imposing a superincreasing structure. If the challenge posed by Nash welfare was solely due to arbitrary cardinal interactions, then this domain should be tractable. Our hardness results, however, indicate that this is not the case.

At the same time, the large-gap structure of lexicographic valuations provides new algorithmic advantages. The domination properties used in both our LP rounding and exact algorithms appear to lack a direct counterpart in the more general setting of additive valuations. Thus, the lexicographic domain presents a more refined picture: its strong priority structure does not trivialize Nash welfare; instead, it introduces enough order to enable sharper algorithms.

The subclasses of \emph{ordered lexicographic} and \emph{doubling} valuations are also practically well-motivated. Ordered preferences capture settings with a public ranking of resources but heterogeneous intensities across agents. Such preferences naturally arise in scenarios such as cloud machines ordered by speed/reliability, ad slots ordered by visibility, interview slots ordered by desirability, and schools/courses/projects with a common prestige or priority ordering. On the other hand, doubling valuations capture agent-specific priority lists with large separations between consecutive ranks.

\subsection{Overview of Our Techniques}
\label{sec:tech_overview}

We now give a high-level overview of the main ideas behind our algorithms and hardness results.

\paragraph{Low-valued goods are important.}

A natural starting point for analyzing Nash welfare with lexicographic valuations is the setting of \emph{identical valuations}, where all agents share the same valuation function. In this case, the large-gap property of lexicographic valuations implies that every Nash optimal allocation must assign the $(n-1)$ most preferred goods to $(n-1)$ distinct agents, while the remaining goods are allocated to the last agent (\Cref{prop:identical-lexicographic}). In particular, this means that the entire ``tail'' of low-valued goods is assigned to a single agent.

Given the ``almost ordinal'' nature of lexicographic valuations, one might expect the same idea to work for the broader class of \emph{ordered valuations}, where agents have the same preference order over goods but assign possibly different values. However, the following example shows that this is not the case.

\begin{table}[H]
    \centering
    \begin{tabular}{c|cccc}
    & $g_1$ & $g_2$ & $g_3$ & $g_4$ \\
    \hline
    $a_1$ & 8 & 4 & 2 & 1 \\
    $a_2$ & 11 & 7 & 2 & 1
    \label{tab:structural_failure}
    \end{tabular}
\end{table}

The Nash optimal allocation in this example assigns $\{g_1,g_4\}$ to agent $a_1$ and $\{g_2,g_3\}$ to agent $a_2$. This example indicates that the numerical aspect of lexicographic valuations should not be ignored when optimizing for Nash welfare.

\paragraph{Bounded lookahead.}

It is known that any Nash optimal allocation is \emph{envy-free up to one good} or EF1~\citep*{CKM+19unreasonable}. This means that for every pair of agents $i$ and $k$, if agent $i$ values agent $k$'s bundle $A_k$ more than its own bundle $A_i$ (in other words, agent $i$ envies agent $k$), then there is some good $g \in A_k$ in the envied bundle whose hypothetical removal restores agent $i$'s preference for its own bundle, i.e., $v_i(A_i) \geq v_i(A_k \setminus \{g\})$.

For lexicographic valuations, the EF1 property is equivalent to requiring that each agent receives at least one good from its top $n$ most preferred goods. Additionally, at least half of each agent's utility is derived from its most preferred good in its bundle. Thus, a natural approach for constructing a Nash optimal allocation is to start with a Nash optimal \emph{matching} of the top $n$ goods to the $n$ agents, and then extend this partial allocation suitably into a complete allocation. This approach readily yields a $\sfrac{1}{2}$-approximation algorithm. 

However, even for ordered valuations, extending a Nash matching into a complete allocation cannot give a better-than-$\sfrac{1}{\sqrt{2}}$ approximation to the optimal Nash welfare. Consider the following example, where $k$ is a large constant and $\varepsilon$ is a small positive constant.

\begin{table}[h]
    \centering
    \begin{tabular}{c|ccc}
    & $g_1$ & $g_2$ & $g_3$ \\
    \hline
    $a_1$ & $k$ & $1+\varepsilon$ & $1$\\
    $a_2$ & $k$ & $1 + 2\varepsilon$ & $\varepsilon$
    \end{tabular}
    \label{tab:matching_fails_root2}
\end{table}

The Nash optimal matching allocates $g_1$ to $a_1$ and $g_2$ to $a_2$. Irrespective of who receives $g_3$, the Nash welfare of the returned allocation is asymptotically $\sqrt{k}$. However, the Nash optimal allocation allocates $g_1$ to $a_2$ and $g_2, g_3$ to $a_1$, and has Nash welfare asymptotically $\sqrt{2k}$.

The example above suggests that finding a Nash optimal allocation for an ordered instance cannot be reduced to just computing a Nash optimal allocation of the top $n$ goods. Instead, a \emph{lookahead} is necessary. For ordered lexicographic valuations, a constant amount of lookahead suffices to obtain a $(1+\varepsilon)$-approximation---specifically, by considering the top $(n+\ell)$ goods, where $\ell = \O(\log \nicefrac{1}{\varepsilon})$, we obtain a polynomial-time approximation scheme (\PTAS{}). A refined analysis of the approximation ratio and runtime allows us to convert this into an efficient polynomial-time approximation scheme~($\EPTAS{}$) (\Cref{thm:EPTAS}). However, any constant-sized lookahead is unlikely to yield an exact algorithm, as maximizing Nash welfare is \NPH{} even for ordered lexicographic valuations~(\Cref{thm:NP-hard_Ordered_Nash}).

\paragraph{Domination-based branching and pruning.} Although lookahead is a useful idea for designing approximation schemes, it still appears limited in leading us toward an exact algorithm. Moreover, the problem of Nash maximization with ordered lexicographic valuations is \NPH{} for an arbitrary number of agents. So, in search of an exact algorithm for ordered instances, it seems reasonable to focus on a constant number of agents. We discussed that the Nash optimal allocation for the overall instance, when restricted to the top $n$ goods, must be a matching, but the restriction may not maximize Nash welfare among all matchings. To bypass this issue, we can iterate over all $n!$~(constant) matchings of the top $n$ goods. 

For the remaining goods, a natural greedy heuristic is to allocate them in decreasing order of preference to the agent with the highest marginal gain. This approach, however, fails to return an optimal allocation, as shown by the following example.

\begin{table}[H]
    \centering
    \begin{tabular}{c|cccc}
    & $g_1$ & $g_2$ & $g_3$ & $g_4$ \\
    \hline
    $a_1$ & 10 & 4 & 2 & $2 -\varepsilon$ \\
    $a_2$ & 15 & 10 & $2 - \varepsilon$ & $\varepsilon$
    \vspace{0.1in}
    \end{tabular}
    \label{tab:greedy_failure}
\end{table}

The optimal allocation here is $A_1 = \{g_1, g_4\}, A_2 = \{g_2, g_3\}$ with Nash product $(12-\varepsilon)^2$. If we fix $g_1 \to a_1$ and $g_2 \to a_2$, the greedy approach allocates both $g_3$ and $g_4$ to $a_1$ based on marginals, resulting in a strictly lower Nash product of $140-10\varepsilon$.

The greedy algorithm offers only \emph{one} branching option per good---specifically, the agent with the best marginal utility---whereas a brute force algorithm explores $n$ options. A significant technical contribution of our work is an algorithmic framework based on a recursive algorithm that strikes a balance between these two extremes. This framework allocates goods in decreasing order of preference, but employs strict pruning conditions to eliminate most of the $n$ branches. The pruning is based on a \emph{domination} condition, which is the key structural insight underlying many of our algorithmic as well as hardness results. 

Suppose a partial allocation has already assigned some high-priority goods, and we are deciding where to allocate the next good. We say that an agent $k$ \emph{dominates} an agent $i$ with respect to good $g$ if, no matter how the remaining lower-ranked goods are allocated, assigning good $g$ to $k$ leads to a better objective value than assigning it to $i$. In that case, no optimal completion assigns the good to $i$, and the corresponding branch can be safely pruned. The resulting algorithm recursively allocates goods in priority order and keeps only undominated branches.

A major challenge is to establish a polynomial bound on the runtime of this recursive algorithm. This is non-trivial because even a branching factor of 2 can result in an exponential number of leaves. Our key structural insight is that the recursion tree satisfies a \emph{mutual-exclusion} property. If two sibling branches make different choices for a good $g$ (i.e., good $g$ is given to agent $i$ in one branch and agent $k$ in the other), then future goods cannot continue branching in the same way in both subtrees. That is, for any subsequent good $g'$, either the subtree rooted at the first sibling prunes the allocation of $g'$ to agent $i$, or the subtree rooted at the second sibling prunes the allocation of $g'$ to agent $k$. We formalize this using a binary matrix that records which agents can receive which goods within a subtree. A leaf-counting lemma, based on pairwise dot products of the matrix columns, then gives a polynomial bound on the number of leaves for a constant number of agents (\Cref{lem:leaf_bound}).

Our algorithmic framework is robust enough to handle ordered lexicographic valuations and doubling valuations. In the ordered case, all agents agree on the priority order, so the recursion follows this common order. In the doubling case, agents may disagree on rankings; nevertheless, the factor-two separation is strong enough to identify goods whose allocation can be branched on while preserving the mutual-exclusion structure (\Cref{thm:Doubling_Constant_n}).

\paragraph{Configuration-LP rounding.}

The starting point for our approximation algorithm is the configuration LP for Nash welfare (\Cref{thm:Root-2-Approximation}). At a high level, the LP assigns fractional bundles to agents, and its value upper-bounds the optimum integral Nash welfare. For additive valuations, configuration-LP rounding has been used in recent work to obtain the best known approximations for Nash welfare and its variants. Our goal is to exploit the stronger structure of lexicographic valuations to obtain a better guarantee.

The configuration-LP rounding proof can be viewed as a comparison between two allocations in a copied auxiliary instance. The fractional LP solution induces a valid allocation $X^V$, while the rounded solution induces a block-balanced allocation $X^B$ in which each agent receives at most one good from each block of its priority order. For general additive valuations, the loss in this comparison is controlled through an EF1 argument~\citep*{FL25note}. For lexicographic valuations, we prove a sharper comparison: after a sequence of transformations that can only increase the ratio $\NSW(X^V)/\NSW(X^B)$, the agents can be partitioned into groups of size at most three, each of which loses a factor at most $\sqrt 2$ in geometric mean.  This is where the large-gap property is crucial: it allows us to certify transfers and groupings by comparing only the decisive, high-ranked goods.

\paragraph{Hardness.}

Our hardness results show that the tractability obtained above is close to the limit of what one can hope for. For ordered lexicographic valuations, we reduce from the Perfect $3$-Dimensional Matching problem and use carefully chosen superincreasing values to force any high-welfare allocation to encode a feasible perfect $3$D matching (\Cref{thm:NP-hard_Ordered_Nash}). For doubling valuations, we construct a gap-preserving reduction from $4$-Dimensional Matching, showing that even very strong separation between consecutive values does not preclude APX-hardness (\Cref{thm:APX-hard_Nash_General_Lexicographic}). These reductions demonstrate that the large-gap domain retains enough combinatorial structure to simulate hard allocation choices.

\subsection{Related Work}

There is a rich literature on the computational aspects of maximizing Nash welfare. 
Specifically, for indivisible goods, maximizing Nash welfare is known to be \APXH{}~\citep*{NNR+14computational,Lee17APX,GHM24satiation,V26improved}. Specifically, it is \NPH{} to maximize Nash welfare to within a factor of $\sqrt{\sfrac{7}{8}} \approx 0.935$ under additive valuations~\citep*{GHM24satiation}. Additionally, assuming the Unique Games Conjecture, it is \NPH{} to maximize Nash welfare to within a factor of $\sqrt[3]{\sfrac{65}{81}} \approx 0.929$~\citep*{V26improved}.

\paragraph{Algorithms for Nash welfare.} A long line of work has explored the development of approximation algorithms for Nash welfare~\citep*{CDG+17convex,CG18approximating,BKV18Finding,GHM+22tractable,GKK23approximating,FL25note,FHL+25constant,ACH+26maximizing,GHL+26approximating,V26better}. The best-known approximations for weighted Nash welfare stand at $e^{-1/e}\approx 0.692$ for additive valuations~\citep*{FL25note} and $(0.281-\varepsilon)$ for submodular valuations~\citep*{BFH+26nash}. In some restricted settings such as matroid-rank valuations~\citep*{VZ25general}, a weighted Nash optimal allocation can be computed in polynomial time. For additive valuations and unweighted Nash welfare, \citet{V26better} recently improved upon the $e^{-1/e}$ factor by a small constant, estimated by the author to be approximately $10^{-80}$.

\paragraph{Lexicographic Valuations.}

As mentioned previously, lexicographic valuations are a subdomain of additive valuations. \citet*{DS15maximizing} showed that maximizing Nash welfare is \NPH{} for lexicographic instances where the valuations of each agent are consecutive powers of $2$. \citet*{BBL+17positional} showed \NPH{}ness of maximizing the egalitarian welfare under lexicographic preferences. Several studies have explored the existence and computation of approximately envy-free and Pareto optimal allocations for lexicographic valuations~\citep*{HSV+21fair,EPS22fairly,HSV+23fairly,HAW23fairly,HAW24almost,ALM+24almost}. 

\paragraph{Configuration LP rounding.} Configuration LPs have proven to be a versatile tool for Nash welfare and related welfare objectives.  The $e^{-1/e}$ approximation for weighted additive valuations of \citet*{FL25note} is obtained by solving a weighted configuration LP approximately and then applying the rounding technique of~\citet*{3113606.3113856}. Subsequent work has used and extended LP-based rounding for Nash welfare and its variants under additive, submodular, and other valuation classes~\citep*{DLR+24constant,FL25note,FHL+25constant,BFH+26nash,V26better}. Our approximation algorithm uses the same rounding template as \citet*{FL25note}, but replaces the additive-domain EF1 loss by a sharper comparison theorem for lexicographic valuations.  

In \Cref{appendix:Additional_Related_Work}, we discuss additional related work on the class of $p$-mean welfare measures.

\subsection{Organization}

\Cref{sec:Preliminaries} introduces the model, lexicographic and doubling valuations, and Nash welfare. \Cref{sec:Sq-Root-2-Approx-Algo} presents the $\left( \sfrac{1}{\sqrt{2}}-\eps \right)$-approximation algorithm and the matching integrality-gap construction. \Cref{sec:Results_Subclasses} develops the domination-based branch-and-prune framework and applies it to ordered and doubling valuations. Additionally, this section discusses the EPTAS for ordered lexicographic valuations. \Cref{sec:Hardness_Results} proves the NP-hardness and APX-hardness results.

\section{Preliminaries}
\label{sec:Preliminaries}

Given any $r \in \mathbb{N}$, let $[r] \coloneqq \{1,2,\dots,r\}$.

\paragraph{Problem instance.} An \emph{instance} of our problem is specified by a tuple $\langle N, M, \V \rangle$, where $N \coloneqq [n]$ is the set of $n$ \emph{agents}, $M \coloneqq \{g_1,g_2,\dots,g_m\}$ is the set of $m$ indivisible \emph{goods}, and $\V \, \coloneqq (v_1, \dots, v_n)$ is a \emph{valuation profile} that specifies the preferences of each agent $i \in N$ in terms of its valuation function $v_i: 2^M \rightarrow \mathbb{Q}_{\ge 0}$. We will assume that each agent values the empty set at~$0$, i.e., for every $i \in N$, $v_i(\emptyset) = 0$. For simplicity, we will write $v_i(g_j)$ in place of $v_i(\{g_j\})$. We say that a valuation function $v_i$ is \emph{additive} if, for every agent $i \in N$ and for every bundle $S \subseteq M$, we have $v_i(S) = \sum_{g \in S} v_i(g)$.

\paragraph{Allocation and utility.} A \emph{bundle} refers to any subset $S \subseteq M$ of the set of goods. A \emph{partial allocation} $A = (A_1, \dots, A_n)$ is an ordered subpartition of the set of goods $M$ into $n$ sets, where $A_i\subseteq M$ is the bundle assigned to agent $i$ and for any $i,k \in N$, we have $A_i \cap A_k = \emptyset$. A partial allocation $A$ is said to be \emph{complete} if $\bigcup_{i \in N} A_i = M$. A complete allocation $B$ is called a ``completion'' of a partial allocation $A$, if $A_i \subseteq B_i$ for all $i \in N$. For ease of exposition, we will use the term `allocation' to refer to a necessarily complete allocation and write `partial allocation' otherwise. Given a partial allocation $A = (A_1, \dots, A_n)$, the \emph{utility} of an agent $i \in N$ is the value it derives under $A$, namely $v_i(A_i)$.  

\paragraph{Lexicographic valuations.}
The valuation function $v$ is said to be \emph{lexicographic} if, 
for all $i \in [m-1],$ we have $v (g_i) > v (g_{i+1}) + v (g_{i+2}) + \ldots + v (g_m)$, where $g_i$ denotes the $i^\text{th}$ most valuable good. 
Under the superincreasing inequalities above, two bundles are compared by the highest-ranked good in their symmetric difference; hence the induced ranking of bundles is strict. 
An instance $\langle N, M, \V \rangle$ is said to be lexicographic if all valuation functions $v_1,v_2,\dots,v_n$ are lexicographic.

\paragraph{Ordered lexicographic valuations.} An instance $\langle N, M, \V \rangle$ is said to be \emph{ordered lexicographic} if it is lexicographic and there exists a permutation of the goods, say $(g_1,g_2,\dots,g_m)$, such that for every agent $i \in N$ and for every index $j \in [m-1]$, we have $v_i(g_{j}) > v_i(g_{j+1})$. In other words, goods can be ordered so that all agents prefer the goods according to that order. Note that in an ordered lexicographic instance, all agents have the same \emph{ranking} over all bundles of goods, though the numerical values assigned by different agents to a bundle may differ. For any $j \in [m-1]$, define $\tail(j) \coloneqq \{g_{j+1},\dots,g_m\}$ to be the set of goods that are less valuable than $g_j$ for all agents.

\paragraph{Doubling valuations.}
An additive valuation function is said to be \emph{doubling} if the goods can be ranked from most preferred to least preferred such that every good in this ranking is worth at least twice as much as the good that appears immediately next to it. An instance $\langle N, M, \V \rangle$ is said to be doubling if all valuation functions $v_1,v_2,\dots,v_n$ are doubling. Note that each doubling valuation is also lexicographic, and hence doubling valuations form a subclass of lexicographic valuations. This is because for any doubling valuation $v$, where goods are preferred in the order $g_1 \succ g_2 \succ \dots \succ g_m$, we have $v(g_j) > \sum_{r=1}^{m-j} \sfrac{v(g_j)}{2^r} \geq v(g_{j+1}) + v(g_{j+2}) + \dots + v(g_m)$ for every $j \in [m-1]$. However, a doubling instance need not be ordered as agents may rank the goods differently.

\paragraph{Nash welfare.}
The Nash welfare of a partial allocation $A$, denoted by $\W^\texttt{\textup{Nash}}(A)$, is the geometric mean of the utilities of the agents under $A$, i.e., $$\W^\texttt{\textup{Nash}}(A) \coloneqq \left( \prod_{i \in N} \, v_i(A_i) \right)^{\nicefrac{1}{n}}.$$ We will use the term \emph{Nash product} to denote the product of agents' utilities. We will assume throughout that $m\ge n$. Indeed, if $m<n$, every allocation has Nash product zero, so the problem becomes degenerate.

For any $\alpha \in [0,1]$, a partial allocation $A$ is said to be \emph{$\alpha$-Nash optimal} if its Nash welfare is at least $\alpha$ times the maximum Nash welfare achieved by any partial allocation for the given instance. A $1$-Nash optimal allocation will simply be called \emph{Nash optimal}.

We will also define a weighted generalization of Nash welfare. Given a set of positive weights $w_1,\dots,w_n$ such that $\sum_{i \in N} w_i = 1$, the weighted Nash welfare of a partial allocation $A$, denoted by $\W^\texttt{\textup{w-Nash}}(A)$, is given by
$$\W^\texttt{\textup{w-Nash}}(A) \coloneqq \prod_{i \in N} \, v_i(A_i)^{w_i}.$$ 
When all weights are equal to $1/n$, we recover the (unweighted) Nash welfare defined earlier.

\paragraph{Generalized $p$-mean welfare.} For any $p \in \mathbb{R}\setminus\{0\}$, the \emph{$p$-mean welfare} of a partial allocation $A$ is the generalized $p$-mean of the utilities of the agents under $A$, i.e., $\mathcal{W}^{p}(A) \coloneqq \left( \frac{1}{n} \sum_{i \in N} v_i^p(A_i) \right)^{\nicefrac{1}{p}}$. Additionally, given positive weights $w_1, \dots, w_n$ such that $\sum_{i \in N} w_i = 1$, the \emph{weighted $p$-mean welfare} is the generalized weighted $p$-mean of the utilities of the agents under $A$, i.e., $\mathcal{W}^{\textup{\texttt{w}-}{p}}(A) \coloneqq \left( \sum_{i \in N} w_iv_i^p(A_i) \right)^{\nicefrac{1}{p}}$. The special cases of $p$-mean welfare for $p=1$, $p \rightarrow -\infty$, and $p \rightarrow 0$ correspond to utilitarian, egalitarian, and Nash welfare, respectively~\citep*{M04fair}. To ensure that $p$-mean welfare is well-defined for $p < 0$, we will only consider partial allocations where all agents derive positive utility when working with welfare measures with $p < 0$.

\paragraph{Separable concave welfare.}
Let $h_1,\dots,h_n:\mathbb{R}_{\ge 0}\rightarrow\mathbb{R} \cup \{-\infty\}$ be increasing concave functions. We define the \emph{separable concave welfare} of a partial allocation $A=(A_1,\dots,A_n)$ as:
\[
\WH(A)\;\coloneqq\;\sum_{i\in N} h_i\!\left(v_i(A_i)\right).
\]

For Nash welfare, taking $h_i(x)=\log x$ gives $\WH(A) = n \cdot \log \W^\texttt{\textup{Nash}}(A)$, so maximizing $W^h$ is equivalent to maximizing Nash welfare. Similarly, $h_i(x) = w_i \cdot \sign(p) \cdot x^p$ corresponds to the weighted $p$-mean welfare for any $p \in (-\infty, 1) \setminus \{0\}$. For some of our results, we will find it convenient to work with a general concave function $h$ instead of $\log(\cdot)$, as that would facilitate extensions to the class of generalized $p$-mean welfare measures (see \Cref{appendix:Extensions_p_Mean}).

\paragraph{Marginal function.} To analyze our algorithms, we will need to track how $\WH$ changes when a single good $g$ is transferred from one agent, say $i$, to another agent, say $k$. Under additivity, the recipient's utility increases by $v_k(g)$ and the donor's utility decreases by $v_i(g)$, so the resulting welfare change depends only on each agent's current utility and this increment/decrement. In particular, for a single transfer, the net change in $\WH$ can be expressed as the recipient's marginal gain minus the donor's marginal loss. We capture this via marginal functions.

For an increasing concave function $f: \mathbb{R}_{\ge 0} \rightarrow \mathbb{R}$, we define its marginal function $\Delta_f: \mathbb{R}_{\ge 0} \times \mathbb{R}_{\ge 0} \rightarrow \mathbb{R}$ as $\Delta_f(x,y) = f(x+y) - f(x)$ for every $x,y \in \mathbb{R}_{\ge 0}$. Note that $\Delta_f(\cdot,y)$ is a non-increasing function for every fixed $y \in \mathbb{R}_{\ge 0}$, and $\Delta_f(x,\cdot)$ is a non-decreasing function for every fixed $x \in \mathbb{R}_{\ge 0}$.

\paragraph{Computational model for separable-concave extensions.}
All item values in the input are rational numbers, and the approximation and hardness results for Nash welfare are stated in the standard Turing model.  For the exact algorithms that are stated for arbitrary separable concave objectives $\WH(A)=\sum_{i\in N} h_i(v_i(A_i))$, we use a \emph{comparison-oracle model} for the functions $h_i$. The oracle is given rational arguments and compares (i) two sums $\sum_i h_i(x_i)$ and (ii) two marginals $\Delta_i(x,y)$ and $\Delta_k(x',y')$.
Here all arguments are rational numbers that arise as utilities of bundles or as sums of item values in the input instance.  The running times of the exact algorithms count both arithmetic operations and calls to this oracle.

For unweighted Nash welfare, this oracle is not needed. Since maximizing Nash welfare is equivalent to maximizing the Nash product, all required objective comparisons can be implemented exactly by comparing products of rational utilities. Thus, our exact Nash-welfare algorithms are ordinary polynomial-time algorithms in the standard Turing model.  For weighted Nash welfare and weighted $p$-mean welfare, the corresponding extensions should be read either in the above comparison-oracle model or under any explicit representation of the functions $h_i$ that supports the stated comparisons in polynomial time.

\section{A \texorpdfstring{$\left( \sfrac{1}{\sqrt{2}}-\eps \right)$}{sqrt2}-Approximation Algorithm}
\label{sec:Sq-Root-2-Approx-Algo}

In this section, we present a $(\nicefrac{1}{\sqrt{2}}-\eps)$-approximation algorithm for weighted Nash welfare under lexicographic valuations~(\Cref{thm:Root-2-Approximation}).

\begin{restatable}[Approximation algorithm]{theorem}{RootTwoApproximation} 
For any fixed $\eps > 0$, there is a polynomial-time algorithm that, given as input any instance with lexicographic valuations, and weights of agents, returns an allocation whose weighted Nash welfare is at least $\left( \sfrac{1}{\sqrt{2}}-\eps \right)$ times the optimal weighted Nash welfare.
    \label{thm:Root-2-Approximation}
\end{restatable}

Our algorithm is the configuration-LP rounding algorithm of~\citet*{FL25note} described in~\Cref{sec:feng_li_rounding}; our contribution is a sharper analysis of its loss for lexicographic valuations. Let $y_i=(y_{i,S})_S$ denote the bundle distribution assigned to agent $i$ by the configuration LP, and let $z_i=(z_{i,S})_S$ denote the bundle distribution produced for that agent by the rounding. Their analysis shows that if, for every agent $i$,
\[
        \sum_S z_{i,S}\ln v_i(S)
        \ge
        \sum_S y_{i,S}\ln v_i(S)-c,
\]
then the rounding gives an $e^{-c}$-approximation, up to the error in solving the configuration LP. To establish such a bound,~\citet*{FL25note} fix an agent and create a copied auxiliary instance with identical clones of the agent, and multiple copies of each good, in which the LP bundle distribution and the rounded bundle distribution are represented by two allocations. The rounded allocation is EF1. For identical additive valuations, every EF1 allocation is an $e^{-1/e}$-approximation to the optimal Nash welfare~\citep*{BKV18Finding}; this yields $c=1/e$. Since $e^{-1/e}<1/\sqrt{2}$, EF1 alone is not enough to obtain our desired guarantee.

For lexicographic valuations, we exploit stronger structure in the rounded distribution. Making this structure visible requires a more careful construction of both the copied auxiliary instance and the two allocations: we create and assign the copied goods block by block, while preserving the LP and rounded bundle distributions. This ensures that the rounded allocation is \emph{block-balanced}, meaning that the goods in each block are allocated to distinct clones. Our main structural result,~\Cref{thm:ratio_root_2} in~\Cref{sec:aux_ratio_bound}, shows that a block-balanced allocation loses a factor of at most $\sqrt{2}$ relative to the LP allocation. Taking logarithms gives the desired bound $c=(\ln 2)/2$, and hence the approximation factor $1/\sqrt{2}-\eps$ in~\Cref{thm:Root-2-Approximation}. \Cref{sec:feng_li_rounding} gives the refined auxiliary construction; \Cref{sec:aux_ratio_bound} formalizes block-balancedness and the resulting allocation comparison and states the structural theorem, whose full proof appears in~\Cref{sec:root_2_proof}.

We also show a matching lower bound on the integrality gap of the LP relaxation. The lower-bound construction is given in \Cref{appendix:config_lp_gap}.

\begin{restatable}[Configuration-LP integrality gap]{proposition}{ConfigLPGAP}
The weighted configuration LP for Nash welfare has integrality gap exactly $\sqrt{2}$, even when all valuations are lexicographic.
\label{prop:config_lp_gap}
\end{restatable}

\subsection{The Feng-Li Configuration LP and its Rounding}
\label{sec:feng_li_rounding}
\paragraph{The relaxation.} ~\citet*{FL25note} use the following weighted configuration LP for additive valuations, where $N = \{1, 2, \ldots, n\}$ is the set of agents, $M$ is the set of goods and $v_i$ is the valuation function for agent $i$:\footnote{We assume that there exists an allocation that gives a positive utility to each agent, and thus only consider those $S$ for each $i$ for which $v_i(S) > 0$.}

\[
\begin{array}{lll}
\max & \displaystyle \sum_{i\in N}\sum_{S\subseteq M}
        w_i y_{i,S}\ln v_i(S) \\
\text{s.t.}
    & \displaystyle \sum_{i\in N}\sum_{S\ni j} y_{i,S}\le 1
        & \forall j\in M, \\
    & \displaystyle \sum_{S\subseteq M} y_{i,S}=1
        & \forall i\in N, \\
    & y_{i,S}\ge 0
        & \forall i,S.
\end{array}                                      \tag{Conf-LP}
\]

The LP can be solved in polynomial time to within an additive error of $\ln(1 + \eps)$ via an approximate separation oracle for the dual. The number of non-zero variables in the solution is polynomially bounded. Define $x_{i,j}$ to be the fraction of the good $j$ given to agent $i$, that is $x_{i,j} = \sum_{S \ni j} y_{i,S}$. They note that the inequality $\sum_{i \in N} \sum_{S \ni j} y_{i,S}$ can be assumed as an equality. Thus, $\sum_{i \in N} x_{i,j} = 1$ for all goods $j \in M$.

\paragraph{Shmoys-Tardos grouping.} For each agent $i$, let its preference over the goods be $\sigma_{i,1} \succ \sigma_{i,2} \ldots \succ \sigma_{i,m}$.  Let $p_i = \lceil \sum_j x_{i,j}\rceil$, and let $q_i = i(m+1)$ be an offset we shall use for agent $i$. As in the Shmoys--Tardos grouping step~\citep*{3113606.3113856}, we make $p_i$ groups of goods from left to right, denoted by $G_{q_i},\ldots,G_{q_i+p_i-1}$. To define these groups, imagine that for agent $i$, each good $j = \sigma_{i,\ell}$ is represented by an interval $I_{i, j} = [q_i + \sum_{h < \ell} x_{i, \sigma_{i,h}}, q_i + \sum_{h \leq \ell} x_{i, \sigma_{i,h}}]$ of length $x_{i, j}$. For each integer $q \in [q_i, q_i+p_i-1]$, group $G_q$ is represented by the interval $[q,q+1]$ and consists of the goods whose intervals have positive-length intersection with $[q,q+1]$. For each good $j\in M$, let $f_{j,q} = |I_{i,j} \cap [q,q+1]|$ denote the length of the overlap between good $j$ and group $G_q$. The offset ensures that intervals for different agents do not intersect: since $p_i\le m$, all intervals for agent $i$ lie in $[q_i,q_i+p_i]\subseteq[q_i,q_i+m]$, while the next agent's offset is $q_i+m+1$.

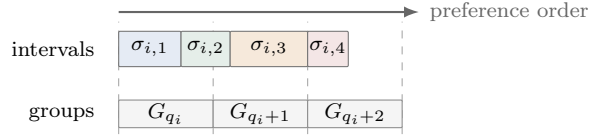
\begin{figure}[t]
\centering
\begin{tikzpicture}[x=1.25cm,y=.85cm]
        \node[font=\scriptsize, anchor=east] at (-.15,.28) {intervals};
        \node[font=\scriptsize, anchor=east] at (-.15,-.74) {groups};
        \draw[-{Latex[length=2mm]}, thick, black!60] (0,.85) -- (3.18,.85)
                node[right, font=\scriptsize] {preference order};

        \foreach \x in {0,1,2,3} {
                \draw[dashed, black!35] (\x,-1.05) -- (\x,.70);
        }

        \path[draw=black!55, fill=proofblue!13, rounded corners=.5pt]
                (0,0) rectangle (.66,.55);
        \node[font=\scriptsize] at (.33,.28) {$\sigma_{i,1}$};
        \path[draw=black!55, fill=proofgreen!14, rounded corners=.5pt]
                (.66,0) rectangle (1.18,.55);
        \node[font=\scriptsize] at (.92,.28) {$\sigma_{i,2}$};
        \path[draw=black!55, fill=prooforange!17, rounded corners=.5pt]
                (1.18,0) rectangle (2,.55);
        \node[font=\scriptsize] at (1.59,.28) {$\sigma_{i,3}$};
        \path[draw=black!55, fill=proofred!14, rounded corners=.5pt]
                (2,0) rectangle (2.43,.55);
        \node[font=\scriptsize] at (2.215,.28) {$\sigma_{i,4}$};

        \path[draw=black!45, fill=black!4, rounded corners=.5pt]
                (0,-.95) rectangle (1,-.52);
        \node[font=\scriptsize] at (.5,-.735) {$G_{q_i}$};
        \path[draw=black!45, fill=black!4, rounded corners=.5pt]
                (1,-.95) rectangle (2,-.52);
        \node[font=\scriptsize] at (1.5,-.735) {$G_{q_i+1}$};
        \path[draw=black!45, fill=black!4, rounded corners=.5pt]
                (2,-.95) rectangle (3,-.52);
        \node[font=\scriptsize] at (2.5,-.735) {$G_{q_i+2}$};
\end{tikzpicture}
\caption{Formation of the groups for agent $i$. Goods are laid out from the offset $q_i$ in the agent's preference order with interval lengths $x_{i,\sigma_{i,\ell}}$; each unit window represents a group $G_q$, and the overlap of good $j$ with group $G_q$ has length $f_{j,q}$. In this example, $G_{q_i}$ contains $\sigma_{i,1},\sigma_{i,2}$; $G_{q_i+1}$ contains $\sigma_{i,2},\sigma_{i,3}$; and $G_{q_i+2}$ contains $\sigma_{i,4}$.}
\label{fig:lp-group-formation}
\end{figure}

Consider the bipartite graph between the goods and the groups, where the edge $(j,G_q)$ has weight $f_{j,q}$. For each good $j \in M$, the sum of the weights on edges from $j$ to the groups created for agent $i$ is $x_{i,j}$; hence, its total incident weight is $\sum_i x_{i,j}=1$. For each group $G_q$, the sum of the incident edge weights is at most $1$, and is equal to $1$ unless $G_q$ is the last group of its agent. Thus, $f$ is a fractional matching that saturates every good and every non-last group. We decompose $f$ as a probability distribution over integral matchings, each of which matches all goods and all non-last groups. An integral matching is sampled from this distribution. Each matching corresponds to an allocation where each good is allocated to the agent whose group it is matched to. We use $z_{i,S}$ to denote the probability that agent $i$ gets the subset $S$. Note that the probability that agent $i$ gets a good $j$, that is $\sum_{S \ni j} z_{i,S}$ is exactly $x_{i,j}$. The crucial observation is that in each allocation in the support, each agent $i$ gets exactly one good from each of its first $p_i - 1$ groups, and possibly one good from its last group. Also, any good that agent $i$ gets from a group is (weakly) more valuable than any good it gets from the next group. This implies that any two subsets $S_1$ and $S_2$ in the support are envy free up to one item (EF1) for agent $i$. It is known from the work of~\citet*{BKV18Finding} that for any instance of \emph{unweighted} Nash welfare under identical additive valuations, any EF1 allocation provides an $e^{\nicefrac{-1}{e}}$ approximation to optimum Nash welfare. Using this result,~\citet*{FL25note} prove that:
\begin{equation}
\label{eq:feng_li_nash_log}    
\sum_{S \subseteq M} z_{i,S} \ln v_i(S) \geq \sum_{S \subseteq M} y_{i,S} \ln v_i(S) - \frac{1}{e}.
\end{equation}

\paragraph{The copied auxiliary instance.} For the purpose of analysis, suppose first that all $y_{i,S}$ and $z_{i,S}$ are rational numbers, and create an auxiliary instance as follows. Let $\Delta$ be a positive integer such that each $y_{i,S}$ and $z_{i,S}$ is an integer multiple of $\frac{1}{\Delta}$. It can be seen that $x_{i,j}$ is also an integer multiple of $\frac{1}{\Delta}$ for all goods $j$. Create $\Delta$ agents, and for each good $j$, create $x_{i,j} \cdot \Delta$ many copies. Consider allocations $A^y, A^z$ where exactly $y_{i,S} \cdot \Delta$ (resp. $z_{i,S} \cdot \Delta$) are allocated one copy each from the set $S$ of goods in $A^y$ (resp. $A^z$). Note that both $A^y, A^z$ are valid allocations in the instance, since $\sum_{S \ni j} y_{i,S} = \sum_{S \ni j} z_{i,S} = x_{i,j}$. Furthermore, $A^z$ is an EF1 allocation. Thus, applying the identical EF1 result of~\cite{BKV18Finding} and taking log, we get~\eqref{eq:feng_li_nash_log}.

\paragraph{The lexicographic improvement.}Our goal is to replace $\nicefrac{1}{e}$ by $\nicefrac{\ln 2}{2}$ in~\eqref{eq:feng_li_nash_log} for the case of lexicographic valuations. For this, fix an original agent $i$, and refine the auxiliary instance for this agent. We create the goods of the auxiliary instance one block at a time, each consisting of (up to) $\Delta$ goods. For each group $G_q$ of agent $i$, recall that $f_{j, q}$ is the probability that the good $j$ is matched to $G_q$ (and thus allocated to $i$). Modify $\Delta$ so that along with all the $y_{i,S}$ values, the probability of each matching in the support is also an integer multiple of $\frac{1}{\Delta}$. Then $f_{j,q}$ is also an integer multiple of $\frac{1}{\Delta}$, since it is the sum of the probabilities of the matchings that match $j$ to $G_q$. For each group $G_q$, we create a block of goods in the auxiliary instance, where we create exactly $f_{j,q}\cdot \Delta$ many copies of good $j$. Note that the total number of copies of good $j$ created is still $x_{i,j} \cdot \Delta$. Each block in the auxiliary instance, except possibly the last one, consists of $\Delta$ goods, and the last block consists of up to $\Delta$ goods. Now we again consider two allocations $A^y, A^z$. $A^y$ is created as before, that is, by giving exactly $y_{i,S} \cdot \Delta$ agents one arbitrary copy each of goods in the bundle $S$. Note that all copies of each good $j$ are allocated to distinct agents. For $A^z$, we refine the definition as follows. Recall that the distribution $z_i$ was obtained by decomposing $f$ (a matching between goods and good-groups) as a probability distribution over matchings, then sampling a matching and then allocating each agent all goods that are matched to one of its groups. For each matching in the support, we do the following. Let $p$ be the probability with which this matching is sampled. We repeat the following $p \cdot \Delta$ times: Pick a new agent. For each good $j$ that was matched to some group $G_q$, allocate one of the $f_{j, q} \cdot \Delta$ copies of $j$ in the block corresponding to $G_q$, to this agent. It can be verified that, for exactly $z_{i,S} \cdot \Delta$ many agents, the set of goods they get is exactly one copy each of the original goods in $S$. Again, we maintain that any two copies of the same good are allocated to distinct agents in $A^z$. Additionally, we have that each agent gets at most one good from each block in $A^z$. Thus, this refined instance has $\Delta$ cloned agents with identical lexicographic valuation $v_i$, at most $\Delta$ copies of each good type, a valid allocation $A^y$, and a block-balanced allocation $A^z$ (see~\Cref{sec:aux_ratio_bound} for definitions of valid and block balanced). It then follows from~\Cref{thm:ratio_root_2} that the ratio of the Nash welfare of $A^y$ to that of $A^z$ is upper bounded by $\sqrt{2}$, that is:
\begin{equation}
    \label{eq:root_2_ratio_lex}
    \frac{\WNash(A^y)}{\WNash(A^z)} = \frac{\left(\prod_{S \subseteq M} v_i(S)^{y_{i,S} \Delta}\right)^{\nicefrac{1}{\Delta}}}{\left(\prod_{S \subseteq M} v_i(S)^{z_{i,S} \Delta}\right)^{\nicefrac{1}{\Delta}}} \leq \sqrt{2}.
\end{equation}

Taking natural logarithm on both sides, we get:
\[\sum_{S \subseteq M} z_{i,S} \ln v_i(S) \geq \sum_{S \subseteq M} y_{i,S} \ln v_i(S) - \frac{\ln 2}{2}.\]

Then, as in~\cite{FL25note}, taking weighted (by $w_i$) sum of this inequality over all agents, and then applying Jensen's inequality shows that the expected Nash welfare of the output allocation is at least $\frac{1}{\sqrt{2}} - \eps$ times the optimal (the $-\eps$ term appears because the LP value for $y$ might be up to $\ln(1 + \eps)$ less than the optimal LP value). Furthermore, since it is possible to decompose the fractional matching $f$ into a polynomial number of integral matchings, it follows that at least one of the allocations corresponding to these matchings in the support must have a Nash welfare at least $\frac{1}{\sqrt{2}} - \eps$ times the optimal. Thus, we get a deterministic polynomial time algorithm with the same guarantee, thus proving~\Cref{thm:Root-2-Approximation}.

\subsection{Nash welfare ratio in the auxiliary instance}
\label{sec:aux_ratio_bound}

The preceding analysis reduces the lexicographic improvement to the auxiliary-instance comparison in~\eqref{eq:root_2_ratio_lex}. We now isolate the structural statement that establishes this comparison. Abstractly, the copied auxiliary instance consists of identical agents and multiple copies of lexicographically ordered good types. The allocation $A^y$ represents the LP benchmark, while the allocation $A^z$ induced by the rounding has the additional property that the goods in every consecutive rank block go to distinct agents. The central result of this subsection,~\Cref{thm:ratio_root_2}, shows that imposing this block structure loses a factor of at most $\sqrt{2}$ in Nash welfare. We next define this comparison setting precisely, state the theorem, and outline its proof.

Let $N=\{1,\ldots,n\}$ be a set of agents and let
$M=\{g_1,\ldots,g_m\}$ be a set of good types. For each type $g_j$, suppose there are $c_j\in\mathbb Z_{\ge 0}$ identical copies, with $c_j\le n$. All agents
have identical additive valuations, and one copy of $g_j$ has value $v_j\geq0$.
The values are indexed in nonincreasing order and satisfy the lexicographic condition $v_j \geq \sum_{k > j} v_k$ for each $j < m$.\footnote{The lexicographic condition defined in~\Cref{sec:Preliminaries} has a strict inequality. However, for our purposes in this section, the weak condition suffices.}

Call a complete allocation of the $m' = \sum c_j$ goods \emph{valid}, for each good, all copies are allocated to distinct agents. For any valid allocation $A$, let $A_i$ denote the good types whose copies (one copy each) are allocated to agent $i$. For a subset $S$, let $v(S):=\sum_{j\in S} v_j$.

List the $m'=\sum_j c_j$ goods in nonincreasing value order as $e_1,e_2,\ldots,e_{m'}$.
Thus all copies of $g_1$ come first, then all copies of $g_2$, and so on. Let us partition the goods into blocks, where the first $n$ goods form the first block, the next $n$ goods form the next block and so on. The final block contains at most $n$ copies. An allocation is called \emph{block-balanced} if it is valid and for every $1 \le q \le \lceil \frac{m'}{n} \rceil$, in the
$q$-th aligned block $\{e_{(q-1)n+1},\ldots,e_{\min(qn,m')}\}$, all goods are allocated to distinct agents.
\begin{restatable}[]{theorem}{StructuralRootTwo} 
\label{thm:ratio_root_2}
For every block-balanced allocation $A^\BB$ and every valid allocation
$A^\VV$,
\[
        \WNash(A^\BB)\ge \frac1{\sqrt2}\WNash(A^\VV).
\]
Equivalently,
\[
        \prod_{i=1}^n v(A^\VV_i)
        \le
        2^{n/2}\prod_{i=1}^n v(A^\BB_i).
\]
\end{restatable}
The details of the proof of~\Cref{thm:ratio_root_2} are deferred to~\Cref{sec:root_2_proof}. We sketch the proof outline below.

\paragraph{Outline of the proof.} We start with an arbitrary valid allocation $A^\VV$ and block-balanced allocation $A^\BB$ and make a sequence of allocation and valuation changes that always weakly increase the ratio $\rho = \WNash(A^\VV)/\WNash(A^\BB)$. We will refer to the allocation $A^\VV$ as the \emph{numerator} allocation and $A^\BB$ as the \emph{denominator} allocation. We first align the first block so that both allocations give agent $i$ the copied good $e_i$. We then show that the first block of goods can be assumed to have at most $3$ distinct types $g_1,g_2,g_3$, corresponding to $3$ mini-blocks $A,B,C$, as follows. First, we consider agents who prefer $A^\BB$ to $A^\VV$. For all such agents, we change the value of their top allocated good to a large value $L$. We call this the $L$-normalization step, and the set of goods with value $L$ forms the mini-block $A$. Then, for all the goods with value strictly less than $L$ and strictly greater than the value of the last good in the first block, we show that we can reduce their values to a common number while maintaining lexicographicity and increasing the ratio $\rho$. This set of goods forms the second mini-block $B$. Finally, all the remaining goods in the first block form the mini-block $C$ (see~\Cref{fig:mini-block-decomposition}). By a slight abuse of notation, we also use $A$, $B$, and $C$ to denote the corresponding sets of agents: agent $i$ belongs to $A$, $B$, or $C$ according to the mini-block containing its top good $e_i$.

\begin{figure}[ht]
\centering
\begin{tikzpicture}[x=1cm,y=1cm]
        \node[proofblock, fill=proofblue!13, minimum width=2.35cm] (A)
                at (1.175,0) {$A$\\ value $v_1=L$};
        \node[proofblock, fill=proofgreen!14, minimum width=2.65cm] (B)
                at (3.675,0) {$B$\\ value $v_2$};
        \node[proofblock, fill=prooforange!17, minimum width=2.1cm] (C)
                at (6.05,0) {$C$\\ value $v_3$};

        \node[proofblock, fill=prooforange!8, minimum width=1.6cm,
                minimum height=.58cm] (Gthree) at (8.05,0)
                {remaining\\ copies of $g_3$};

        \draw[proofarrow] (Gthree.east) -- ++(1.55,0)
                node[midway, above, font=\scriptsize] {other types}
                node[midway, below, font=\scriptsize] {$g_4,g_5,\ldots$};

        \draw[decorate, decoration={brace, mirror, amplitude=4pt}]
                (0,-.62) -- (7.1,-.62)
                node[midway, yshift=-14pt, font=\scriptsize]
                {first block after compression};
\end{tikzpicture}
\caption{The normalized first block: grouping into three mini-blocks $A$, $B$, and
$C$.}
\label{fig:mini-block-decomposition}
\end{figure}
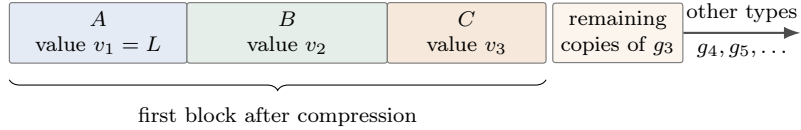

Next, we discard from the denominator allocation all goods beyond the second block and sort the assignment of the second-block goods so that agent $i$ receives $e_{n+i}$. Both operations can only decrease the denominator Nash welfare and hence increase $\rho$. We then normalize the allocation of the remaining (second block) copies of $g_3$ in the numerator: these copies are assigned first to agents in $B$ and then to agents in $A$, with the recipients forming a prefix within each mini-block. To compare the two Nash products, we seek to partition the agents into groups of size at most $3$. We call a group $G$ \emph{safe} if $\prod_{i\in G}v(A^\VV_i)\leq 2^{|G|/2}\prod_{i\in G}v(A^\BB_i)$. An agent who receives a copy of $g_3$ from the second block in both allocations forms a safe singleton group. We then match each agent in $A$ who receives $g_3$ only in the denominator with an agent in $B$ who receives it only in the numerator. Each resulting pair is safe, but we keep it \emph{open} so that a third agent may later be added to it. We call the unassigned agents and the representatives of these open pairs the \emph{open agents}; these agents form a consecutive interval. We process the lower-valued types $g_4,g_5,\ldots$ one at a time. For each type, its recipients among the open agents form a prefix in the denominator allocation, while a sequence of safe transformations allows us to assume that they form a suffix in the numerator allocation. We show that all agents in the union of this prefix and suffix can be placed into safe groups of size at most $3$, while the remaining open agents continue to form a consecutive interval. Repeating this argument eventually eliminates all open agents in $C$, after which all remaining open singletons and pairs can also be safely finalized. We thus obtain a partition of the agents into safe groups. Multiplying the defining inequalities over all these groups and taking $n$-th roots gives $\rho\leq\sqrt{2}$. Since every preceding transformation weakly increased $\rho$, the same bound holds for the original pair of allocations.~\Cref{fig:proof-roadmap} summarizes the proof roadmap.

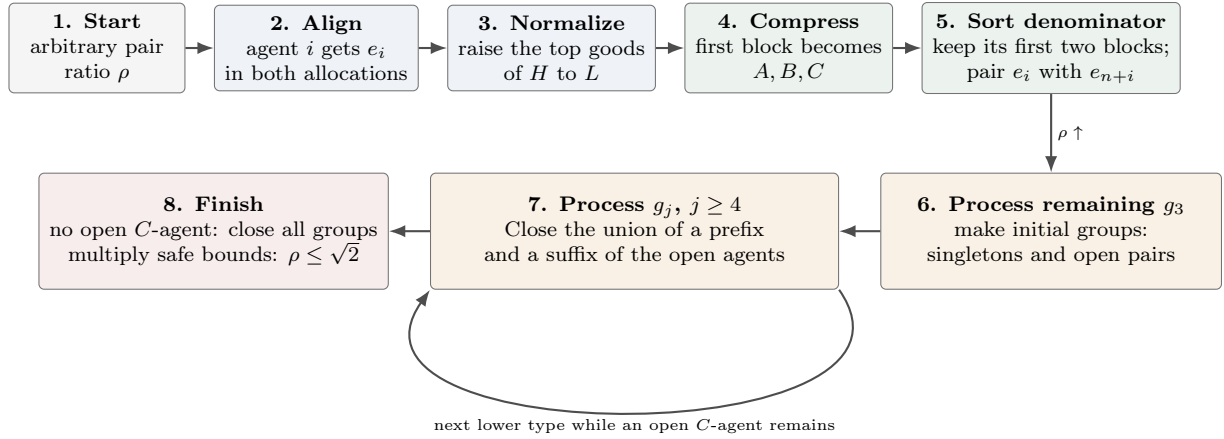
\begin{figure}[ht]
\centering
\resizebox{0.98\textwidth}{!}{%
\begin{tikzpicture}[
        node distance=.38cm,
        prooflabel/.style={font=\tiny, fill=white, inner sep=1pt}
]
        \node[proofstage, fill=black!4, minimum width=2.35cm] (start)
                {\textbf{1. Start}\\ arbitrary pair\\ ratio $\rho$};
        \node[proofstage, fill=proofblue!8, minimum width=2.4cm,
                right=of start] (align)
                {\textbf{2. Align}\\ agent $i$ gets $e_i$\\ in both allocations};
        \node[proofstage, fill=proofblue!8, minimum width=2.65cm,
                right=of align] (normalize)
                {\textbf{3. Normalize}\\ raise the top goods\\ of $H$ to $L$};
        \node[proofstage, fill=proofgreen!10, minimum width=2.55cm,
                right=of normalize] (compress)
                {\textbf{4. Compress}\\ first block becomes\\ $A,B,C$};
        \node[proofstage, fill=proofgreen!10, minimum width=3.15cm,
                right=of compress] (denominator)
                {\textbf{5. Sort denominator}\\ keep its first two blocks;\\ pair $e_i$ with $e_{n+i}$};

        \draw[proofarrow] (start) -- (align);
        \draw[proofarrow] (align) -- (normalize);
        \draw[proofarrow] (normalize) -- (compress);
        \draw[proofarrow] (compress) -- (denominator);

        \node[proofstage, fill=prooforange!12, minimum width=4.55cm,
                minimum height=1.55cm, below=1.05cm of denominator] (gthree)
                {\textbf{6. Process remaining $g_3$}\\
                make initial groups:\\
                singletons and open pairs};
        \node[proofstage, fill=prooforange!12, minimum width=5.45cm,
                minimum height=1.55cm, left=.55cm of gthree] (lower)
                {\textbf{7. Process $g_j$, $j\ge4$}\\
                Close the union of a prefix\\
                and a suffix of the open agents};
        \node[proofstage, fill=proofred!10, minimum width=3.75cm,
                minimum height=1.55cm, left=.55cm of lower] (finish)
                {\textbf{8. Finish}\\
                no open $C$-agent: close all groups\\
                multiply safe bounds: $\rho\le\sqrt2$};

        \draw[proofarrow] (denominator) --
                node[prooflabel,right=1pt] {$\rho\uparrow$} (gthree);
        \draw[proofarrow] (gthree) -- (lower);
        \draw[proofarrow] (lower) -- (finish);
        \draw[proofarrow] (lower.south east)
                to[out=-55,in=-125,looseness=1.25]
                node[prooflabel,below=1pt]
                {next lower type while an open $C$-agent remains}
                (lower.south west);
\end{tikzpicture}
}%
\caption{Roadmap of the proof of~\Cref{thm:ratio_root_2} The first row transforms the two allocations
into a normal form without decreasing $\rho$. The second row constructs the
safe groups: the remaining $g_3$ copies create the initial singletons and open pairs, and each lower type certifies a denominator prefix against a numerator suffix.}
\label{fig:proof-roadmap}
\end{figure}

\section{Algorithms for Structured Subclasses}
\label{sec:Results_Subclasses}

In this section, we present our exact and special-case approximation algorithms. We will find it convenient to express our results in terms of $n$ increasing concave functions $h_1, \dots, h_n$ and compute an allocation $A$ that maximizes $\WH(A) \coloneqq \sum_{i \in N} h_i(v_i(A_i))$. As discussed in \Cref{sec:Preliminaries}, for an appropriate choice of $h_i$'s, we recover the weighted Nash and $p$-mean welfare measures. 

For simplicity, we assume that every agent assigns a strictly positive value to every good. Our algorithms extend directly to instances in which some agent--good pairs have value $0$ since an optimal allocation never allocates a good to an agent who values it at $0$ (unless the good is valued at $0$ by all the agents). Hence, during the execution of our algorithm, we can simply avoid allocating a good to an agent who values it at $0$. Any good valued at $0$ by all agents can be assigned to an arbitrary agent at the end of the algorithm without affecting the Nash welfare of the resulting allocation. Similarly, even though our algorithms (with tie breaking rules) work for increasing concave functions, we shall assume that the functions $h_i$ are strictly increasing and strictly concave for the sake of exposition. Note that the functions discussed in \Cref{sec:Preliminaries} are indeed strictly increasing and strictly concave. 

Recall that, for a function $f$, its marginal function is defined as $\Delta_f(x, y) = f(x + y) - f(x)$. For brevity, we will write $\Delta_{i}(x,y)$ instead of $\Delta_{h_i}(x,y)$ for the marginal function of $h_i$. As discussed in the case of Nash welfare, when $h_i(0)$ is not defined, it will be treated as $-\infty$. Since the functions $h_i$ are strictly increasing and strictly concave, we have:
$$\Delta_i(x, y) > \Delta_i(x', y) \qquad \forall 0 \leq x < x', y > 0 \qquad \text{and}$$
$$\Delta_i(x, y) < \Delta_i(x, y') \qquad \forall y < y', x \ge 0$$

In order to prune the search space of our algorithms, we will use the following domination condition, that determines if it is sub-optimal to allocate a good to a particular agent given the current partial allocation of goods.

\paragraph{Domination with respect to good $g$.} Consider a partial allocation $A = (A_1, A_2, \ldots A_n)$, where a good $g$ is unallocated. Let $Y = M \setminus (A_1 \cup \ldots \cup A_n \cup \{g\})$ be the set of unallocated goods other than $g$. We say that \emph{agent $k$ dominates agent $i$ with respect to good $g$} under the partial allocation $A$, if allocating good $g$ to agent $i$ and all remaining goods to agent $k$ results in a lower objective value than allocating all the remaining goods to agent $k$. Formally,
\begin{equation*}
    h_i(v_i(A_i \cup \{g\})) + h_k(v_k(A_k \cup Y)) < h_i(v_i(A_i)) + h_k(v_k(A_k \cup Y \cup \{g\}))
\end{equation*}

Since the functions $h_i$ are concave, the affinity for good $g$ is higher for an agent when they have a lower current utility. Informally, the domination condition states that, even in the worst case, when agent $i$ has as low a utility as possible (does not get anything from $Y$) and agent $k$ has as high a utility as possible (gets everything from $Y$), it is still better to allocate good $g$ to agent $k$ than to agent $i$. If $h_i(0)$ is undefined we treat it as $-\infty$ and any such agent with $v_i(A_i) = 0$ (or equivalently $A_i = \emptyset$, as all goods are positively valued by all agents) is not dominated by any other agent. Similarly, if $Y = \emptyset$, any agent $k$ with $h_k(0)$ undefined and $A_k = \emptyset$ dominates every agent with a non-empty bundle. In terms of the corresponding marginal functions, the domination condition is equivalent to:
\begin{equation}
    \Delta_i(v_i(A_i), v_i(g)) < \Delta_k(v_k(A_k) + v_k(Y), v_k(g)).
    \label{eqn:domination_condition}
\end{equation}

The following lemma, which we prove in~\Cref{sec:algo_appendix},
will allow us to prune certain branches of the recursion tree of our algorithm based on the domination condition, without losing optimality.

\begin{restatable}[\textbf{Domination lemma}]{lemma}{DominationLemmaMain}
    Consider a partial allocation $A = (A_1, A_2, \ldots A_n)$, where a good $g$ is unallocated. Let $Y = M \setminus (A_1 \cup \ldots \cup A_n \cup \{g\})$ be the set of unallocated goods other than $g$. If agent $k$ dominates agent $i$ with respect to good $g$ under the partial allocation $A$, then any completion of $A$ that allocates good $g$ to agent $i$ cannot be optimal.
    \label{lem:domination}
\end{restatable}

\subsection{Warm-Up: Identical Valuations and Uniform \texorpdfstring{$h_i$}{hi} Functions}
\label{sec:identical}
In this section, we present a simple polynomial-time algorithm for maximizing $\WH$ when all agents share an identical valuation function $v$ and the functions $h_i$ are identical across agents, i.e., $h_i = f$ for all $i \in N$, where $f$ is a strictly increasing and strictly concave function. We make the following observation about the optimal solution, which we prove in~\Cref{sec:appendix_identical}.

\begin{restatable}{lemma}{IdenticalCardinality}
    In any optimal allocation, any agent who has two or more goods must have a lower utility than all other agents.
    \label{lem:identical}
\end{restatable}

The above lemma completely characterizes the optimal solution:
\begin{restatable}{corollary}{IdenticalStructure}
    In any optimal allocation, there exists an agent who receives exactly the worst $m - n + 1$ goods, and the best $n - 1$ goods are allocated one each to the remaining $n - 1$ agents.
    \label{cor:id_st}
\end{restatable}
We include the proof of~\Cref{cor:id_st} in~\Cref{sec:appendix_identical}. Thus, we get the following result:
\begin{proposition}
    When all agents have identical lexicographic valuation functions and the concave functions $h_i$ are identical across agents and computable in polynomial time, there is a polynomial-time algorithm to compute an optimal allocation that maximizes $\WH = \sum_{i \in N} h_i(v_i(A_i))$.
    \label{prop:identical-lexicographic}
\end{proposition}

\subsection{Ordered Lexicographic Valuations}\label{subsec:Results_Ordered_Nash}

In this subsection, we prove that for ordered lexicographic valuations, there exists a polynomial-time algorithm to compute a weighted Nash optimal allocation when the number of agents is constant. In fact, our algorithm works for any increasing concave functions of the utilities of the agents, and in particular applies to weighted $p$-mean welfare for any $p < 1$. Note that we allow these increasing concave functions to depend on the specific instance being solved as long as they are computable in polynomial time. Formally, our algorithm makes $\O((m + 1)^{\binom{n}{2}} \cdot \text{poly}(n, m))$ calls to the concave functions.

\begin{restatable}
{theorem}{OrderedConstantnWeighted}
For any constant positive integer $n$, there exists a polynomial-time algorithm that, given an instance with $n$ agents having ordered lexicographic valuations, along with increasing concave functions $h_1, h_2, \ldots, h_n$ of the utilities of the agents which can be computed in polynomial time, returns an allocation $A$ that maximizes $\WH(A) = \sum_{i \in N} h_i(v_i(A_i))$.
\label{thm:Ordered_Constant_n_Weighted}
\end{restatable}

Recall that, under ordered lexicographic valuations, there exists a permutation of goods\\$(g_1,g_2,\dots,g_m)$ such that for every $j \in [m-1]$ and every agent $i \in N$, we have $v_i(g_{j}) > v_i(g_{j+1})+v_i(g_{j+2})+\ldots+v_i(g_m)$. Our algorithm assigns goods one by one in the order $g_1, g_2, \ldots, g_m$. For any good $g_j$ there are $n$ options for the agent to whom it can be allocated, leading to a branching factor of $n$ in the recursion tree. However,~\Cref{lem:domination} will help us prune many of these branches: we never assign $g_j$ an agent that is dominated by some other agent with respect to $g_j$ under the current partial allocation. For each agent that is not dominated by any other agent, we allocate $g_j$ to that agent and then make a recursive call for the allocation of $g_{j+1},\ldots g_m$. The pseudo-code of our recursive algorithm(\textsc{OptimalCompletion}) that takes as input a partial allocation $A$ of the goods $g_1, \ldots g_{j-1}$ is presented in \Cref{alg:OptimalCompletion} and returns an optimal completion of $A$. The overall optimal allocation will thus be computed by calling $\textsc{OptimalCompletion}\big((\emptyset,\dots,\emptyset), 1\big)$.

\begin{algorithm}[t]
    \caption{\textsc{OptimalCompletion}}
    \label{alg:OptimalCompletion}
    \DontPrintSemicolon
    \textbf{Global:} Instance $\langle N, M, \V \rangle$ with ordered lexicographic valuations and concave functions $h_1, h_2, \ldots, h_n$.\;
    \KwIn{Partial allocation $A = (A_1,\dots,A_n)$ of goods $g_1,\dots,g_{j-1}$.}
    \KwOut{A completion $A^\star$ of $A$ that maximizes $\WH(C)$ among all completions $C$ of $A$.}
    \If{All the goods are allocated in $A$}{
        \Return $A$ \tcc*{$A$ is already a complete allocation.}
    }
    $A^\star \gets $ an arbitrary completion of $A$. \tcc*{best completion found so far}

    \For{$i \in N$}{
        \If{agent $i$ is not dominated by any agent w.r.t.\ good $g_j$ under $A$}{
            $B \gets A$\\
            $B_i \gets B_i \cup \{g_j\}$ \tcc*{Allocate good $g_j$ to agent $i$}
            $C \gets \textsc{OptimalCompletion}(B, j+1)$\;
            \If{$\WH(C) > \WH(A^\star)$}{
                $A^\star \gets C$\;
            }
        }
    }

    \Return $A^\star$
\end{algorithm}

The correctness of our algorithm follows from \Cref{lem:domination}, which says that we only prune those branches that lead to (strictly) sub-optimal allocations. We will now analyze its running time. To do so, we will consider the recursion tree of \textsc{OptimalCompletion}. Each node of the recursion tree corresponds to a call to \textsc{OptimalCompletion} with some partial allocation $A$. We will label each node with the corresponding partial allocation. The children of a node correspond to the recursive calls made in the for-loop of \Cref{alg:OptimalCompletion}. The depth of the recursion tree is $m$, where the root is at depth $0$ and leaves are at depth $m$. The branching factor at any node is at most $n$, but can be much smaller due to pruning via domination. Consider the recursion tree (Figure~\ref{fig:recursion_tree_example}) for the following instance in the case of Nash welfare ($h_i(x) = \log(x)$ for all agents $i$).

\begin{example}
Consider an instance with $m=4$ goods and $n=2$ agents with the following lexicographic valuations: 
\begin{table}[H]
    \centering
    \begin{tabular}{c|cccc}
    & $g_1$ & $g_2$ & $g_3$ & $g_4$ \\
    \hline
    $a_1$ & 8 & 4 & 2 & 1 \\
    $a_2$ & 19 & 13 & 3 & 2
    \end{tabular}
    \label{tab:example_valuations}
\end{table}

\label{exmpl:example_valuations}
\end{example}

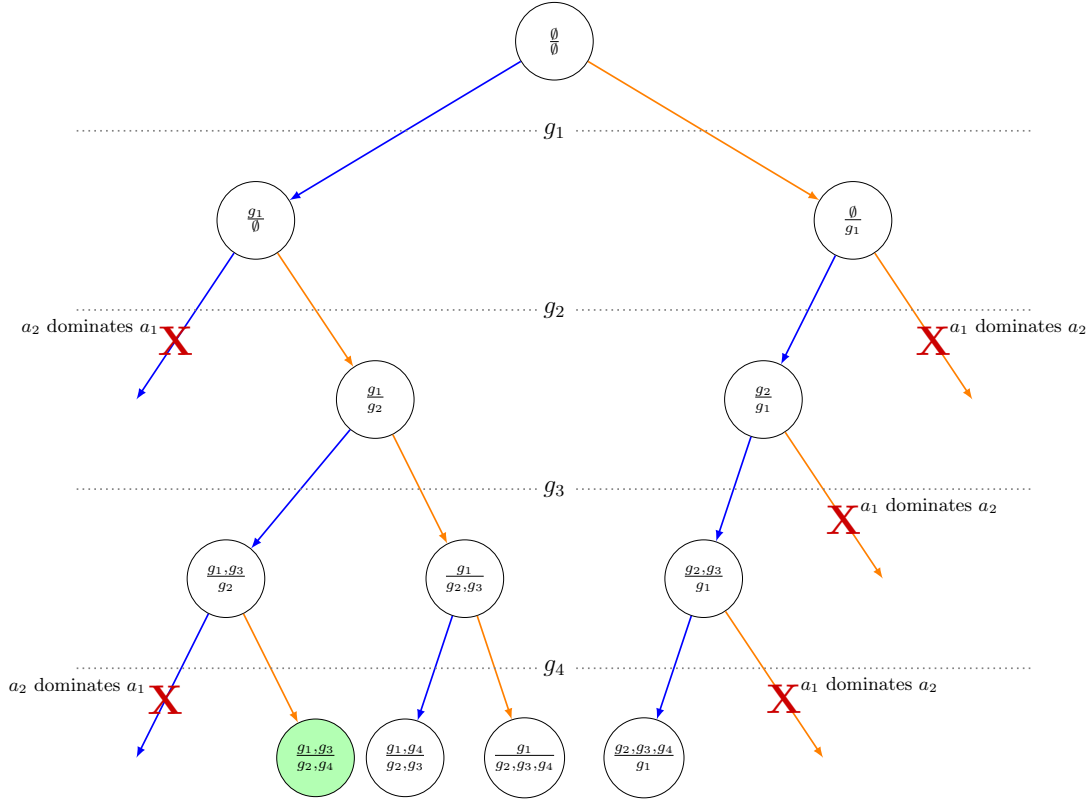
\begin{figure}[t]
\centering
\scalebox{0.79}{%
    \usetikzlibrary{arrows.meta, positioning, decorations.markings, calc}
\begin{tikzpicture}[
    my node/.style={circle, draw=black, inner sep=2pt, font=\small, minimum size=1.3cm},
    leaf node/.style={my node, fill=green!30},
    blue edge/.style={->, >=latex, draw=blue, thick},
    orange edge/.style={->, >=latex, draw=orange, thick},
    cross/.style={decoration={markings, mark=at position 0.6 with {\node[red!80!black, font=\bfseries, scale=1.8] {X};}}, postaction={decorate}},
    level line/.style={dotted, thick, gray},
    dominates label/.style={font=\footnotesize, align=center}
]

\node[my node] (root) at (0, 0) {$\frac{\emptyset}{\emptyset}$};

\node[my node] (l1_left) at (-5, -3) {$\frac{g_1}{\emptyset}$};
\node[my node] (l1_right) at (5, -3) {$\frac{\emptyset}{g_1}$};

\node[my node] (l2_center) at (-3, -6) {$\frac{g_1}{g_2}$};
\node[my node] (l2_right) at (3.5, -6) {$\frac{g_2}{g_1}$};

\node[my node] (l3_center_left) at (-5.5, -9) {$\frac{g_1, g_3}{g_2}$};
\node[my node] (l3_center_right) at (-1.5, -9) {$\frac{g_1}{g_2, g_3}$};
\node[my node] (l3_right_left) at (2.5, -9) {$\frac{g_2, g_3}{g_1}$};

\node[leaf node] (l4_leaf) at (-4, -12) {$\frac{g_1, g_3}{g_2, g_4}$};
\node[my node] (l4_center_left) at (-2.5, -12) {$\frac{g_1, g_4}{g_2, g_3}$};
\node[my node] (l4_center_right) at (-0.5, -12) {$\frac{g_1}{g_2, g_3, g_4}$};
\node[my node] (l4_right_left) at (1.5, -12) {$\frac{g_2, g_3, g_4}{g_1}$};

\draw[blue edge] (root) -- (l1_left);
\draw[orange edge] (root) -- (l1_right);

\coordinate (l2_left_cross_end) at (-7, -6);
\draw[blue edge, cross] (l1_left) -- (l2_left_cross_end) node[midway, left=0.3cm, dominates label] {$a_2$ dominates $a_1$};
\draw[orange edge] (l1_left) -- (l2_center);

\draw[blue edge] (l1_right) -- (l2_right);
\coordinate (l2_right_cross_end) at (7, -6);
\draw[orange edge, cross] (l1_right) -- (l2_right_cross_end) node[midway, right=0.3cm, dominates label] {$a_1$ dominates $a_2$};

\draw[blue edge] (l2_center) -- (l3_center_left);
\draw[orange edge] (l2_center) -- (l3_center_right);

\draw[blue edge] (l2_right) -- (l3_right_left);
\coordinate (l3_right_cross_end) at (5.5, -9);
\draw[orange edge, cross] (l2_right) -- (l3_right_cross_end) node[midway, right=0.3cm, dominates label] {$a_1$ dominates $a_2$};

\coordinate (l4_left_cross_end) at (-7, -12);
\draw[blue edge, cross] (l3_center_left) -- (l4_left_cross_end) node[midway, left=0.3cm, dominates label] {$a_2$ dominates $a_1$};
\draw[orange edge] (l3_center_left) -- (l4_leaf);

\draw[blue edge] (l3_center_right) -- (l4_center_left);
\draw[orange edge] (l3_center_right) -- (l4_center_right);

\draw[blue edge] (l3_right_left) -- (l4_right_left);
\coordinate (l4_right_cross_end) at (4.5, -12);
\draw[orange edge, cross] (l3_right_left) -- (l4_right_cross_end) node[midway, right=0.3cm, dominates label] {$a_1$ dominates $a_2$};

\foreach \y/\subscript in {-1.5/1, -4.5/2, -7.5/3, -10.5/4} {
    \draw[level line] (-8, \y) -- (8, \y);
    \node[fill=white, inner sep=3pt, font=\large] at (0, \y) {$g_{\subscript}$};
}

\end{tikzpicture}
}
    
    \caption{Recursion tree for \textsc{OptimalCompletion} on the instance in \Cref{exmpl:example_valuations}, for the case of unweighted Nash welfare. Each node displays the partial allocation $\frac{A_1}{A_2}$, where $A_1$ and $A_2$ are the bundles assigned to agents $a_1$ and $a_2$, respectively. Blue edges represent allocating the current good to agent $a_1$, and orange edges represent allocating it to agent $a_2$. The horizontal dotted lines separate levels of the tree, with each level corresponding to the allocation of a particular good (labeled $g_1, g_2, g_3, g_4$ from top to bottom). Edges marked with an ``X'' indicate pruned branches where the domination condition (\Cref{lem:domination}) applies, along with the dominating agent. The green leaf node represents the Nash optimal allocation $A_1 = \{g_1, g_3\}$ and $A_2 = \{g_2, g_4\}$.}
    \label{fig:recursion_tree_example}
\end{figure}

Note that, for this example, only $4$ out of the $n^m = 16$ possible complete allocations are eventually explored by the algorithm, thanks to pruning via domination. We will generalize this observation to arbitrary instances and bound the number of leaves in the recursion tree. Note that, since the total number of nodes in the recursion tree is at most $m$ times the number of leaves, a polynomial bound (for constant $n$) on the number of leaves will imply a polynomial running time for the algorithm.

\begin{restatable}[Running time of \Cref{alg:OptimalCompletion}]{lemma}{OrderedLexicographicRunningTime}
    The number of leaves in the recursion tree of \Cref{alg:OptimalCompletion} is at most $\O((m+1)^{\binom{n}{2}})$.
    \label{lem:Ordered_Lexicographic_Running_Time}
\end{restatable}

The following key lemma which we prove in~\Cref{sec:ordered_appendix} will be used in the proof of \Cref{lem:Ordered_Lexicographic_Running_Time}. Importantly, note that this lemma holds for general lexicographic valuations, not just ordered ones.

\begin{restatable}
{lemma}{DominationLemma}
    Consider partial allocations $B, A$ and $A'$, such that $B_r \subseteq A_r$ and $B_r \subseteq A'_r$ for all agents $r \in N$. Suppose there exists a good $g$ that is the most valuable good for both the agents $i$ and $k$ among the unallocated goods in $B$, such that $g \in A_i$ and $g \in A'_k$. Then, for any good $e$ that is unallocated in both $A$ and $A'$, either agent $k$ dominates agent $i$ with respect to good $e$ under the partial allocation $A$, or agent $i$ dominates agent $k$ with respect to good $e$ under the partial allocation $A'$.
\label{lem:domination_based_pruning_lemma}
\end{restatable}

For our algorithm for ordered valuations, \Cref{lem:domination_based_pruning_lemma} implies that once two sibling nodes disagree on which of two agents receives a good (which, by design, is currently the most valuable remaining good for both of them), then for every later good at least one of the two symmetric possibilities (giving that later good to the first agent versus to the second) will be ruled out by pruning, \emph{at the level where that later good is actually allocated}. Concretely, consider a node (partial allocation) $B$ at depth $j-1$ in the recursion tree of \textsc{OptimalCompletion}, and fix two distinct agents $i,k\in N$. Let $C_i$ and $C_k$ be two children of $B$ in which the current good $g_j$ is allocated to $i$ and to $k$, respectively.

Now fix any later good $g_r$ with $r>j$. Consider any two descendants $\widehat{C}_i$ of $C_i$ and $\widehat{C}_k$ of $C_k$ that lie at depth $r-1$ and correspond to some partial allocations obtained after allocating exactly the goods $g_1,\dots,g_{r-1}$. (Thus, $\widehat{C}_i$ and $\widehat{C}_k$ any are two nodes in the subtrees of $C_i$ and $C_k$ respectively, at which the algorithm is about to decide who receives $g_r$.)

Because the instance is ordered lexicographic and goods are processed in the common decreasing order, $g_j$ is the most valuable good that is unallocated in $B$ for both agents $i$ and $k$. Hence we may apply \Cref{lem:domination_based_pruning_lemma} with
\[
A \coloneqq \widehat{C}_i,\qquad
A' \coloneqq \widehat{C}_k,\qquad
g \coloneqq g_j,\qquad
e \coloneqq g_r.
\]
The lemma guarantees that either agent $k$ dominates agent $i$ with respect to good $g_r$ under the partial allocation $\widehat{C}_i$, or agent $i$ dominates agent $k$ with respect to good $g_r$ under the partial allocation $\widehat{C}_k$. In the first case, by \Cref{lem:domination}, when expanding $\widehat{C}_i$ the algorithm will prune the branch that allocates $g_r$ to agent $i$; in the second case, when expanding $\widehat{C}_k$ it will prune the branch that allocates $g_r$ to agent $k$. Therefore, at the point where $g_r$ is allocated in these two sibling subtrees, at least one of the two symmetric choices (giving $g_r$ to $i$ versus giving $g_r$ to $k$) cannot survive. Equivalently, either $g_r$ is never allocated to $i$ in the subtree of $C_i$, or it is never allocated to $k$ in the subtree of $C_k$.

\begin{corollary}[Mutual exclusion across sibling branches]
\label{cor:mutual_exclusion_siblings}
Let $B$ be a node at depth $j-1$ in the recursion tree of \textsc{OptimalCompletion}, and let $i\neq k$ be two agents such that both children $C_i$ and $C_k$ of $B$ (allocating $g_j$ to $i$ and to $k$, respectively) are explored (i.e., neither is pruned at depth $j$). Then for every $r>j$, at least one of the following holds:
\begin{enumerate}
    \item no descendant of $C_i$ ever allocates $g_r$ to agent $i$; or
    \item no descendant of $C_k$ ever allocates $g_r$ to agent $k$.
\end{enumerate}
\end{corollary}

    Now, let $A$ be any node in the recursion tree at depth $j$. Consider a binary matrix $M(A)$, with $m - j$ rows corresponding to the unallocated goods $g_{j+1}, \ldots, g_m$ and $n$ columns corresponding to the agents in $N$. The entry for the row corresponding to good $g_r$ and the column corresponding to agent $i$ is $1$ if there exists a descendant of $A$ in which good $g_r$ is allocated to agent $i$, and $0$ otherwise. Let $M(A)^i$ denote the column vector in $M(A)$ corresponding to agent $i$. Let $\text{leaves}(A)$ denote the number of leaves in the subtree of node $A$. The following lemma, which we prove using induction in~\Cref{sec:ordered_appendix} gives an upper bound on $\text{leaves}(A)$ in terms of the allocation matrix $M(A)$.

\begin{restatable}{lemma}{LeafBound}
    The number of leaves in the subtree of $A$ is at most:
    $$\text{leaves}(A) \leq \prod_{1 \leq i < k \leq n} \big(1 + M(A)^i \cdot M(A)^k\big), $$
where ``$\cdot$'' represents the dot product between two vectors.
    \label{lem:leaf_bound}
\end{restatable}

Note that this implies \Cref{lem:Ordered_Lexicographic_Running_Time} since $M(A)^i \cdot M(A)^k \leq m$ for all $i,k \in N$. Hence, substituting $A$ to be the root node (empty allocation), the total number of leaves in the recursion tree is at most $(1 + m)^{\binom{n}{2}}$.

\subsubsection{An EPTAS for General \texorpdfstring{$n$}{n}} 
When the number of agents $n$ is not necessarily a constant, we show the existence of an EPTAS for (unweighted) Nash welfare under ordered lexicographic valuations.
In this section, we present an EPTAS for maximizing (unweighted) Nash welfare.

\begin{restatable}{theorem}{EPTASNashOrdered}
    For any $\varepsilon > 0$, there exists an algorithm that, given any instance with ordered lexicographic valuations, returns an allocation whose Nash welfare is at least $(1-\varepsilon)$ times the optimal Nash welfare, and runs in time $(\nicefrac{1}{\varepsilon})^{\O(\log \nicefrac{1}{\varepsilon})} \cdot (n^3 + nm)$.
    \label{thm:EPTAS}
\end{restatable}

We prove this theorem in~\Cref{sec:EPTASthmappendix}.

\subsection{Doubling Valuations}
\label{subsec:Doubling_Valuations}

In this section, we design a polynomial-time algorithm for finding a Nash optimal allocation when the instance has doubling valuations and a constant number of agents. We will in fact prove a more general result for maximizing any separable concave welfare function $\WH(A) = \sum_{i \in N} h_i(v_i(A_i))$ under doubling valuations, subject to a technical condition on the functions $h_i$'s. Recall that, for an agent $i$, $\Delta_i$ denotes the marginal function of $h_i$, that is, $\Delta_i(x, y) = h_i(x+y) - h_i(x)$.

\begin{restatable}[Algorithm for doubling valuations for constant $n$]{theorem}{GeneralConstantn} For any constant positive integer $n$, there exists a polynomial-time algorithm that takes as input an instance with $n$ agents with doubling valuations, along with strictly increasing, strictly concave, and polynomial-time computable functions $h_1, h_2, \ldots, h_n$ that satisfy
    \[
        \Delta_i(2x,2y)\ge \Delta_i(x,y)\qquad \forall x>0,\ y\ge 0,
    \]
and returns an allocation $A$ that maximizes $\WH(A)=\sum_{i\in N} h_i\!\left(v_i(A_i)\right).$
    \label{thm:Doubling_Constant_n}
\end{restatable}

Note that this requirement is satisfied for the case of weighted Nash welfare as well as weighted $p$-mean welfare for $p \in (0, 1)$:

\begin{itemize}
    \item For weighted Nash welfare, we have $h_i(x) = w_i \log x$, and hence $\Delta_i(x, y) = w_i \log(1 + \frac{y}{x})$. Thus, $\Delta_i(2x, 2y) = \Delta_i(x, y)$.
    \item For weighted $p$-mean welfare with $p \in (0, 1)$, we have $h_i(x) = w_i x^p$. Hence, $\Delta_i(2x, 2y) = w_i ((2x + 2y)^p - (2x)^p) = 2^p\,w_i((x+y)^p-x^p)> w_i((x+y)^p-x^p) = \Delta_i(x, y)$.
\end{itemize}
We iterate over all $2^n-1$ choices for the subsets of agents who receive at least one good. Thus, in the rest of the section, we assume that each agent receives at least one good. Our algorithm (\Cref{alg:Doubling_Constant_n}) begins with choosing (by enumerating over all $\O(m^n)$ choices) the favourite good assigned to each agent in the optimal allocation. For each such choice, it runs~\Cref{alg:OptimalRestrictedCompletion}, that builds on~\Cref{alg:OptimalCompletion} to return an optimal allocation subject to the choice of the maximum valued goods. We refer to such an allocation as a \emph{restricted completion} of the initial partial allocation, defined formally below.
        
\begin{definition}[Restricted completion]
    Given a partial allocation $B = (B_1, B_2, \ldots B_n)$, where $B_i \neq \phi $ for all $i \in N$, we say that an allocation $A = (A_1, A_2, \ldots A_n)$ is a restricted completion of $B$ if for all $i \in N$, $B_i \subseteq A_i$ and $\max_{g \in A_i} v_i(g) = \max_{g \in B_i} v_i(g)$.
\end{definition}

 We will design an algorithm \textsc{OptimalRestrictedCompletion}, which takes as input a partial allocation $B = (B_1, B_2, \ldots B_n)$, and outputs a restricted completion of $B$ that maximizes the required objective. Since a restricted completion only allocates goods that are less valuable than the highest valued good to each agent, the most valuable good in each agent's bundle always remains the same (as the choice made in the beginning). Given the current partial allocation $B$, let $F$ be the set of unallocated goods in $B$, that is $F = M \setminus \bigcup_{i \in N} B_i$. For each agent $i \in N$, let $F_i$ denote the set of unallocated goods that can be assigned to it in a restricted completion, i.e., $F_i = \{g \in F: v_i(g) < \max_{g' \in B_i} v_i(g')\}$. We will call a good \emph{allocable} to agent $i$ if $g \in F_i$. For each $i$ such that $F_i$ is non-empty, let $f_i$ denote the favorite good of agent $i$ in $F_i$. We will assume that each good is allocable to some agent, that is, for each good $g \in F$, there exists some agent $i$ such that $g \in F_i$. If this condition is not satisfied, then we shall not consider this choice of maximum valued goods further (see lines~\ref{l:check_feasible_begin} to~\ref{l:check_feasible_end}) of~\Cref{alg:Doubling_Constant_n}) since there is no feasible restricted completion. 
 Finally, for each good $g \in F$, we define $Z_g$ to be the set of agents $i$ for whom $f_i$ equals $g$.

 Since we are now working in the space of restricted completions, we would only need a weaker definition of domination, which we shall refer to as \emph{weak domination}.

\begin{definition}[Weak Domination]
    Given a partial allocation $B$, consider two agents $i, k \in N$ and a good $g \in F_i \cap F_k$. We say that agent $k$ weakly dominates agent $i$ with respect to good $g$ under the partial allocation $B$ if the following condition holds:
    \begin{equation*} \label{eqn:weak_domination}
    h_i(v_i(B_i \cup \{g\})) + h_k(v_k(B_k \cup F_k \setminus \{g\})) < h_i(v_i(B_i)) + h_k(v_k(B_k \cup F_k))
    \end{equation*}
\end{definition}

The condition says that even if agent $k$ gets all the goods that can be allocated to it except $g$, and agent $i$ gets only $g$ in addition to its current bundle, it is still better to allocate $g$ to agent $k$ instead of agent $i$. Upon re-arrangement, it can be seen that the above inequality is equivalent to:
\begin{equation} \label{eqn:weak_domination_2}
    \Delta_k(v_k(B_k \cup F_k \setminus \{g\}), v_k(g)) > \Delta_i(v_i(B_i), v_i(g))
\end{equation}

Analogous to~\Cref{lem:domination}, we can prove that it is never optimal to allocate a good to an agent that is weakly dominated by some other agent with respect to that good.

\begin{lemma}
    Let $k$ be an agent that weakly dominates an agent $i$ with respect to a good $g \in F_i \cap F_k$ for a partial allocation $B$. Then, in every optimal restricted completion of $B$, good $g$ is not allocated to agent $i$.
    \label{lem:weak_domination}
\end{lemma}

Now, note that the main observation (\Cref{lem:domination_based_pruning_lemma}) that helped us bound the number of leaves for the ordered lexicographic case does not assume that the valuations are ordered. The only requirement is that, both the agents involved share the same favorite good. Hence, our plan is to identify a good $g$ that \emph{must} be allocated to some agent in $Z_g$ in \emph{each} optimal restricted completion of $A$. Then, we try all possible assignments of $g$ to agents in $Z_g$ (subject to pruning based on weak domination), and recursively solve the problem for the remaining goods. Since agents in $Z_g$ have the same favorite allocable good $g$, we can use a refinement of~\Cref{lem:domination_based_pruning_lemma} (stated as~\Cref{lem:weak_pruning_lemma} later) to prune the search space. The following claim, that we prove in~\Cref{sec:missing_proofs_doubling} shows that such a good $g$ can be identified in polynomial time.

\begin{restatable}{claim}{GoodIdentification}
    Let $s$ be an agent such that $F_s \neq \emptyset$ that maximizes $ \Delta_s(2v_s(B_s), v_s(f_s))$. Then, in every optimal restricted completion of $B$, the good $f_s$ is allocated to some agent in $Z_{f_s}$.
    \label{clm:good_identification}
\end{restatable}

    \begin{algorithm}[t]
    \caption{\textsc{OptimalRestrictedCompletion}}
    \label{alg:OptimalRestrictedCompletion}
    \DontPrintSemicolon
    \textbf{Global:} Instance $\langle N, M, \V \rangle$ with doubling valuations and concave functions $h_1, h_2, \ldots, h_n$.\;
    \KwIn{Partial allocation $B = (B_1,\dots,B_n)$ with $B_i \neq \phi$ for all $i \in N$.}
    \KwOut{An optimal restricted completion $B^\star$ of $B$ that maximizes $\WH(C)$ among all restricted completions $C$ of $B$.}

    $F \gets M \setminus \bigcup_{r \in N} B_r$ \tcc*{unallocated goods under $B$}
    \If{$F = \emptyset$}{
        \Return $B$ \tcc*{All goods allocated}
    }
    \For{$r \in N$}{
        $F_r \gets \{g \in F : v_r(g) < \max_{g' \in B_r} v_r(g')\}$\;
        \If{$F_r$ is non-empty}{
            $f_r \gets \arg\max_{g \in F_r} v_r(g)$ \tcc*{Favorite feasible good (break ties arbitrarily)}
        }
    }
    \For{$g \in F$}{
        $Z_g \gets \{r \in N : F_r \neq \emptyset \text{ and } f_r = g\}$\;
    }
    
    $s \gets \underset{\substack{r \in N: F_r \neq \emptyset}}{\arg\max}\ \Delta_r(2v_r(B_r), v_r(f_r))$ \tcc*{Identify agent as in~\Cref{clm:good_identification}}
    $g \gets f_s$ \tcc*{Good to be allocated}
    $B^\star \gets $ an arbitrary restricted completion of $B$ \tcc*{Best completion found so far}

    \For{$k \in Z_g$}{
        \If{agent $k$ is not weakly dominated by any agent w.r.t.\ good $g$ under $B$}{
            $B' \gets B$\\
            $B'_k \gets B'_k \cup \{g\}$ \tcc*{Allocate good $g$ to agent $k$}
            $C \gets \textsc{OptimalRestrictedCompletion}(B')$\;
            \If{$\WH(C) > \WH(B^\star)$}{
                $B^\star \gets C$\;
            }
        }
    }

    \Return $B^\star$
\end{algorithm}

We present our algorithm \textsc{OptimalRestrictedCompletion} in \Cref{alg:OptimalRestrictedCompletion}. Let us first discuss its correctness. Our algorithm identifies a good $g = f_s$ that must be allocated to some agent in $Z_g$ in every optimal restricted completion of $B$. We then try all possible allocations of $g$ to agents in $Z_g$, pruning only those agents that are weakly dominated by some other agent (and hence could not have received the good $g$ in any optimal completion, because of~\Cref{lem:weak_domination}). Hence, we never prune any optimal restricted completion. For bounding the runtime, we will prove the exact same bound on the number of leaves in the recursion tree of~\Cref{alg:OptimalRestrictedCompletion} as we did for~\Cref{alg:OptimalCompletion}.

\begin{lemma}
    The number of leaves in the recursion tree of \textsc{OptimalRestrictedCompletion} is at most $(m + 1)^{\binom{n}{2}}$.
    \label{lem:leaf_bound_doubling}
\end{lemma}

The proof of~\Cref{lem:leaf_bound_doubling} is completely analogous to that of~\Cref{lem:Ordered_Lexicographic_Running_Time}. First, we need the following refinement of~\Cref{lem:domination_based_pruning_lemma}:

\begin{restatable}{lemma}{WeakDominationLemma}
    Consider partial allocations $B, A$ and $A'$, where $B_r \subseteq A_r$ and $B_r \subseteq A'_r$ for all agents $r \in N$, such that $f_i = f_k$ (defined with respect to the partial allocation $B$) for some agents $i, k \in N$. Suppose that $f_i \in A_i$ and $f_i \in A'_k$. Then, for any good $e \in F_i \cap F_k$ that is unallocated in both $A$ and $A'$, either agent $k$ weakly dominates agent $i$ with respect to $e$ under the partial allocation $A$ or agent $i$ weakly dominates agent $k$ with respect to $e$ under the partial allocation $A'$.
    \label{lem:weak_pruning_lemma}
\end{restatable}

The proof of~\Cref{lem:weak_pruning_lemma} is similar to that of~\Cref{lem:domination_based_pruning_lemma}. The only difference is that, $T$ and $T'$ are replaced by $T \cap F_k\setminus \{g\}$ and $T' \cap F_i\setminus \{g\}$, respectively, in the proof. We include the proof in~\Cref{sec:missing_proofs_doubling} for completeness.

Now, let us fix a choice of maximum valued goods for each agent in an optimal allocation. This gives a partial allocation $X$ where each agent $i$ has a single good. We now argue that the recursion tree of \textsc{OptimalRestrictedCompletion}$(X)$ has the same type of ``mutual exclusion'' structure as in the ordered lexicographic case. Fix a node $B$, and let $g$ be the good selected to be allocated at $B$ with candidate set $Z_g$. Consider two distinct children of $B$ obtained by allocating $g$ to agents $i,k \in Z_g$, yielding partial allocations $C_i$ and $C_k$, respectively. Let $e$ be any good in $(F_i \cap F_k) \setminus \{g\}$ that is unallocated in $B$. Then, we claim that, either $e$ is never allocated to $i$ in the subtree of $C_i$, or $e$ is never allocated to $k$ in the subtree of $C_k$. To see this, suppose not. Then, there exists a partial allocation $A$ in the subtree of $C_i$, and a partial allocation $A'$ in the subtree of $C_k$, such that $e$ is unallocated in both $A$ and $A'$, and both $A$ and $A'$ have a child corresponding to the allocation of $e$ to $i$ and $k$, respectively. But then, by~\Cref{lem:weak_pruning_lemma}, either $k$ weakly dominates $i$ with respect to $e$ under the partial allocation $A$, or $i$ weakly dominates $k$ with respect to $e$ under the partial allocation $A'$, a contradiction.

This yields the following corollary, analogous to \Cref{cor:mutual_exclusion_siblings}, which we will use to bound the number of leaves in the recursion tree. Note that the above discussion only implies~\Cref{cor:weak_mutual_exclusion_siblings} for $e \in F_i \cap F_k$. However, if $e \notin F_i \cap F_k$, then at least one of the two conditions of the corollary is trivially true (if $e \notin F_i$, then $e$ is never allocated to $i$).

\begin{corollary}
    \label{cor:weak_mutual_exclusion_siblings}
    Fix any node $B$ in the recursion tree of \textsc{OptimalRestrictedCompletion} and let $g$ be the good selected at $B$, with candidate set $Z_g$. For any two distinct agents $i,k \in Z_g$, let $C_i$ and $C_k$ denote the corresponding children of $B$ obtained by allocating $g$ to $i$ and $k$, respectively. Then, for every good $e$ other than $g$, that is unallocated in $B$, at least one of the following holds:
    \begin{enumerate}
        \item no partial allocation in the subtree rooted at $C_i$ allocates $e$ to agent $i$, or
        \item no partial allocation in the subtree rooted at $C_k$ allocates $e$ to agent $k$.
    \end{enumerate}
\end{corollary}

 Similar to the proof of~\Cref{lem:leaf_bound}, we now define a matrix $M(B)$ as follows. The rows correspond to the goods not allocated under $B$. The columns correspond to the agents. An entry $M(B)[g,i]$ is $1$ if good $g$ is ever allocated to agent $i$ in the subtree of node $B$ in the recursion tree. Trivially, note that $M(B)[g,i] = 0$ for all goods $g \notin F_i$. Define $\text{leaves}(B)$ to be the number of leaves in the subtree rooted at $B$. Then, using~\Cref{cor:weak_mutual_exclusion_siblings} completely analogously to the use of~\Cref{cor:mutual_exclusion_siblings} in the proof of ~\Cref{lem:leaf_bound}, we get the following result, that proves~\Cref{lem:leaf_bound_doubling}. The only difference is that, now we remove the row corresponding to the good $g$, instead of the good $g_{j+1}$ from the matrix $M(A)$ to get the matrix $M'(A)$. For the sake of completeness, we include the analogous proof in~\Cref{sec:missing_proofs_doubling}.
 \begin{restatable}{lemma}{LeafBoundWeak}
    For any node $A$ in the recursion tree of \textsc{OptimalRestrictedCompletion},
    
    $$\text{leaves(A)} \leq \prod_{1 \le i < j \leq n} \left( 1 + M(A)^i \cdot M(A)^j \right)$$
    
    where $M(A)^i$ is the $i^{th}$ column of matrix $M(A)$ and $\cdot$ denotes the dot product.
    \label{lem:leaf_bound_weak}
\end{restatable}

Hence, each call to \textsc{OptimalRestrictedCompletion}$(X)$ takes $\O((m + 1)^{\binom{n}{2}} \cdot \text{poly}(n,m))$ time. Since we run \textsc{OptimalRestrictedCompletion} for each of the $\O(m^n)$ choices of maximum valued goods, we get~\Cref{thm:Doubling_Constant_n}. The pseudocode for the overall algorithm is presented in~\Cref{alg:Doubling_Constant_n}.

\begin{algorithm}[t]
    \caption{\textsc{OptimalCompleteAllocation}}
    \DontPrintSemicolon
    \textbf{Global:} Instance $\langle N, M, \V \rangle$ with doubling valuations and increasing concave functions $h_1, h_2, \ldots, h_n$.\;
    \KwIn{Instance $\langle N,M,\V\rangle$ with doubling valuations and constant $n=|N|$.}
    \KwOut{An allocation $A^\star$ that maximizes $\WH(A)$ among all allocations $A$.}

    $A^\star \gets$ an arbitrary allocation\;

    \ForEach{tuple $(g_1,\dots,g_n)\in M^n$}{
        \If{$|\{g_1,\dots,g_n\}|<n$}{
            \textbf{continue}\tcc*{$g_i$ must be all distinct.}
        }

        $X \gets (X_1,\dots,X_n)$ with $X_i \gets \{g_i\}$ for all $i\in N$\;
        $F \gets M \setminus \bigcup_{i\in N} X_i$\;
        $\textsc{Feasible} \gets \texttt{true}$\;
        \For{$g\in F$}{\label{l:check_feasible_begin}
            \If{$v_i(g) > v_i(g_i)$ for every $i\in N$}{
                $\textsc{Feasible} \gets \texttt{false}$\;
                \textbf{break}\tcc*{$g$ cannot be allocated in a restricted completion.}
            }\label{l:check_feasible_end}
        }
        \If{$\textsc{Feasible}=\texttt{false}$}{
            \textbf{continue}\tcc*{Move to the next tuple.}
        }

        $A \gets \textsc{OptimalRestrictedCompletion}(X)$\;

        \If{$\WH(A) > \WH(A^\star)$}{
            $A^\star \gets A$\;
        }
    }
    \Return $A^\star$\;
\label{alg:Doubling_Constant_n}
\end{algorithm}

\section{Hardness Results}
\label{sec:Hardness_Results}

In this section, we establish hardness results for maximizing Nash welfare under lexicographic valuations. In all the hard instances we create, each good is valued \emph{positively} by every agent. The following lemma, which follows from~\Cref{lem:domination} will be useful in our reductions. We include its proof in~\Cref{sec:hardness_appendix}.

\begin{restatable}{lemma}{TransferLemma}
    Suppose in a complete allocation $A$, agent $i$ holds good $g$, and each agent has a positive utility. Suppose there exists an agent $k$ such that
    \[
        \frac{v_i(A_i \setminus \{g\})}{v_k(A_k)} > \frac{v_i(g)}{v_k(g)}.
    \]
    Then, $A$ cannot be Nash optimal.
    \label{lem:transfer_lemma}
\end{restatable}

We start by proving NP-hardness of Nash welfare maximization for ordered lexicographic valuations. The decision version of the Nash welfare problem under ordered lexicographic valuations is formally stated below.

\smallskip

\begin{center}
\fbox{
\begin{minipage}{0.9\linewidth}
\textbf{\textsc{Ordered-Lex-NSW.}} Given an instance $\langle N, M, \{v_i\}_{i \in N} \rangle$ with ordered lexicographic valuations, and a threshold $\theta$, decide whether there exists an
allocation $A$ such that $W^\texttt{Nash}(A) \geq \theta$.
\end{minipage}%
}
\end{center}

\begin{theorem}
    \textsc{Ordered-Lex-NSW} is NP-hard.
    \label{thm:NP-hard_Ordered_Nash}
\end{theorem}

Our proof relies on a reduction from the \textsc{Perfect 3D Matching} problem to \textsc{Ordered-Lex-NSW}. An instance of the perfect 3D matching problem consists of three disjoint sets
$X, Y, Z$ with $\abs{X}=\abs{Y}=\abs{Z}=n$, and a set of triples
$T \subseteq X \times Y \times Z$.
A \emph{3D matching} is a subset $T' \subseteq T$ such that no two distinct triples in $T'$
share an element, i.e., for any $(x,y,z), (x',y',z') \in T'$ we have
$x \neq x'$, $y \neq y'$, and $z \neq z'$. The problem asks whether there exists a perfect 3D matching in the input instance, that is, whether there exists a collection of $n$ triples in $T$
that \emph{covers} every element of $X \cup Y \cup Z$ exactly once.
This problem is known to be \NP-complete~\citep*{DBLP:conf/coco/Karp72}.

Given an instance of \textsc{Perfect 3D Matching}, we construct an instance of the Max Nash welfare problem as follows. We create an agent for each tuple in $T$. Each element of $X \cup Y \cup Z$ corresponds to a good. Additionally, we create a set of $|T| - n$ dummy goods $D = \{d_1, d_2, \ldots, d_{|T|-n}\}$. Thus, the total number of goods is $3n + (|T| - n) = |T| + 2n$. Let $X = \{x_1, x_2, \ldots, x_n\}$, $Y = \{y_1, y_2, \ldots, y_n\}$ and $Z = \{z_1, z_2, \ldots, z_n\}$. The \emph{common} preference order (say $\pi$) for each agent is as follows:
$$
\underbrace{d_{|T| - n} \succ d_{|T| - n - 1} \cdots \succ d_1}_{D}
\succ
\underbrace{x_n \succ x_{n-1} \cdots \succ x_1}_{X}
\succ
\underbrace{y_n \succ y_{n-1} \cdots \succ y_1}_{Y}
\succ
\underbrace{z_n \succ z_{n-1} \cdots \succ z_1}_{Z}
$$

Consider an agent corresponding to a tuple $t = (x_i, y_j, z_k) \in T$. We say that the goods $x_i, y_j$ and $z_k$ are the \emph{signature goods} for agent $t$. Our goal is to set up the valuations such that, whenever there is a perfect matching, the allocation that gives each agent its signature goods and assigns dummy goods to the remaining agents achieves the optimal Nash welfare. To achieve this, we carefully ``bump up'' the valuations of the signature goods for each agent. Let $B = 10, C = 3$ be constants. Then, we set the valuation function $v_t$ for agent $t = (x_i, y_j, z_k)$ as follows:

\begin{itemize}
    \item For each dummy good $d_r$, $v_t(d_r) = B^{10n + r}$.
    \item For each $r \neq i, v_t(x_r) = B^r$. Also, $v_t(x_i) = C \cdot B^i$.
    \item For each $r \neq j$, $v_t(y_r) = v_t(x_i) \cdot B^{-10n + r}$. Also, $v_t(y_j) = v_t(x_i) \cdot C \cdot B^{-10n + j}$.
    \item For each $r \neq k$, $v_t(z_r) = v_t(\{x_i, y_j\}) \cdot B^{-20n + r}$. Also, $v_t(z_k) = v_t(\{x_i, y_j\}) \cdot C \cdot B^{-20n + k}$.
\end{itemize}

It is easy to see that our instance has ordered lexicographic (in fact, also doubling) valuations, since for each agent, the value of each good is at least $B/C > 2$ times larger than the value of the next good in the common preference order $\pi$. The valuation functions are engineered to make the Nash product behave in \emph{phases}. 

First, dummy goods are so valuable that any Nash-optimal allocation must distribute them across distinct agents. Additionally, the gap in dummy and non-dummy goods is so large that in a Nash optimal allocation, any agent who receives a dummy good must not receive any other good. Hence, after allocating all the dummy goods, exactly $n$ agents remain to share the $3n$ goods in $X \cup Y \cup Z$. We will set up a threshold $\theta$ that is only achievable if and only if each of these $n$ agents is allocated all its signature goods. The construction makes the assignment of $X$ behave as the first ``phase'' of the objective. In particular, when we restrict attention to the $n$ agents without dummy goods, the Nash welfare can reach the threshold $\theta$ only when each of the $n$ goods in $X$ is given to an agent for whom it is a signature good. Any allocation that assigns some $x_r$ to an agent for whom $x_r$ is \emph{not} the signature loses a multiplicative factor (the $C$-boost) which cannot be compensated by any rearrangement of $Y$ and $Z$ goods since the goods in $X$ heavily dominate the goods in $Y \cup Z$.

The key idea is then to treat $Y$-goods as an \emph{increment relative to the agent's baseline after receiving its signature $x$}. 
Concretely, the values of goods $y\in Y$ are defined \emph{with respect to} the value of the signature $x$ so that once every agent (who does not receive a dummy good) has its signature $x$, giving an agent some $y$ produces essentially the same \emph{multiplicative} gain regardless of which agent it is (everyone is being ''scaled'' from their own $x$ baseline), except that the multiplicative is slightly higher when $y$ is the agent's signature $Y$-good. 
As a result, the Nash product no longer cares about how we permute the $y$-goods \emph{except} that assigning the \emph{signature} $y_j$ to agent $(x_i,y_j,z_k)$ gives a slightly larger boost; hence maximizing Nash welfare pushes the allocation toward making as many agents as possible get their signature $y$ simultaneously, which corresponds to choosing non-conflicting triples.

Finally, $z$-goods are defined similarly, but now \emph{relative to the utility from the union of the agent's signature $x$ and $y$}. 
This makes $z$ act as the third ``layer'' of controlled multiplicative improvement: once allocation of $X$ and $Y$ is fixed, allocating $Z$ is again almost neutral up to permutation, with an added premium for giving an agent its signature $z_k$. 
Thus, any allocation that achieves the maximum Nash product must align $x$, then $y$ given $x$, and then $z$ given $(x,y)$—exactly mirroring the requirement that the chosen triples form a perfect 3D matching. The corresponding threshold $\theta$ is defined as follows:
$$\theta = \left[\left(\displaystyle \prod_{r=1}^{|T| - n}B^{10n + r}\right) \left(\displaystyle \prod_{r=1}^{n}C \cdot B^{r}\right) \left(\displaystyle  \prod_{r=1}^{n} (1 + C \cdot B^{-10n + r})\right) \left(\displaystyle  \prod_{r=1}^{n}(1 + C \cdot B^{-20n + r})\right)\right]^{1/|T|}$$
The following lemma then establishes the correctness of the reduction. We include its proof in \Cref{sec:hardness_appendix}.

\begin{restatable}{lemma}{OrderedLexReduction}
    There exists a perfect 3D matching in the original instance if and only if there exists an allocation with Nash welfare at least $\theta$ in the constructed instance.
\end{restatable}
    \begin{remark}[Encoding size]
        We verify that the constructed instance has polynomial encoding size. By multiplying every agent's valuation by $B^{30n}$, all values become integers. The largest original value is $B^{10n+|T|-n}$ (for $d_{|T|-n}$), so after scaling every value is at most $B^{40n+|T|}$. Hence each scaled value has $\O(n+|T|)$ digits, i.e., $\O((n+|T|)\log B) = \O(n+|T|)$ bits. The threshold $\theta$ is also scaled by $B^{30n}$, and $\theta^{|T|}$ (after scaling) is a product of $|T|$ such integer values, hence an integer with $\O(|T|(n+|T|))$ bits. Hence, the entire constructed instance is encoded in $\mathrm{poly}(n,|T|)$ bits.
    \end{remark}


    For doubling valuations, we give a gap-preserving reduction from 4-dimensional matching (\Cref{thm:APX-hard_Nash_General_Lexicographic}). The proof of this result is presented in~\Cref{sec:apxhard_appendix}.

    \begin{restatable}[APX-hardness]{theorem}{NashLexAPXHard}
        It is NP-hard to approximate the Nash welfare objective to within a factor of \(0.9996\) even for doubling valuations.
    \label{thm:APX-hard_Nash_General_Lexicographic}
    \end{restatable}

\section{Conclusion and Future Directions}

We studied Nash welfare maximization under lexicographic, equivalently superincreasing additive, valuations. On the algorithmic side, we obtained a $(\sfrac{1}{\sqrt{2}}-\varepsilon)$-approximation for general lexicographic valuations via a tight analysis of the configuration LP, and exact polynomial-time algorithms for ordered and doubling subclasses when the number of agents is constant.  The exact algorithms are based on a domination-based branch-and-prune framework whose recursion tree is polynomially bounded by a mutual-exclusion argument.  On the hardness side, we showed that exact optimization is NP-hard even for ordered lexicographic valuations, and that APX-hardness persists even for doubling valuations. Overall, our work highlighted the rich computational landscape of Nash welfare for lexicographic valuations.

Several questions remain open.  First, can the $(\sfrac{1}{\sqrt{2}}-\varepsilon)$ guarantee be improved by using a relaxation stronger than the configuration LP, or is there a matching hardness threshold for lexicographic valuations?  Second, our exact algorithms are XP parameterized by the number of agents; it would be interesting to determine whether ordered or doubling lexicographic valuations admit an $f(n) \cdot \poly(m)$ exact algorithm for Nash welfare, or whether such an FPT algorithm is unlikely. Finally, it would be useful to understand whether our algorithmic results extend to \emph{weakly} lexicographic valuations that allow for indifference (ties) between items.

\section*{Acknowledgments}
We are grateful to Vignesh Viswanathan for helpful discussions. RV acknowledges support from DST INSPIRE grant no. DST/INSPIRE/04/2020/000107, ANRF grant no. CRG/2022/002621, and Mr. D.P. Gupta Chair Professorship. JY acknowledges support from a Google PhD fellowship.

\paragraph{AI usage.} The authors used GPT-5.5 for editorial tasks during the preparation of this manuscript, including proofreading for spelling and grammar and assisting with the preparation and formatting of figures. AI assistance was not used to develop the technical content of the work. All technical ideas, including the algorithms and hardness proofs, were conceived and developed by the authors.

\providecommand\showeprint[2][arxiv]{%
	arXiv:\href{https://arxiv.org/abs/#2}{#2}%
}

\bibliographystyle{plainnat}
\bibliography{arxiv_bib}

\clearpage

\appendix

\begin{center}
    \Large{\textbf{Appendix}}
\end{center}

\section{Additional Related Work: Generalized Mean welfare}
\label{appendix:Additional_Related_Work}

The class of $p$-mean welfare measures has also received attention in recent years. For subadditive valuations, an $\Omega(\sfrac{1}{n})$-approximation algorithm is known for all weighted $p$-mean welfare measures with $p \in (-\infty,1]$~\citep*{CGM21fair}. This approximation is essentially tight, since $\omega(\sfrac{1}{n})$ approximations are known to require an exponential number of value queries~\citep*{DNS05approximation,BBK+20tight}. For additive valuations, the complexity landscape of weighted $p$-mean welfare maximization is more diverse: While the weighted utilitarian welfare ($p=1$) is easy to maximize, the weighted Nash welfare ($p=0$) admits an approximation algorithm of $e^{-1/e}\approx 0.692$~\citep*{FL25note} and is \NPH{} to approximate beyond a factor of $\sqrt{\sfrac{7}{8}} \approx 0.935$~\citep*{GHM24satiation}, and the egalitarian welfare ($p \rightarrow -\infty$) admits an $\tilde{\Omega}(1/n^\varepsilon)$ approximation for $n$ agents~\citep*{CCK09allocating} and is \NPH{} to approximate beyond a factor of $1/2$~\citep*{BD05allocating,CCK09allocating}. 
For identical additive valuations, a polynomial-time approximation scheme (\PTAS{}) is known for weighted $p$-mean welfare~\citep*{GHM+22tractable}. Additionally, for a constant number of agents with possibly non-identical additive valuations, a fully polynomial-time approximation scheme (\FPTAS{}) is known for weighted $p$-mean welfare~\citep*{GHM+22tractable}.

\section{Extensions to Generalized \texorpdfstring{$p$}{p}-Mean Welfare Measures}
\label{appendix:Extensions_p_Mean}

The exact algorithms in \Cref{sec:Results_Subclasses} are stated for separable concave objectives
\[
        \WH(A)=\sum_{i\in N} h_i(v_i(A_i)).
\]
This formulation immediately captures weighted Nash welfare by taking $h_i(x)=w_i\log x$.  It also captures weighted $p$-mean welfare for every $p\in(-\infty,1)\setminus\{0\}$ by taking $h_i(x)=w_i\,\sign(p)\,x^p$. Indeed, maximizing the weighted $p$-mean is equivalent to maximizing $\sign(p)\sum_i w_i v_i(A_i)^p$, because the outer power $1/p$ is increasing for $p>0$ and decreasing for $p<0$.

For ordered lexicographic valuations, \Cref{thm:Ordered_Constant_n_Weighted} therefore applies to all weighted $p$-mean welfare objectives with $p<1$, $p\ne0$, as well as to weighted Nash welfare in the limit $p\to0$.  For doubling valuations, \Cref{thm:Doubling_Constant_n} requires the additional scaling condition
\[
        \Delta_i(2x,2y)\ge \Delta_i(x,y) \qquad \text{for all }x,y\ge0.
\]
This condition holds for $h_i(x)=w_i\log x$ and for $h_i(x)=w_i x^p$ when $p\in(0,1)$, and hence the doubling-valuations exact algorithm extends to weighted Nash welfare and to weighted $p$-mean welfare for $p\in(0,1)$.

\section{Integrality gap of the config LP for lexicographic valuations}
\label{appendix:config_lp_gap}

In this section, we show that the config LP has an integrality gap of $\sqrt{2}$ in the case of lexicographic valuations. Let $L$ be a large enough positive integer, and let $\delta_L = 1/L$ be such that $\lim_{L \to \infty} L^{\delta_L} = 1$. In the instance, there are $4$ agents $a_1, a_2, a_3, a_4$ and $5$ goods $g_1, g_2, g_3, g_4, g_5$. The weights of the agents are $1/2-\delta, 1/2-\delta, \delta, \delta$ respectively. Let $\gamma_L  = 2^{-L^2}$. The following table lists the valuation matrix. We use $\delta, \gamma$ instead of $\delta_L, \gamma_L$ for brevity.

\begin{table}[H]
    \centering
    \begin{tabular}{c|ccccc}
         & $g_1$ & $g_2$ & $g_3$ & $g_4$ & $g_5$ \\
         \hline
         $a_1$& $L$ & $1+4\gamma$ & $1$ & $2\gamma$ & $\gamma$\\
         $a_2$& $L$ & $2\gamma$ & $\gamma$ & $1+4\gamma$ & $1$\\
         $a_3$& $4\gamma$ & $1+8\gamma$ & $1$ & $2\gamma$ &$\gamma$ \\
         $a_4$& $4\gamma$ & $2\gamma$ & $\gamma$ & $1+8\gamma$ & $1$\\
    \end{tabular}
    \caption{An instance with integrality gap approaching $\sqrt{2}$.}
    \label{tab:integrality_gap}
\end{table}

It can be easily verified that this instance has lexicographic valuations for large enough $L$. First, consider the following fractional solution to the config LP: 
\[y_{a_1, \{g_1\}} = y_{a_1,\{g_2,g_3\}} = 1/2,\]
\[y_{a_2, \{g_1\}} = y_{a_2, \{g_4,g_5\}} = 1/2,\]
\[y_{a_3, \{g_2\}} = y_{a_3, \{g_3\}} = 1/2,\]
\[y_{a_4, \{g_4\}} = y_{a_4, \{g_5\}} = 1/2.\]

The LP objective for this LP solution therefore is:
\[\left(\frac{1}{2} - \delta\right) \ln L + \left(\frac{1}{2} - \delta\right)\ln (2+4\gamma) + \delta \ln (1 + 8 \gamma)\]

For the allocation $A = (\{g_1\}, \{g_4\}, \{g_2,g_3\}, \{g_5\})$, the weighted Nash welfare is $L^{1/2 - \delta} (1+4\gamma)^{1/2-\delta} (2+8\gamma)^\delta$ which goes to $\infty$ when $L$ approaches $\infty$. For any other allocation to have a better weighted Nash welfare, note that each agent must receive utility at least one since $L^{1/2-\delta}O(\gamma)^\delta$ approaches $0$ as $L \to \infty$, as $\delta=1/L, \gamma=2^{-L^2}$. Also, $g_1$ must be allocated to $a_1$ or $a_2$ since otherwise for large enough $L$, the utility of each agent is at most $3$. Hence, the weighted Nash welfare is also at most $3$, which is less than that of $A$ for large enough $L$. Since agents $a_1,a_2$ are symmetric, assume that $g_1$ is allocated to $a_1$. Then, $g_4,g_5$ must be given one each to $a_2,a_4$ and both of them have utility $1 + O(\gamma)$. $a_3$ has utility at most $2 + O(\gamma)$, and $a_1$ has utility at most $L + O(1)$. Thus, the logarithm of the weighted Nash welfare of the optimal allocation is

\[\left(\frac{1}{2} - \delta\right) \ln (L+O(1)) + \left(\frac{1}{2} - \delta\right) \ln(1+O(\gamma)) + \delta \ln (2 + O(\gamma)) + \delta \ln (1 +  O(\gamma)) \]

Recall that $\delta = 1/L, \gamma = 2^{-L^2}$. Thus, as $L \to \infty$, we have $\delta, \gamma \to 0$ and $\ln(\frac{L}{L + O(1)}) \to 0$. Therefore, difference in the fractional and integral optima of the LP approaches $1/2 \ln 2$.

\section{Proof of~\Cref{thm:ratio_root_2}}
\label{sec:root_2_proof}
In this section, we prove~\Cref{thm:ratio_root_2}.
\StructuralRootTwo*

We will use $\val(e_i)$ to denote the value of the good $e_i$. If $m'<n$ or if at least one of $e_1, \ldots e_n$ have value $0$, then every allocation gives utility $0$ to at least one agent, so
both Nash products are $0$. Hence the theorem is trivial. Assume from now on
that $m'\ge n$. We may also assume that $m'\ge 2n$. Indeed, if $m'<2n$,
append a dummy type with $2n-m'$ zero-valued copies. This number is at most
$n$, since $m'\ge n$. Complete $A^\VV$ by assigning these dummy copies to
distinct agents. Complete $A^\BB$ by assigning them to exactly the agents not
yet represented in the second aligned block. Then $A^\VV$ remains valid,
$A^\BB$ remains block-balanced, and all utilities, hence both Nash welfares,
are unchanged.

First, let us establish the following lemma that shall be crucial in structurally modifying the allocation $A^\VV$ without decreasing the Nash welfare.
\begin{lemma}[Safe transfer]
\label{lem:safe-transfer}
Consider an allocation under identical additive valuations with a non-zero Nash product. Let $U_i$ and
$U_j$ be the utilities of agents $i$ and $j$, respectively, and suppose
agent $i$ holds a good $g$ of value $x$. If $U_i-x\ge U_j$, then transferring $g$ from $i$ to $j$ does not decrease the Nash welfare. If the inequality is strict, then the Nash welfare strictly increases.
\end{lemma}

\begin{proof}
Only the utilities of $i$ and $j$ change. The two-agent product changes by
\[
        (U_i-x)(U_j+x)-U_iU_j
        =
        x\bigl((U_i-x)-U_j\bigr)
        \ge 0.
\]
Thus the Nash product, and hence the Nash welfare, does not decrease. If
$U_i-x>U_j$, the displayed change is strictly positive.
\end{proof}

\subsection{Aligning the first block}

Because $A^\BB$ is block-balanced, the first $n$ goods $e_1,\ldots,e_n$
are assigned to distinct agents in $A^\BB$. We first show that $A^\VV$
may also be assumed to assign these $n$ goods to distinct agents, without
decreasing its Nash product. Let $t$ be the smallest index such that $c_1+\cdots+c_t\ge n$. Thus the first $n$ goods consist of all copies of $g_1,\ldots,g_{t-1}$, together with some copies of $g_t$. 

Suppose some agent $i$ receives at least two first block goods, and an agent $j$ receives none in $A^\VV$. Let $g_s$ be the type of the least valuable first-block good held by $i$. Let $g_h$ be the highest-valued good type in agent $j$'s bundle. If $h = s$, then we can swap the goods of type $g_s$ in the two bundles, without changing the Nash welfare. Else if $h > s$, then we can transfer the good of type $g_s$ from agent $i$ to agent $j$. This is because, the utility $U_i$ of agent $i$ in $A^\VV$ is at least $v_s + v_{s - 1}$, and utility $U_j$ of agent $j$ is at most $v_{h} + v_{h+1} + \ldots + v_m \leq v_{h-1} \le v_s$. Hence, \Cref{lem:safe-transfer} applies and the transfer does not decrease the Nash product. 

Repeating this finitely many times yields a valid allocation, still denoted $A^\VV$, in which the first $n$ goods go to distinct agents. Next, modify $A^\BB$ by making agents swap whole bundles so that agent $i$ receives $e_i$, for every $i$. This preserves the block-balanced property and leaves the Nash
product unchanged. Modify $A^\VV$ in the same way, by making agents swap whole bundles so that agent $i$ receives $e_i$, for every $i$; this keeps $A^\VV$ valid and also leaves its Nash product unchanged. We therefore have,
for every $i$,$e_i\in A^\VV_i\cap A^\BB_i$.

\subsection{The $L$-normalization}

Let $L = 1 + 2\sum_{i=1}^m v_i$ be a large enough value. Let $H$ denote the set of agents who prefer their bundle in $A^\BB$ over $A^\VV$:
\[
        H:=\{i\in N:v(A^\VV_i)<v(A^\BB_i)\}.
\]

If $|H| = n$, then all agents prefer $A^\BB$ to $A^\VV$ and we are done. Thus, we will assume $|H| < n$. We will now modify the valuations as follows: create a new good type with value $L$. For each agent $i \in H$, remove $e_i$ from the instance and replace it with a copy of the new good type of value $L$. The allocations $A^\VV$ and $A^\BB$ remain unchanged. In the new ordering of the goods, all the copies of the new good are now before the other goods. The first block of $n$ goods in the ordering however remains unchanged, and all these goods are still allocated to distinct agents. Thus $A^\VV$ and $A^\BB$ remain valid and block-balanced respectively.

We will now argue that this modification does not decrease the Nash welfare ratio $\rho$. Indeed, for any agent $i \in H$, let $\alpha_i$ be the total value of the goods in $A^\VV_i$ excluding $e_i$, and let $\beta_i$ be the total value of the goods in $A^\BB_i$ excluding $e_i$. Note that, since $\alpha_i < \beta_i$, the function $\frac{x + \alpha_i}{x + \beta_i}$ is increasing in $x$. Also, since $L > \val(e_i)$, we have that $\frac{L + \alpha_i}{L + \beta_i} > \frac{\val(e_i) + \alpha_i}{\val(e_i) + \beta_i}$. Thus, the Nash welfare ratio $\rho$ does not decrease after this modification.

We renumber the goods to $e_1, e_2, \ldots e_{m'}$. The $L$-valued goods form a prefix of this sequence. We also renumber the good types to $g_1, g_2, \ldots$ (and the corresponding values to $v_1, v_2, \ldots$), so that $g_1$ is now the good of type $1$ and $v_1 = L$. Let $c_i$ be the number of copies of good type $g_i$. Also, we again perform bundle permutations to ensure that agent $i$ gets $e_i$ in both $A^\VV$ and $A^\BB$, for every $i$. Note that, by construction, for all agents $i \in \{c_1 + 1, c_1 + 2, \ldots, n\}$, we have that $v(A^\VV_i) \geq v(A^\BB_i)$.

\subsection{Compressing the first block to three mini-blocks}

Let $t$ be the smallest integer such that $c_1 + c_2 + \ldots c_t \geq n$. For each good type $g_j$ with $j \in \{2, 3, \ldots t-1\}$, we will decrease its value $v_j$ to $v_t + v_{t+1} + \ldots v_m$. Note that this is a decrement in value, since for any such $j$, $v_j \geq v_{t-1} \geq v_t + v_{t + 1} + \ldots v_m$. This new value will correspond to a new good type with $c_2 + c_3 + \ldots c_{t-1}$ copies. For any agent $i \in \{c_1 + 1, \ldots c_1 + c_2 + \ldots c_{t-1}\}$, we argue that their ratio only increases. Indeed, let $\alpha_i$ denote the value of all the goods in $A^\VV_i \setminus e_i$ and $\beta_i$ denote the value of all the goods in $A^\BB_i \setminus e_i$. Since $v(A^\VV_i) \geq v(A^\BB_i)$, we have $\alpha_i \geq \beta_i$. Thus, the function $\frac{x + \alpha_i}{x + \beta_i}$ is decreasing and the ratio only increases upon decreasing $x$ from $\val(e_i)$ to $v_t + v_{t+1} + \ldots v_m$. Formally, in the new instance, the value $v_1$ and the goods of type $g_1$ remain the same, the goods of types $g_2,g_3,\ldots g_{t-1}$ now have type $2$ and $c_2, v_2$ now become $c_2 + c_3 + \ldots c_{t-1}$, $v_t+v_{t+1}+\ldots v_m$ respectively. For each $\ell \geq t$, the goods that earlier had type $g_\ell$ now have type $g_{\ell - t + 3}$ and $c_{\ell - t + 3}, v_{\ell-t+3}$ now become the earlier $c_\ell, v_\ell$. The sequence $e_1, e_2, \ldots e_{m'}$ remains the same, and the instance is still lexicographic by construction. The first block of $n$ items consists of $3$ good types $g_1, g_2, g_3$ with values $v_1 = L, v_2, v_3$ respectively and number of first-block copies $c_1,c_2,n-c_1-c_2$ respectively. We denote these as the three (possibly empty) mini-blocks $A,B,C$. Let $r:=c_1+c_2+c_3-n$ be the number of remaining copies of $g_3$. These appear in the second block. Note that the mini-blocks $A$ and $B$ can be empty, but the mini-block $C$ is non-empty by definition. The instance now looks as in ~\Cref{fig:three-block-normal-form}. 

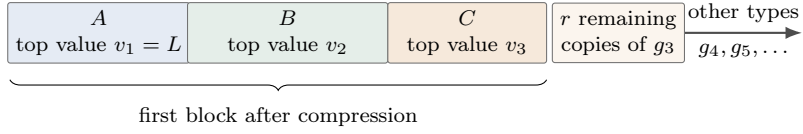
\begin{figure}[ht]
\centering
\begin{tikzpicture}[x=1cm,y=1cm]
        \node[proofblock, fill=proofblue!13, minimum width=2.35cm] (A)
                at (1.175,0) {$A$\\ top value $v_1=L$};
        \node[proofblock, fill=proofgreen!14, minimum width=2.65cm] (B)
                at (3.675,0) {$B$\\ top value $v_2$};
        \node[proofblock, fill=prooforange!17, minimum width=2.1cm] (C)
                at (6.05,0) {$C$\\ top value $v_3$};

        \node[proofblock, fill=prooforange!8, minimum width=1.6cm,
                minimum height=.58cm] (Gthree) at (8.05,0)
                {$r$ remaining\\ copies of $g_3$};

        \draw[proofarrow] (Gthree.east) -- ++(1.55,0)
                node[midway, above, font=\scriptsize] {other types}
                node[midway, below, font=\scriptsize] {$g_4,g_5,\ldots$};

        \draw[decorate, decoration={brace, mirror, amplitude=4pt}]
                (0,-.62) -- (7.1,-.62)
                node[midway, yshift=-14pt, font=\scriptsize]
                {first block after compression};

\end{tikzpicture}
\caption{The normalized first block: agents are grouped into three mini-blocks $A$, $B$, and
$C$. There are $r$ remaining goods of type $g_3$.}
\label{fig:three-block-normal-form}
\end{figure}
\subsection{Second block allocation in $A^\BB$}
\label{sec:second_block_XB}

Now discard from the denominator all copies after $e_{2n}$. This can only decrease $\WNash(A^\BB)$ and hence can only increase $\rho$. The first two aligned blocks contain exactly one copy for each agent in $A^\BB$.This might result in the allocation $A^\BB$ no longer being block-balanced (since it is not a complete allocation, a requirement for block-balancedness). However, we would not require block-balancedness of $A^\BB$ going forward.  We next
replace the denominator's assignment of the second block by the sorted pairing in which agent $i$ receives $e_i$ and $e_{n+i}$.

Let $a_i:=\val(e_i)$ and $b_i:=\val(e_{n+i})$. Both sequences are
nonincreasing. For $i<k$, if agent $i$ receives the smaller second value and
agent $k$ receives the larger one, then swapping these two second-block copies
changes the two-agent denominator product from the crossed value to the aligned
value, and the difference is
$(a_i+b_k)(a_k+b_i)-(a_i+b_i)(a_k+b_k)=(a_i-a_k)(b_i-b_k)\ge0$.
Thus removing inversions weakly decreases the denominator product. The final
sorted pairing is valid, because positions $i$ and $n+i$ cannot contain copies
of the same type: a type has at most $n$ copies, while these positions are
$n$ apart. We keep the notation $A^\BB$ for this reduced denominator vector.
Consequently, after this point the denominator utility of agent $i$ is
$\val(e_i)+\val(e_{n+i})$.

\subsection{Allocation of the remaining copies of $g_3$ in $A^\VV$}
\label{sec:g_3_rem_XV}
 Since $A^\VV$ is valid, the remaining copies of $g_3$ can only be allocated to agents (whose top items are) in $A$ and $B$. First, while there is an agent $i \in A$ who receives a copy of $g_3$ and an agent $j \in B$ who does not, we can transfer the copy of $g_3$ from agent $i$ to agent $j$. This only increases the Nash welfare by~\Cref{lem:safe-transfer}. This is because agent $i$ already has a copy of $v_1$ that is more valuable than the rest of the good types combined. Additionally, within $A$ and $B$, we can assume that the agents receiving a copy of $g_3$ form some prefix of $A$ and $B$ respectively. Indeed, if agent $i$ receives a copy of $g_3$ and agent $i+1$ does not, where $i$ and $i + 1$ belong to the same mini-block, we can swap the tails (everything except the top good) of the bundles of agents $i$ and $i + 1$. This maintains the validity and the Nash welfare of $A^\VV$. 

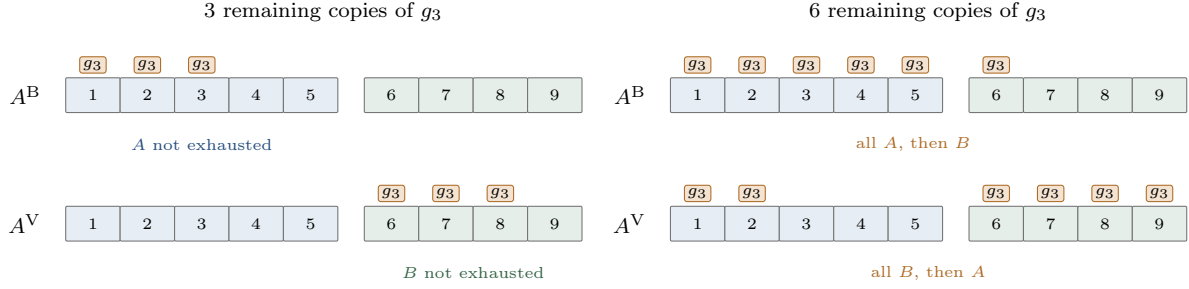
\begin{figure}[ht]
\centering
\begin{tikzpicture}[x=1cm,y=1cm]
        \begin{scope}
                \node[font=\scriptsize] at (3.05,1.95) {3 remaining copies of $g_3$};
                \node[font=\scriptsize, anchor=east] at (-.55,.85) {$A^\BB$};
                \node[font=\scriptsize, anchor=east] at (-.55,-.85) {$A^\VV$};

                \foreach \x/\agent in {0/1,.72/2,1.44/3,2.16/4,2.88/5}{
                        \node[proofcell, fill=proofblue!12] at (\x,.85) {$\agent$};
                        \node[proofcell, fill=proofblue!12] at (\x,-.85) {$\agent$};
                }
                \foreach \x/\agent in {3.95/6,4.67/7,5.39/8,6.11/9}{
                        \node[proofcell, fill=proofgreen!14] at (\x,.85) {$\agent$};
                        \node[proofcell, fill=proofgreen!14] at (\x,-.85) {$\agent$};
                }

                \foreach \x in {0,.72,1.44}{
                        \node[proofcopy] at (\x,1.25) {$g_3$};
                }
                \foreach \x in {3.95,4.67,5.39}{
                        \node[proofcopy] at (\x,-.45) {$g_3$};
                }

                \node[font=\tiny, text=proofblue!75!black] at (1.44,.2)
                        {$A$ not exhausted};
                \node[font=\tiny, text=proofgreen!75!black] at (5.03,-1.5)
                        {$B$ not exhausted};
        \end{scope}

        \begin{scope}[xshift=8.0cm]
                \node[font=\scriptsize] at (3.05,1.95) {6 remaining copies of $g_3$};
                \node[font=\scriptsize, anchor=east] at (-.55,.85) {$A^\BB$};
                \node[font=\scriptsize, anchor=east] at (-.55,-.85) {$A^\VV$};

                \foreach \x/\agent in {0/1,.72/2,1.44/3,2.16/4,2.88/5}{
                        \node[proofcell, fill=proofblue!12] at (\x,.85) {$\agent$};
                        \node[proofcell, fill=proofblue!12] at (\x,-.85) {$\agent$};
                }
                \foreach \x/\agent in {3.95/6,4.67/7,5.39/8,6.11/9}{
                        \node[proofcell, fill=proofgreen!14] at (\x,.85) {$\agent$};
                        \node[proofcell, fill=proofgreen!14] at (\x,-.85) {$\agent$};
                }

                \foreach \x in {0,.72,1.44,2.16,2.88,3.95}{
                        \node[proofcopy] at (\x,1.25) {$g_3$};
                }
                \foreach \x in {3.95,4.67,5.39,6.11,0,.72}{
                        \node[proofcopy] at (\x,-.45) {$g_3$};
                }

                \node[font=\tiny, text=prooforange!85!black] at (2.85,.2)
                        {all $A$, then $B$};
                \node[font=\tiny, text=prooforange!85!black] at (3.05,-1.5)
                        {all $B$, then $A$};
        \end{scope}
\end{tikzpicture}
\caption{Two examples with five agents in $A$ and four agents in $B$, numbered
$1,\ldots,9$ from left to right. With
three remaining copies of $g_3$, $A^\BB$ does not exhaust $A$ and $A^\VV$
does not exhaust $B$. With six remaining copies, $A^\BB$ exhausts $A$ before
assigning to $B$, while $A^\VV$ exhausts $B$ before assigning to $A$.}
\label{fig:g3-grouping-example}
\end{figure}

For $q\in\{1,\ldots,m+1\}$, let $V_q:=\sum_{\ell\ge q}v_\ell$, with
$V_{m+1}=0$. We have the following inequalities from the lexicographic nature of
the values:
\[
  v_1\ge v_2+v_3+V_4,\qquad
  v_2\ge v_3+V_4,\qquad
  v_3\ge V_4 .
\]
The proof strategy is to partition the agents into groups of size at most $3$.
Call a group $G$ \emph{safe} if
$\prod_{i\in G}v(A^\VV_i)\le2^{|G|/2}\prod_{i\in G}v(A^\BB_i)$. So, a singleton is safe if its utility ratio is at most $\sqrt{2}$, a pair is safe if the ratio of the product of their utility ratios is at most $2$, and a triple is safe if the ratio of the product of their utility ratios is at most $2\sqrt{2}$. If every group is safe, then the overall ratio $\rho$ of geometric mean of utilities is at most $\sqrt2$. We build these
groups incrementally. An agent is \emph{unassigned} if it has not yet been placed in any group, and \emph{assigned} otherwise. To make progress in the proof, we will either form new groups from unassigned agents or add an unassigned agent to an existing group, while always maintaining that the groups are safe.

First, note that any agent from $A \cup B$ who receives $g_3$ in both the allocations can already be assigned to a singleton group: if its first-block value is
$\theta\in\{v_1,v_2\}$, then its ratio is at most
$(\theta+v_3+V_4)/(\theta+v_3) = 1 + \frac{V_4}{\theta + v_3}\le4/3<\sqrt2$. The first inequality holds since $v_3 \geq V_4$ and $\theta \geq v_2 \geq v_3 + V_4 \geq 2V_4$.

Then, note via the prefix condition that there is no agent in $A$ who receives $g_3$ only in $A^\VV$ and no agent in $B$ who receives $g_3$ only in $A^\BB$. Also, the number of agents in $A$ receiving $g_3$ only in $A^\BB$ and the same number of $B$-agents
receiving $g_3$ only in $A^\VV$ are the same. Match them arbitrarily into pairs. For a matched pair $(a,b)$, call $a$ the anchor and $b$ the representative; we also call $b$ a \emph{bad $B$-representative}. Assign both $a$ and $b$ to this pair, but keep the pair open and represented by $b$. The product of their utility ratios is at most $2$: the anchor contributes ratio at most $(v_1+V_4)/(v_1+v_3)\le1$, while the representative contributes ratio at most $(v_2+v_3+V_4)/v_2\le2$. Thus, the product of the ratios is at most $2$, and the pair is safe. However, we do not finalize this pair yet and call it an \emph{open pair}. 

For example, in the example with $3$ remaining $g_3$ copies from~\Cref{fig:g3-grouping-example}, agents $\{1,2,3\}$ receive $g_3$ only in $A^\BB$ and agents $\{6, 7, 8\}$ receive $g_3$ only in $A^\VV$. We can match them into pairs $(1,6),(2,7),(3,8)$. In the example with $6$ remaining $g_3$ copies, agents $\{1, 2, 6\}$ receive $g_3$ in both allocations and are thus assigned as singletons. Agents $\{3,4,5\}$ receive $g_3$ only in $A^\BB$ and agents $\{7,8,9\}$ receive $g_3$ only in $A^\VV$. We can match them into pairs $(3,7),(4,8),(5,9)$.

For the rest of the proof, we call the agents that are either unassigned or are the representative of some pair, as an \emph{open agents}. The open agents, at this point, are the precisely $\{r+1,r+2,\ldots, n\}$, that is, the agents who do not receive any good of type $g_3$ from the second block in $A^\BB$. The open agents have four possible classes presented in the table below. The second and third columns record the numerator and denominator bases:
\[
\begin{array}{c|c|c}
\text{class} &
\begin{array}{c}\text{numerator utility derived}\\\text{from types }\{g_1,g_2,g_3\}\end{array} &
\begin{array}{c}\text{denominator utility derived}\\\text{from types }\{g_1,g_2,g_3\}\end{array}\\ \hline
\text{unassigned }A\text{-singleton} & v_1 & v_1\\
\text{bad }B\text{-representative} & v_2+v_3 & v_2\\
\text{unassigned }B\text{-singleton} & v_2 & v_2\\
\text{unassigned }C\text{-singleton} & v_3 & v_3
\end{array}
\]

\subsection{Allocation of the lower types}

If $V_4=0$, every remaining singleton has ratio $1$, and every remaining open pair is safe by the preceding paragraph on open pairs. Hence assume $V_4>0$. For $j\ge3$, let $P_j$ be the total number of copies of types $g_1,\ldots,g_j$. Thus $P_3=n+r$. We process the unaccounted good types in the order $g_4,g_5,\ldots, g_m$. The invariant before processing $g_j$ is the following. 

\textbf{Invariant:} The open agents, if any, form a consecutive interval $\ell_j,\ell_j+1,\ldots,h_j$, where $\ell_j=P_{j-1}-n+1$ and $h_j\le n$;
none of these representatives has received any already processed lower type in
either allocation. 

This invariant is initially satisfied since the open agents are $\{r+1, \ldots, n\}$ and $P_3=n+r$. Suppose we have already processed $g_4,g_5,\ldots g_{j-1}$ and the invariant holds. By the sorted form of the denominator allocation from \Cref{sec:second_block_XB}, agent $i$ receives the second-block item $e_{n+i}$. Hence, among the current open agents, the agents who receive a copy of $g_j$ in $A^\BB$ are exactly
\[
        I_j=\{\ell_j,\ell_j+1,\ldots,\min(h_j,\ell_j+c_j-1)\}.
\]
Thus $I_j$ is a prefix of the current open interval $\{\ell_j,\ell_j+1,\ldots,h_j\}$.

We next modify $A^\VV$, without decreasing its Nash product, so that the open
agents who receive a copy of $g_j$ in $A^\VV$ form a suffix of the current
open interval. Indeed, suppose there are open agents $i<k$ such that $i$
receives a copy of $g_j$ in $A^\VV$ and $k$ does not. By the invariant, neither
agent has received any of the already processed lower types
$g_4,\ldots,g_{j-1}$. If $i$ and $k$ have the same bundle of types from
$\{g_1,g_2,g_3\}$, we may swap their lower-type tails, leaving the Nash product
unchanged. Otherwise, by the ordering of the open classes established after the
$g_3$ step, agent $i$ has a strictly more valuable base from
$\{g_1,g_2,g_3\}$ than agent $k$. The lexicographic inequalities imply that this
base advantage is at least the total value $V_{j+1}$ of all types below $g_j$.
Therefore, after removing the copy of $g_j$ from $i$, agent $i$ still has utility
at least that of agent $k$, and \Cref{lem:safe-transfer} allows us to transfer
this copy of $g_j$ from $i$ to $k$ without decreasing the Nash product. Repeating
this operation finitely many times yields a numerator allocation in which the
open agents receiving $g_j$ form some suffix
\[
        T_j=\{s_j,s_j+1,\ldots,h_j\}
\]
of the current open interval. Note that this modification does not affect the \emph{safety} of open pairs. This is because, when an open pair was opened, the upper bound of $2$ on the product of ratios for the two agents in the pair was proven for the fixed allocation of $\{g_1,g_2,g_3\}$ good types, and the worst possible allocation of the lower types.

We now argue in the following lemma that we can finalize the groups for the agents $I_j \cup T_j$, that is, all the open agents who receive a copy of $g_j$.

\begin{lemma}[One-step certification]
\label{lem:one-step-certification}
 All agents in $I_j\cup T_j$ can be assigned to safe groups and removed from the open list. If the set of open agents is non-empty after this removal, the invariant is maintained.
\end{lemma}

\begin{proof}
Note that $|I_j| \geq |T_j|$, as $I_j$ is either the set of all the open agents, or has a size equal to the total number of goods $c_j$ of type $g_j$. Let $\Lambda_j:=I_j\setminus T_j$ and $\Gamma_j:=T_j\setminus I_j$. Since
$I_j$ forms a prefix and $T_j$ a suffix of the open agents, every agent of $\Lambda_j$ lies to the left of every agent of $\Gamma_j$, and $|\Gamma_j|\le |\Lambda_j|$.
By lexicographicity, $V_{j+1}\le v_j$, $v_3\ge V_j$, and $v_2\ge v_3+V_j$.

Agents in $I_j\cap T_j$ are immediately safe. A singleton there has top value $\theta\ge v_3$ and ratio at most $(\theta+V_j)/(\theta+v_j)\le4/3<\sqrt2$. An open pair there has
anchor ratio at most $1$ and representative ratio at most $(v_2+v_3+V_j)/(v_2+v_j)\le2$.

For an agent in $\Lambda_j$, a singleton is safe by itself because it gains $v_j$ in the numerator and can gain only $V_{j+1}\le v_j$ in the denominator; thus it has ratio at most $1 < \sqrt{2}$. An open pair in
$\Lambda_j$ is safe unless it is needed for a match below; by itself its
representative ratio is at most $(v_2+v_3+V_{j+1})/(v_2+v_j)\le2$. An open pair
in $\Gamma_j$ is also safe by itself, since its representative ratio is at most
$(v_2+v_3+V_j)/v_2\le2$.

It remains to handle singleton agents in $\Gamma_j$. Match each such singleton
injectively to a distinct agent of $\Lambda_j$. If the matched left agent is a
singleton, the two singletons form a safe pair: the left ratio is at most $1$,
and the right ratio is at most $(\theta+V_j)/\theta\le2$ because
$\theta\ge v_3\ge V_j$.

If the matched left agent is an open pair $(a',b')$ and the gain singleton is
$c'$, put all three agents in one group. The anchor $a'$ has ratio at most $1$.
For the representative $b'$ and singleton $c'$, it is enough to bound
\begin{equation*}
  \frac{v_2+v_3+V_{j+1}}{v_2+v_j}\cdot
  \frac{v_3+V_j}{v_3} .
\end{equation*}
Since $V_j=v_j+V_{j+1}$ and $v_j\ge V_{j+1}$, we have
$0\le V_{j+1}\le V_j/2$. Thus we may write
$v_3=V_j+\alpha$ and $v_2=v_3+V_j+\beta$ with $\alpha,\beta\ge0$. The displayed
product is at most
\begin{equation*}
  \frac{3V_j+2\alpha+\beta+V_{j+1}}
       {3V_j+\alpha+\beta-V_{j+1}}\cdot
  \frac{2V_j+\alpha}{V_j+\alpha}.
\end{equation*}
The first factor is at least $1$, so increasing $\beta$ decreases it. Thus
$\beta=0$ is worst. A direct cross multiplication gives
\begin{equation*}
  \frac{3V_j+V_{j+1}}{3V_j-V_{j+1}}\cdot2
  -
  \frac{3V_j+2\alpha+V_{j+1}}{3V_j+\alpha-V_{j+1}}\cdot
  \frac{2V_j+\alpha}{V_j+\alpha}
  =
  \frac{\alpha(3V_j^2+6V_jV_{j+1}-V_{j+1}^2+4\alpha V_{j+1})}
       {(3V_j-V_{j+1})(3V_j+\alpha-V_{j+1})(V_j+\alpha)}
  \ge0 .
\end{equation*}
Therefore the product is at most
$2(3V_j+V_{j+1})/(3V_j-V_{j+1})\le14/5<2\sqrt2$, using
$V_{j+1}\le V_j/2$. Hence the triple is safe. After these matches, all unmatched agents of $\Lambda_j$ are handled by the loss cases above.

Since we remove a prefix and a suffix from the set of open agents, the remaining set of open agents (if any) still forms a contiguous interval, and starts at $\ell_{j+1} = P_j - n + 1$. Thus, the invariant is maintained.
\end{proof}

We process the good types $g_4,g_5,\ldots$ according to
\Cref{lem:one-step-certification} until no open $C$-agent remains; this stopping
point is well-defined because the invariant forces the open interval to move
past $C$. Indeed, after type $g_j$ is processed, if any open agents remain, they
form an interval starting at $\ell_{j+1}=P_j-n+1$ and ending at some
$h_{j+1}\le n$. Since $P_j$ eventually reaches at least $2n$, this left endpoint
would eventually be at least $n+1$, which is impossible for a nonempty open
interval. Thus all open agents, and in particular all open $C$-agents, are
removed no later than the point when the second block has been fully processed.
Suppose the open $C$-agents are exhausted just before the next type to be
processed is $g_j$. Then $j\ge5$, since the set of $C$-agents was nonempty. At
this point, we can finalize all remaining open singleton agents as safe
singleton groups and finalize all remaining open pairs as safe pairs. This is
because any remaining open singleton has base at least $v_2$ in both allocations
and can gain only goods of total value at most $V_j\le V_5$ in $A^\VV$. Since
$v_2 \ge v_3 + v_4 + V_5 \ge (v_4 + V_5) + v_4 + V_5 \ge 4V_5$, the ratio of
the utilities in the two allocations for any remaining open singleton is at most
$5/4 < \sqrt{2}$. The open pairs are already proven to be safe, so they can be
finalized.

After this finalization, all agents have been partitioned into groups $G$ of
size $1$, $2$, or $3$ such that
\[
        \prod_{i\in G}v(A^\VV_i)
        \le
        2^{|G|/2}\prod_{i\in G}v(A^\BB_i)
\]
for every $G$. Multiplying over the groups gives
\[
        \prod_{i=1}^n v(A^\VV_i)
        \le
        2^{n/2}\prod_{i=1}^n v(A^\BB_i),
\]
and taking
$n$-th roots gives $\rho\le\sqrt2$. Since every transformation weakly increased
$\rho$, the original ratio was also at most $\sqrt2$, thus proving~\Cref{thm:ratio_root_2}.

\section{Missing Proofs from Section~\ref{sec:Results_Subclasses}} 
\label{sec:algo_appendix}

\DominationLemmaMain*
\begin{proof}
    Let $B = (B_1, B_2, \ldots, B_n)$ be a completion of the partial allocation $A$, that assigns good $g$ to agent $i$, and let $B'=(B'_1,\ldots,B'_n)$ be the completion obtained from $B$ by reallocating $g$ from agent $i$ to agent $k$, that is, $B'_r = B_r$ for all $r \in N \setminus \{i,k\}$, $B'_i = B_i \setminus \{g\}$ and $B'_k = B_k \cup \{g\}$. Note that
    \begin{align*}
        \WH(B') - \WH(B) &= \Delta_k(v_k(B_k), v_k(g)) - \Delta_i(v_i(B_i \setminus \{g\}), v_i(g))\\
        &\geq \Delta_k(v_k(A_k) + v_k(Y), v_k(g)) - \Delta_i(v_i(A_i), v_i(g))\\
        &> 0,
    \end{align*}

    where the first inequality follows from $B_k \subseteq A_k \cup Y, A_i \subseteq B_i \setminus \{g\}$ and the fact that the marginal functions are decreasing in their first argument. The second inequality follows from the domination condition. Hence, $\WH(B') > \WH(B)$, and thus $B$ cannot be optimal.
    
\end{proof}

\subsection{Missing Proofs from Section~\ref{sec:identical}}
\label{sec:appendix_identical}

\IdenticalCardinality*
\begin{proof}
    Let $A^\star$ be any optimal solution. Suppose $i$ is an agent satisfying $|A^\star_i| \ge 2$ and $k$ is any other agent such that $v(A^\star_i) > v(A^\star_k)$. Let $g_i$ be the favorite good of agent $i$ in $A^\star_i$, and $g'$ be any other good in $A^\star_i$. Consider the partial allocation $A$ obtained from $A^\star$ by unallocating $g'$ from agent $i$. Clearly, $g'$ is the only unallocated good under $A$, and the set $Y$ of unallocated goods other than $g'$ is empty. Also, since the valuation function $v$ is lexicographic, $v(A_i) = v(A^\star_i \setminus g') \geq v(g_i) > v(A^\star_k) = v(A_k)$. Thus, by the strict concavity of $f$, $\Delta_i(v(A_i), v(g')) < \Delta_k(v(A_k), v(g'))$ since $\Delta_i = \Delta_k$. Hence, the domination condition(\Cref{eqn:domination_condition}) is satisfied and agent $k$ dominates agent $i$ with respect to good $g'$ under the partial allocation $A$. Thus, $A^\star$ cannot be optimal by the domination lemma (~\Cref{lem:domination}).
\end{proof}

\IdenticalStructure*
\begin{proof}
    If $m=n$,~\Cref{lem:identical} implies that every agent receives exactly one good, and the result follows. If $m > n$,~\Cref{lem:identical} implies that there is exactly one agent who receives two or more goods (if there were more than two agents, one of them must have a higher utility, contradicting~\Cref{lem:identical}). Also, this agent(say agent $i$) must have a lower utility than all other agents. Hence, agent $i$ must not receive any of the best $n-1$ goods, for otherwise, there would be an agent who does not receive any of the best $n-1$ goods and thus has a lower utility than agent $i$, contradicting~\Cref{lem:identical}. Thus, agent $i$ receives the worst $m-n+1$ goods, and the best $n-1$ goods are allocated one each to the remaining $n-1$ agents.
\end{proof}

\subsection{Missing Proofs from Section~\ref{subsec:Results_Ordered_Nash}} 
\label{sec:ordered_appendix}

\DominationLemma*
\begin{proof}
    Suppose neither of the two domination conditions hold. Let $T$ be the set of all unallocated goods in $A$ except $e$ and let $T'$ be the set of all unallocated goods in $A'$ except $e$. The negation of the domination condition~(\Cref{eqn:domination_condition}) gives us the following two inequalities:

    \begin{equation}
    \Delta_i(v_i(A_i), v_i(e)) \geq \Delta_k(v_k(A_k) + v_k(T), v_k(e)),
        \label{eqn:helper_1}
    \end{equation}

    \begin{equation}
        \Delta_k(v_k(A'_k), v_k(e)) \geq \Delta_i(v_i(A'_i) + v_i(T'), v_i(e)).
        \label{eqn:helper_2}
    \end{equation}

Now, let $U$ be the set of goods that were unallocated in $B$. As the instance is lexicographic and $g$ is the most valued good in $U$ for both agents, we have that
\begin{equation}\label{eqn:midi}
    v_i(A_i) \geq v_i(B_i \cup \{g\}) > v_i(B_i \cup U \setminus \{g\}) \geq v_i(A'_i \cup T') = v_i(A'_i) + v_i(T')
\end{equation}

where the third inequality follows from $A'_i \cup T' \subseteq B_i \cup U \setminus \{g\}$, since $A'_i \setminus B_i$ and $T'$ are both subsets of $U \setminus g$. Similarly, we have that
\begin{equation}\label{eqn:midk}
    v_k(A'_k) \geq v_k(B_k \cup \{g\}) > v_k(B_k \cup U \setminus \{g\}) \geq v_k(A_k \cup T) = v_k(A_k) + v_k(T).
\end{equation}

Hence, applying the inequalities in \Cref{eqn:midi} and \Cref{eqn:midk} to the strict concavity of $h_i$ and $h_k$ in \Cref{eqn:helper_1} and \Cref{eqn:helper_2} respectively, we get the following two inequalities:

\begin{equation}
    \Delta_i(v_i(A'_i) + v_i(T'), v_i(e)) > \Delta_i(v_i(A_i), v_i(e)) \geq \Delta_k(v_k(A_k) + v_k(T), v_k(e))
\end{equation}

\begin{equation}
    \Delta_k(v_k(A_k) + v_k(T), v_k(e)) > \Delta_k(v_k(A'_k)) \geq \Delta_i(v_i(A'_i) + v_i(T'), v_i(e)),
\end{equation}

a contradiction.
\end{proof}

\LeafBound*
\begin{proof}[Proof of \Cref{lem:leaf_bound}]
Fix a node $A$ at depth $j$, and recall that $M(A)$ has one row per remaining good $g_{j+1},\dots,g_m$ and one column per agent. We prove the claim by induction on the height of $A$ (i.e., on the number of goods remaining to be allocated).

\smallskip
\noindent\textbf{Base case.}
If $A$ is a leaf, then $\text{leaves}(A)=1$.  Also $M(A)$ has zero rows, hence $M(A)^r\cdot M(A)^s=0$ for all $r\neq s$, so the RHS equals $\prod_{r<s}(1+0)=1$.

\smallskip
\noindent\textbf{Inductive step.}
Assume the statement holds for all nodes of smaller height.  Consider $A$ (height $\ge 1$), and let $S\subseteq N$ be the set of agents that are \emph{not} pruned when allocating the next good $g_{j+1}$.  For each $i\in S$, let $C_i$ be the corresponding child. Clearly,
\begin{equation}\label{eq:leaves-sum-clean}
    \text{leaves}(A)=\sum_{i\in S}\text{leaves}(C_i).
\end{equation}
By induction, for each $i\in S$,
\begin{equation}\label{eq:IH-clean}
    \text{leaves}(C_i)\le \prod_{1\le r<s\le n}\bigl(1+M(C_i)^r\cdot M(C_i)^s\bigr).
\end{equation}
Let $M'(A)$ denote $M(A)$ with the row corresponding to $g_{j+1}$ removed. Every descendant of $C_i$ is a descendant of $A$, hence for every agent $r$,
\begin{equation}\label{eq:monotone-clean}
    M(C_i)^r \le M'(A)^r \qquad \text{(entry-wise)}.
\end{equation}
Therefore, for any pair $\{r,s\}$,
\[
1+M(C_i)^r\cdot M(C_i)^s \le 1+M'(A)^r\cdot M'(A)^s \le 1+M(A)^r\cdot M(A)^s ,
\]
Additionally, for pairs involving $i$ we keep the $i$-column from $M(C_i)$ and upper bound the other by $M'(A)$:
\[
1+M(C_i)^i\cdot M(C_i)^k \le 1+M(C_i)^i\cdot M'(A)^k \qquad (k\neq i).
\]
Plugging these bounds into \eqref{eq:IH-clean} yields
\begin{equation}\label{eq:child-bound-factorized}
\text{leaves}(C_i)\le
\Biggl(\prod_{1\le r<s\le n}\bigl(1+M(A)^r\cdot M(A)^s\bigr)\Biggr)
\cdot
\prod_{k \in N \setminus \{i\}} h_{i,k},
\end{equation}
where we define
\[
h_{i,k}\coloneqq
\frac{1+M(C_i)^i\cdot M'(A)^k}{1+M(A)^i\cdot M(A)^k}
\qquad (i\in S,\ k\in N\setminus\{i\}).
\]
Summing \eqref{eq:child-bound-factorized} over $i\in S$ and using
\eqref{eq:leaves-sum-clean}, we obtain
\begin{equation}\label{eq:reduce-to-sum-clean}
\begin{aligned}
\text{leaves}(A)&=\sum_{i\in S}\text{leaves}(C_i)\\
&\le \Biggl(\prod_{1\le r<s\le n}\bigl(1+M(A)^r\cdot M(A)^s\bigr)\Biggr) \cdot \sum_{i\in S}\ \prod_{k\in N\setminus\{i\}} h_{i,k}.\\
&\le \Biggl(\prod_{1\le r<s\le n}\bigl(1+M(A)^r\cdot M(A)^s\bigr)\Biggr) \cdot \sum_{i\in S}\ \prod_{k\in S\setminus\{i\}} h_{i,k}.
\end{aligned}
\end{equation}
where the final inequality uses $h_{i, k} \le 1$ for all $i \in S, k \in N \setminus \{i\}$. Thus it remains to prove
\begin{equation}\label{eq:sum-leq-1-target}
\sum_{i\in S}\ \prod_{k\in S\setminus\{i\}} h_{i,k}\ \le\ 1.
\end{equation}

Let us first prove that $h_{i,k}+h_{k,i}\le 1$ for all $i\neq k \in S$. Recall that
\[
h_{i,k}=
\frac{1+M(C_i)^i\cdot M'(A)^k}{1+M(A)^i\cdot M(A)^k} \qquad\text{and}\qquad h_{k,i}= \frac{1+M(C_k)^k\cdot M'(A)^i}{1+M(A)^i\cdot M(A)^k}.
\]
Let $D = 1+M(A)^i\cdot M(A)^k$ be the common denominator in $h_{i,k}$ and $h_{k,i}$. Because $i,k\in S$, the row corresponding to $g_{j+1}$ of $M(A)$ has a $1$ in both the columns $i$ and $k$, and $M'(A)$ denotes the matrix $A$ with the row corresponding to $g_{j+1}$ removed, we have
\[
D = 2 + M'(A)^i\cdot M'(A)^k.
\]
Now, fix any remaining good $g_{r}$ with $r>j+1$ and define the row-indicators
\[
x_r \coloneqq M'(A)_{g_r,i},\qquad y_r \coloneqq M'(A)_{g_r,k},\qquad
u_r \coloneqq M(C_i)_{g_r,i},\qquad v_r \coloneqq M(C_k)_{g_r,k}
\]
where $x_r,y_r,u_r,v_r\in\{0,1\}$. Since $M(C_i) \leq M'(A)$ and $M(C_k) \leq M'(A)$ entry-wise, we have:
\begin{equation}\label{eq:monotone-row}
u_r \le x_r
\qquad\text{and}\qquad
v_r \le y_r
\qquad\text{for all } r>j+1.
\end{equation}
Moreover, by \Cref{cor:mutual_exclusion_siblings} (mutual exclusion across sibling branches), for each $r>j+1$ we cannot have simultaneously: (i) in the subtree of $C_i$ there exists a node that gives $g_r$ to $i$, and (ii) in the subtree of $C_k$ there exists a node that gives $g_r$ to $k$. Equivalently, at most one of $u_r$ and $v_r$ can be $1$, that is, $u_r + v_r \le 1$. Observe that
\begin{equation}
u_r y_r + v_r x_r \le x_r y_r.
\label{eqn:dot_product_contribution}
\end{equation}
Indeed, if $x_r y_r = 0$, then either $x_r=0$ or $y_r=0$ and by \eqref{eq:monotone-row} we get $u_r y_r + v_r x_r \leq 2x_r y_r = 0$. If $x_r y_r = 1$, then $x_r=y_r=1$ and therefore $u_r y_r + v_r x_r = u_r+v_r \le 1 = x_r y_r$.
Summing inequality~\eqref{eqn:dot_product_contribution} over all rows $r>j+1$ yields
\[
M(C_i)^i\cdot M'(A)^k \;+\; M(C_k)^k\cdot M'(A)^i
\;\le\;
M'(A)^i\cdot M'(A)^k.
\]
Putting everything together, we get
\begin{equation}\label{eq:h-sum}
\begin{aligned}
h_{i,k}+h_{k,i}
&=
\frac{2 + \big(M(C_i)^i\cdot M'(A)^k\big) + \big(M(C_k)^k\cdot M'(A)^i\big)}{D}\\
&\le
\frac{2 + M'(A)^i\cdot M'(A)^k}{2 + M'(A)^i\cdot M'(A)^k}
= 1.
\end{aligned}
\end{equation}

\smallskip
\noindent\textbf{Proof of \eqref{eq:sum-leq-1-target}.} Let us generate a random directed graph as follows. Start with a complete undirected graph on the vertex set $S$. For each undirected edge $(i,k)$, do the following independently of other edges: direct it $i\to k$ with probability $h_{i,k}$, direct it $k\to i$ with probability $h_{k,i}$, and remove it with probability $1-h_{i,k}-h_{k,i}\ge 0$ (this is where we use $h_{i,k} + h_{k,i} \le 1$). For any $i\in S$, the probability that there exists an edge $i \to k$ for all $k \in S \setminus \{i\}$ in the random directed graph is exactly $\prod_{k\in S\setminus\{i\}} h_{i,k}$,
by independence. Moreover, two distinct vertices cannot both have all incident edges directed away from themselves. Hence
\[
\sum_{i\in S}\ \prod_{k\in S\setminus\{i\}} h_{i,k}
\]
is the probability of a disjoint union of events, and is therefore at most $1$. This proves \eqref{eq:sum-leq-1-target}. Plugging into \eqref{eq:reduce-to-sum-clean} gives the desired bound on $\text{leaves}(A)$ and completes the induction.
\end{proof}

\subsubsection{Proof of Theorem~\ref{thm:EPTAS}}
\label{sec:EPTASthmappendix}

\EPTASNashOrdered*

We begin with proving the following property of Nash optimal allocations under ordered lexicographic valuations:

\begin{lemma}
    In every Nash optimal allocation, the goods $g_1, g_2, \ldots g_n$ are allocated to distinct agents.
\end{lemma}

\begin{proof}
    Suppose not. Then, there exists an agent $i$ who receives at least two of the goods $g_1, g_2, \ldots g_n$. and an agent $j$ who receives none of these goods. Suppose agent $i$ receives goods $g_r$ and $g_s$ with $r < s \leq n$. Consider the allocation $A'$ obtained by transferring good $g_s$ to agent $j$. Because of the lexicographic nature of the valuations, the utility of agent $i$ reduces by at most half, while the utility of agent $j$ increases by at a factor strictly greater than $2$. Hence, the Nash welfare of $A'$ is strictly greater than that of $A$, a contradiction.
\end{proof}

Also, note that for any lexicographic valuation $v$, where the preference order of goods is $g_1 \succ g_2 \succ \ldots \succ g_m$, we have the following property:
\begin{equation}
    v(g_i) > 2^{t-1} (v(g_{i + t}) + v(g_{i + t + 1}) + \ldots v(g_m))
    \label{eq:lex_gap}
\end{equation}

for all $i$ and $t$ such that $i + t \leq m$. This can be easily seen through induction on $t$. For $t=1$, it follows through the definition of lexicographic valuations. For the inductive step, we have:
\begin{align*}
v(g_i) &> \sum_{j=i+1}^{m} v(g_j) = \sum_{j=i+1}^{i+t-1} v(g_j) + \sum_{j=i+t}^{m} v(g_j) \\
&> (\sum_{j=i+t}^{m} v(g_j)) \cdot (1 + 1 + 2^1 + \ldots 2^{t-2})\\
&= 2^{t-1} \cdot \sum_{j=i+t}^{m} v(g_j)
\end{align*}

Let $A^\star$ be the Nash optimal allocation. We know that each agent receives exactly one good from $\{g_1, g_2, \ldots g_n\}$ in $A^\star$. Additionally, we will guess (by enumerating over all $\O(n^\ell)$ choices) the agent who gets the good $g_r$ in $A^\star$, for all $r \in {n + 1, \ldots n+\ell}$, for some parameter $\ell$ to be determined later. For each agent $i \in N$, let $B_i$ be the set of goods among $\{g_{n+1}, \ldots g_{n+\ell}\}$ that are allocated to agent $i$ in $A^\star$. We will find a matching $\pi$ between the agents and the goods $\{g_1, g_2, \ldots g_n\}$ that maximizes the following:

\[
W'(\pi) = \prod_{i \in [n]} \bigl(v_i(\pi(i)) + v_i(B_i)\bigr)^{\nicefrac{1}{n}}
\]

\begin{algorithm}[t]
    \caption{\textsc{OrderedLexicographicNashEPTAS}}

    \DontPrintSemicolon
    \KwIn{An instance $\langle N, M, \V \rangle$ with ordered lexicographic valuations, parameter $\varepsilon > 0$.}
    \KwOut{An allocation $A'$ with $W^\texttt{\textup{Nash}}(A') \geq (1-\varepsilon) \cdot W^\texttt{\textup{Nash}}(A^\star)$.}

    \If{$n >  2/\varepsilon $}{
        $\ell \gets 0$\;
    }
    \Else{
        $\ell \gets \min(m-n, \lceil \log_2(2/\varepsilon) \rceil)$\;
    }

    $A^\star_{\text{best}} \gets \texttt{null}$\;
    $W_{\text{best}} \gets -\infty$\;

    \ForEach{assignment $(B_i)_{i \in N}$ of goods $\{g_{n+1}, \ldots, g_{n+\ell}\}$ to agents}{
        \tcc{There are at most $n^\ell$ such assignments}

        Construct the $n \times n$ weight matrix $W$ where $W_{i,j} = \log(v_i(g_j) + v_i(B_i))$ for $i \in N, j \in [n]$\;

        Find a maximum weight perfect matching $\pi: N \to \{g_1, \ldots, g_n\}$ using the Hungarian algorithm\;

        \tcc{$\pi(i)$ is the good from $\{g_1, \ldots, g_n\}$ assigned to agent $i$}

        Construct allocation $A'$: for each agent $i \in N$, set $A'_i \gets \{\pi(i)\} \cup B_i$\;

        Assign all remaining goods $R = \{g_{n+\ell+1}, \ldots, g_m\}$ to an arbitrary agent\;

        \If{$W^\texttt{\textup{Nash}}(A') > W_{\text{best}}$}{
            $W_{\text{best}} \gets W^\texttt{\textup{Nash}}(A')$\;
            $A^\star_{\text{best}} \gets A'$\;
        }
    }

    \Return $A^\star_{\text{best}}$\;
    \label{alg:OrderedLexicographicNashEPTAS}
\end{algorithm}
Such a matching can be found in polynomial time using the Hungarian algorithm. Let $R = \{g_{n + \ell + 1},\ldots g_{m}\}$ be the set of remaining goods. We will return the allocation $A'$ obtained by giving each agent $i$ the good $\pi(i)$ and all the goods in $B_i$, and giving all of $R$ to an arbitrary agent. The overall runtime will therefore be $\O(n^\ell \cdot (n^3 + nm))$, since we run the Hungarian algorithm for each of the $\O(n^\ell)$ choices for $B$. We present the pseudo-code of our algorithm in~\Cref{alg:OrderedLexicographicNashEPTAS}. For each agent $i\in N$, let $r_i^\star\in[n]$ be the index of the unique good from $\{g_1,\ldots,g_n\}$ allocated to $i$ in $A^\star$, and define the matching $\pi^\star$ by $\pi^\star(i):=g_{r_i^\star}$. Since these goods are allocated to distinct agents, $(r_i^\star)_{i\in N}$ is a permutation of $[n]$. From~\eqref{eq:lex_gap}, we get that $v_i(g_{r_i^\star}) \geq 2^{n+\ell-r_i^\star}v_i(R)$. Hence, we have:
\begin{align*}
    v_{i}(A^\star_i) &\leq v_{i}(g_{r_i^\star}) + v_{i}(B_{i}) + v_{i}(R)\\
    &\leq v_{i}(g_{r_i^\star}) + v_{i}(B_{i}) + \frac{v_{i}(g_{r_i^\star})}{2^{n+\ell-r_i^\star}}\\
    &\leq \bigl(v_i(g_{r_i^\star}) + v_i(B_i)\bigr) \cdot \bigl(1 + 2^{r_i^\star-n-\ell}\bigr)\\
\end{align*}
Using this inequality to upper bound the optimum, we get:
\begin{align*}
\W^\texttt{\textup{Nash}}(A^\star) &= \prod_{i \in N} v_i(A^\star_i)^{\nicefrac{1}{n}}\\
&\leq \prod_{i \in N} \bigl(v_i(g_{r_i^\star}) + v_i(B_i)\bigr)^{\nicefrac{1}{n}} \cdot \prod_{i \in N} \bigl(1 + 2^{r_i^\star-n-\ell}\bigr)^{\nicefrac{1}{n}}\\
&= W'(\pi^\star) \cdot \prod_{j \in [n]} \bigl(1 + 2^{j-n-\ell}\bigr)^{\nicefrac{1}{n}}\\
&\leq W'(\pi) \cdot \exp\bigg(\frac{1}{n} \cdot \sum_{j=1}^{n} 2^{j-n-\ell}\bigg)\\
&< W'(\pi) \cdot \exp\bigg(\frac{2^{1-\ell}}{n}\bigg)\\
&\leq \W^\texttt{\textup{Nash}}(A') \cdot \exp\bigg(\frac{2^{1-\ell}}{n}\bigg)
\end{align*}

If $n \geq \frac{2}{\varepsilon}$, then even $\ell = 0$ gives  $W^\texttt{\textup{Nash}}(A') > e^{-\varepsilon} \cdot W^\texttt{\textup{Nash}}(A^\star) > (1-\varepsilon) \cdot W^\texttt{\textup{Nash}}(A^\star)$. When $n$ is smaller than $\frac{2}{\varepsilon}$, we can set $\ell = \min(m - n, \lceil\log_2(\frac{2}{\varepsilon})\rceil)$\footnote{Note that, if $\ell = m - n$, our algorithm in fact finds the optimal Nash welfare, since then $v_i(R) = 0$ as $R = \emptyset$.} to get the same guarantee, and the overall runtime will be $(\nicefrac{1}{\varepsilon})^{\O(\log \nicefrac{1}{\varepsilon})} \cdot (n^3 + nm)$.

\subsection{Missing Proofs from Section~\ref{subsec:Doubling_Valuations}}
\label{sec:missing_proofs_doubling}

\GoodIdentification*
\begin{proof}
    We will prove that agent $s$ weakly dominates every agent $t \in N \setminus Z_{f_s}$ such that $f_s \in F_t$ with respect to $f_s$, under the partial allocation $B$. Note that for every such agent $t$, $f_t \neq f_s$. Since the valuations are doubling, we have $v_t(f_t) \geq 2 v_t(f_s)$. Also, by definition, each agent $r$ prefers its favorite good from $B_r$ over every good in $F_r$. Since the valuations are lexicographic (each doubling valuation is also lexicographic), $v_r(F_r \setminus \{f_r\}) < v_r(B_r)$ for all agents $r \in N$. Thus, we get:

    \begin{align*}
        \Delta_s(v_s(B_s \cup F_s \setminus \{f_s\}), v_s(f_s)) &> \Delta_s(2v_s(B_s), v_s(f_s))\\
        &\geq \Delta_t(2v_t(B_t), v_t(f_t))\\
        &\geq \Delta_t(2v_t(B_t), 2v_t(f_s))\\
        &\geq \Delta_t(v_t(B_t), v_t(f_s))\\
    \end{align*}
    where the first inequality follows from $v_s(F_s \setminus \{f_s\}) < v_s(B_s)$ and the strict concavity of $h_s$, the second inequality follows from the definition of $s$, the third inequality follows from the doubling property and the fact that $h_t$ is increasing, and the last inequality follows from the requirement $\Delta_t(2x, 2y) \geq \Delta_t(x, y)$ for all $x > 0$.

    Hence, every agent outside $Z_{f_s}$ that can be allocated $f_s$ is weakly dominated by agent $s$ with respect to $f_s$. Hence, in every optimal restricted completion of $B$, good $f_s$ is allocated to some agent in $Z_{f_s}$.
\end{proof}

\WeakDominationLemma*
\begin{proof}
    Suppose neither of the two domination conditions hold. Let $T$ be the set of all unallocated goods in $A$ except $e$ and let $T'$ be the set of all unallocated goods in $A'$ except $e$. Let $F_i$ denote the set of unallocated goods in $B$ that are allocable to agent $i$ and $F_k$ denote the set of unallocated goods in $B$ that are allocable to $k$. The set of goods except $e$ allocable to agent $k$ for the partial allocation $A$ is precisely $T \cap F_k\setminus \{g\}$. Similarly, the set of goods except $e$ allocable to agent $i$ for the partial allocation $A'$ is precisely $T' \cap F_i\setminus \{g\}$. The negation of the domination condition~(\Cref{eqn:weak_domination_2}) thus the following two inequalities:

    \begin{equation}
    \Delta_i(v_i(A_i), v_i(e)) \geq \Delta_k(v_k(A_k \cup (T \cap F_k\setminus \{g\})), v_k(e)),
        \label{eqn:helper_w1}
    \end{equation}

    \begin{equation}
        \Delta_k(v_k(A'_k), v_k(e)) \geq \Delta_i(v_i(A'_i \cup (T' \cap F_i\setminus \{g\})), v_i(e));
        \label{eqn:helper_w2}
    \end{equation}

As the instance is lexicographic and $g$ is the most valued good in $F_i$ that is allocable to $i$,
\begin{equation}\label{eqn:wmidi}
    v_i(A_i) \geq v_i(B_i \cup \{g\}) > v_i(B_i \cup F_i \setminus \{g\}) \geq v_i(A'_i \cup (T' \cap F_i \setminus \{g\}))
\end{equation}

where the third inequality holds since $A'_i \setminus B_i$ and $T' \cap F_i \setminus \{g\}$ are both subsets of $F_i \setminus \{g\}$. This is because on the path from $B$ to $A'$ in the recursion tree, only goods in $F_i$ can be allocated to agent $i$. Similarly, we have that
\begin{equation}\label{eqn:wmidk}
    v_k(A'_k) \geq v_k(B_k \cup \{g\}) > v_k(B_k \cup F_k \setminus \{g\}) \geq v_k(A_k \cup (T \cap F_k \setminus \{g\})).
\end{equation}

Hence, applying the inequalities in \Cref{eqn:wmidi} and \Cref{eqn:wmidk} to the strict concavity of $h_i$ and $h_k$ in \Cref{eqn:helper_w1} and \Cref{eqn:helper_w2} respectively, we get the following two inequalities:

\begin{equation}
    \Delta_i(v_i(A'_i \cup (T' \cap F_i\setminus \{g\})), v_i(e)) > \Delta_i(v_i(A_i), v_i(e)) \geq \Delta_k(v_k(A_k \cup (T \cap F_k\setminus \{g\})), v_k(e)),
\end{equation}

\begin{equation}
    \Delta_k(v_k(A_k \cup (T \cap F_k\setminus \{g\})), v_k(e)) > \Delta_k(v_k(A'_k), v_k(e)) \geq \Delta_i(v_i(A'_i \cup (T' \cap F_i\setminus \{g\})), v_i(e))
\end{equation}

a contradiction.
\end{proof}

\LeafBoundWeak*
\begin{proof}
    We proceed exactly as in the proof of \Cref{lem:leaf_bound}. The base case is again when the node $B$ is a leaf node. For the inductive step, let us look at the row corresponding to the good $g$ that is being allocated at the current node. Define $M'(A)$ to be the matrix $M(A)$ with the row corresponding to good $g$ removed. Let us similarly define $S$ to be the set of agents that are not pruned away by the weak domination condition. For each agent $i \in S$, let $C_i$ denote the child node corresponding to the partial allocation where $g$ was allocated to $i$. Same as in the proof of \Cref{lem:leaf_bound}, we have that for all agents $r$,

    \begin{equation}\label{eq:weak_monotone-clean}
    M(C_i)^r \le M'(A)^r \qquad \text{(entry-wise)}.
    \end{equation}

    Then, equations~\ref{eq:child-bound-factorized} and~\ref{eq:reduce-to-sum-clean} from~\Cref{lem:leaf_bound} hold and we again get:

    \[\text{leaves}(A) \le \Biggl(\prod_{1\le r<s\le n}\bigl(1+M(A)^r\cdot M(A)^s\bigr)\Biggr) \cdot \sum_{i\in S}\ \prod_{k\in S\setminus\{i\}} h_{i,k}.\]

    where

    \[h_{i,k} = \frac{1+M(C_i)^i\cdot M'(A)^k}{1+M(A)^i\cdot M(A)^k} \qquad (i\in S, k\in N\setminus\{i\})\]

    Next, in~\Cref{lem:leaf_bound}, we showed that $h_{i,k} + h_{k,i} \le 1$ for all $i \neq k \in S$. The only additional property (that was implied by~\Cref{cor:mutual_exclusion_siblings}) used in the proof of this key inequality was that for every $r > j + 1$, $u_r = 0$ or $v_r = 0$, or equivalently 

    \begin{equation}\label{eq:zero_dot_product}
        M(C_i)^i \cdot M(C_k)^k = 0.
    \end{equation}

    Note that~\Cref{cor:weak_mutual_exclusion_siblings} implies the exact same property in our case as well. Hence, the rest of the proof of \Cref{lem:leaf_bound} goes through with the only modification that $g_{j+1}$ (the good that was being allocated in the proof of~\Cref{lem:leaf_bound}) is replaced by $g$ (the good that is being allocated at the current node now).

    Once we have that $h_{i,k} + h_{k,i} \le 1$, we can use the exact same probabilistic argument as in the proof of \Cref{lem:leaf_bound} to finish the proof.
\end{proof}

\section{Missing Proofs from Section~\ref{sec:Hardness_Results}}
\label{sec:hardness_appendix}

\TransferLemma*
\begin{proof}
    Consider the partial allocation $A' = (A_1, \ldots, A_i \setminus \{g\}, \ldots, A_k, \ldots, A_n)$, where good $g$ is unallocated. The set $Y$ of unallocated goods in $A'$ other than $g$ is empty. Also, $v_k(A_k) > 0$ since each agent has a positive utility. The given inequality implies $v_i(A'_i) = v_i(A_i\setminus \{g\})$ is also positive. Maximizing the nash welfare is equivalent to maximizing the sum of log of utilities. Using $h_i(x) = \log x$, we get $\Delta_i(x, y) = \log(1 + y/x)$. Thus, we have:

    \begin{align*}
            \Delta_i(v_i(A'_i), v_i(g)) &= \log\left(1 + \frac{v_i(g)}{v_i(A_i\setminus \{g\})}\right)\\
            &< \log\left(1 + \frac{v_k(g)}{v_k(A_k)}\right)\\
            &= \Delta_k(v_k(A_k), v_k(g))\\
            &= \Delta_k(v_k(A'_k \cup Y), v_k(g))
    \end{align*}
    where the final equality uses $Y = \emptyset$ and $A'_k = A_k$. Thus, agent $k$ dominates agent $i$ with respect to good $g$ under the partial allocation $A'$. Hence, by~\Cref{lem:domination}, $A$ cannot be optimal.
\end{proof}

\OrderedLexReduction*
\begin{proof}
    For each agent \(t\in T\), define \(i(t)\), \(j(t)\), and \(k(t)\) so that \(t=(x_{i(t)},y_{j(t)},z_{k(t)})\);
    thus the signature goods of \(t\) are \(x_{i(t)},y_{j(t)},z_{k(t)}\).

    \smallskip
    \noindent(\(\Rightarrow\)) Suppose the \textsc{Perfect 3D Matching} instance has a perfect matching \(T'\subseteq T\) of size \(n\).
    Define an allocation \(A\) as follows: assign the \(|T|-n\) dummy goods \(D=\{d_1,\dots,d_{|T|-n}\}\) bijectively to the \(|T|-n\) agents in \(T\setminus T'\), and for each \(t\in T'\) give \(t\) the three signature goods \(\{x_{i(t)},y_{j(t)},z_{k(t)}\}\).
    This is feasible since \(T'\) covers each element of \(X\cup Y\cup Z\) exactly once.

    Agents in \(T\setminus T'\) contribute Nash-product factor \(\prod_{r=1}^{|T|-n} B^{10n+r}\).
    For \(t\in T'\), by construction
    \[
    \begin{aligned}
    v_t(\{x_{i(t)},y_{j(t)},z_{k(t)}\})
    &= v_t(x_{i(t)}) + v_t(y_{j(t)}) + v_t(z_{k(t)}) \\
    &= (v_t(x_{i(t)}) + v_t(y_{j(t)}))(1 + C\cdot B^{-20n+k(t)}) \\
    &= v_t(x_{i(t)})\Bigl(1 + C\cdot B^{-10n+j(t)}\Bigr)\Bigl(1 + C\cdot B^{-20n+k(t)}\Bigr) \\
    &= C\cdot B^{i(t)} \Bigl(1 + C\cdot B^{-10n+j(t)}\Bigr)\Bigl(1 + C\cdot B^{-20n+k(t)}\Bigr),
    \end{aligned}
    \]

    Since \(T'\) is a perfect matching, the indices \(i(t)\), \(j(t)\), and \(k(t)\) each range over \([n]\) exactly once (as \(t\) ranges over \(T'\)), hence the contribution of \(T'\) to the Nash welfare equals
    \[
    \left[
    \left(\prod_{r=1}^{n} C\cdot B^r\right)
    \left(\prod_{r=1}^{n} \bigl(1+C\cdot B^{-10n+r}\bigr)\right)
    \left(\prod_{r=1}^{n} \bigl(1+C\cdot B^{-20n+r}\bigr)\right)
    \right]^{1/n}
    \]
    Multiplying with the dummy-good factor gives a Nash welfare of exactly \(\theta\).
    
    \smallskip
    
    \noindent(\(\Leftarrow\)) Now, suppose that there exists an allocation \(B\) with Nash welfare at least \(\theta\). Then, the optimal allocation, say $A$, also has Nash welfare at least \(\theta\). We begin with some crucial observations about the structure of \(A\).

    \begin{claim}
        Any agent who receives a dummy good in $A$ must not receive any other good.
    \end{claim}

    \begin{proof}
        Consider an agent $t$ who receives a dummy good. Let $d$ be the most valuable dummy good assigned to agent $t$. For the sake of contradiction, assume that it also receives another good $g$. Since there are only $|T|-n$ dummy goods and $|T|$ agents, there must exist an agent $t'$ who does not receive any dummy good. If $g$ is also a dummy good, then transferring $g$ from $t$ to $t'$ reduces the utility of $t$ by at most a factor of $1/2$ and increases the utility of $t'$ by at least a factor of $B^{8n}$. Hence, the Nash product strictly increases, contradicting the optimality of $A$. Thus $g \in X \cup Y \cup Z$.

        We first record a uniform bound on relative values of non-dummy goods: for any good $g \in X\cup Y\cup Z$ and any two agents $p,q \in T$,
        \begin{equation}
            \frac{v_p(g)}{v_q(g)} \le B^{2n}.
            \label{eq:ratio_bound_nondummy}
        \end{equation}
        Indeed, if $g\in X$, then $v_u(g)\in\{B^r,\,C\cdot B^r\}$ for some $r\in[n]$, so
        $\frac{v_p(g)}{v_q(g)}\le C \le B \le B^{2n}$.
        If $g=y_s\in Y$, then for any agent $u$ we have
        $v_u(y_s)\in\{v_u(x_{i(u)})B^{-10n+s},\,v_u(x_{i(u)})C\,B^{-10n+s}\}$, and thus
        \[
            \frac{v_p(y_s)}{v_q(y_s)}
            \le
            C\cdot \frac{v_p(x_{i(p)})}{v_q(x_{i(q)})}
            \le
            C\cdot B^{n-1}
            \le
            B^{2n},
        \]
        using $C=3$ and $B = 10$. If $g=z_s\in Z$, its value is defined multiplicatively from $v_u(\{x_{i(u)},y_{j(u)}\})$ times a factor $B^{-20n+s}$ (and possibly an extra $C$ for the signature), so the same calculation as for $Y$ (together with $v_u(x_{i(u)}) \le v_u(\{x_{i(u)},y_{j(u)}\}) \le v_u(x_{i(u)})(1+CB^{-9n}))$) yields \eqref{eq:ratio_bound_nondummy}.
        
        \smallskip
        
        Now, since $d \in A_t \setminus \{g\}$, we have $v_t(A_t \setminus \{g\}) \ge v_t(d) \ge B^{10n+1}$. Since $t'$ has no dummy goods, its utility $v_{t'}(A_{t'})$ is at most the sum of values of all goods in $X \cup Y \cup Z$, which is less than $2 \cdot C \cdot B^n$. Hence, 

        $$
        \dfrac{v_t(A_t)}{v_{t'}(A_{t'})} \geq \dfrac{v_t(A_t \setminus \{g\})}{v_{t'}(A_{t'})} \geq \dfrac{B^{10n+1}}{2 \cdot C \cdot B^n} = \dfrac{B^{9n+1}}{2C} > B^{2n} \geq \dfrac{v_t(g)}{v_{t'}(g)}
        $$

        Hence, the condition of the transfer lemma(~\Cref{lem:transfer_lemma}) is satisfied, contradicting the optimality of $A$.
    \end{proof}

    From the above claim, the set of agents who receive a dummy good contribute exactly $(\prod_{r=1}^{|T|-n} B^{10n + r})^{\nicefrac{1}{|T|}}$ to the Nash welfare. This is precisely the first part of our threshold $\theta$. Let $R$ be the set of agents who do not receive any dummy good in $A$. Note that $|R| = |T| - (|T| - n) = n$. We will prove that $R$ must form a perfect 3D matching. First, we argue that each agent in $R$ must receive at most one (and hence exactly one, since $|X| = |R| = n$) good from $X$.

    \begin{claim} \label{clm:at_most_one_x}
        Each agent in \(R\) receives exactly one good from \(X\).
    \end{claim}

    \begin{proof}
        Say for the sake of contradiction that some agent $t$ receives at least two goods (say $x_{r_1}, x_{r_2}$) from $X$. There must exist some agent who receives no good from $X$. Without loss of generality, assume $r_1 < r_2$. Then, transferring $x_{r_1}$ to $t'$ decreases the utility of $t$ by at most a factor of $1/2$ and increases the utility of $t'$ by a factor of at least $B^{7n}$ (each value in $X$ is at least $B$, and the sum of all the values in $Y\cup Z$ is at most twice the value of $y_n$, which is at most $CB^{n} \cdot B^{-9n} \cdot C = C^2B^{-8n}$). Hence, the Nash product strictly increases, contradicting the optimality of $A$.
    \end{proof}

    Next, we argue that each agent in $R$ must receive its signature good from $X$.

    \begin{claim}
        Each agent in \(R\) receives its signature good from \(X\).
    \end{claim}

    \begin{proof}
       Note that for each agent $t \in R$,
       \[
       v_t(Y \cup Z) < 2 \cdot v_t(y_n) \leq 2 v_t(x_{i(t)}) \cdot CB^{-9n}
       \leq 2 (C v_t(x_1) B^{n-1}) C B^{-9n} \leq B^{-7n} \cdot v_t(x_1).
       \]

       From \Cref{clm:at_most_one_x}, each agent in $R$ receives exactly one good from $X$. Say each agent $t \in R$ receives the good $x_{s(t)}$. Then,
       \[
       v_t(A_t) \leq v_t(x_{s(t)})\left(1 + \frac{v_t(Y \cup Z)}{v_t(x_{s(t)})}\right) \leq v_t(x_{s(t)})\cdot (1 + B^{-7n}).
       \]

       Now, suppose that there is at least one agent $\ell \in R$ who does not receive its signature good from $X$, i.e., $s(\ell) \neq i(\ell)$. Recall that $v_t(x_r) = B^r$ if $r \neq i(t)$, else $CB^r$. Then, using the above inequality and the fact that $s$ is a permutation of $[n]$, we get that
       \begin{equation*}
       \begin{aligned}
            \prod_{t\in R} v_t(A_t) &\leq \left(\prod_{t \in R} v_t(x_{s(t)})\right) \cdot (1+B^{-7n})^n\\
            &\leq B^{s(\ell)} \cdot \left(\prod_{t \in R\setminus\ell} CB^{s(t)}\right)\cdot (1 + B^{-7n})^n\\
            &= \frac{1}{C} \cdot \left(\prod_{r=1}^{n} C\cdot B^{r}\right)\cdot \left(1+B^{-7n}\right)^{n}\\
            &< \left(\prod_{r=1}^{n} C\cdot B^{r}\right),
       \end{aligned}
       \end{equation*}
        where the last inequality follows from $(1 + B^{-7n})^n < 1 + 2^n \cdot B^{-7n} < 2$, and $C = 3$. However, this implies that the Nash product of all agents is strictly less than $\theta^{|T|}$, a contradiction.
    \end{proof}

    Now, we argue that each agent in $R$ must receive at most one (and hence exactly one) good from $Y$.

    \begin{claim} \label{clm:at_most_one_y}
        Each agent in \(R\) receives exactly one good from \(Y\).
    \end{claim}

    \begin{proof}
        Suppose not. Since $|R| = |Y| = n$, there exists some agent $t\in R$ who receives two distinct goods $y_{s_1},y_{s_2}\in Y$, and an agent $t' \in R$ who receives no good from $Y$.
        Since $t$ has two $Y$-goods, at least one of them is \emph{not} signature for $t$; without loss of generality, let $y_{s_1}$ be such a non-signature good in $A_t$ (so $s\neq j(t)$). Consider transferring the good $y_{s_1}$ from $t$ to $t'$ while keeping everything else fixed. Since $y_{s_1}$ is not the signature $Y$-good for $t$, we have $v_t(y_{s_1}) = v_t(x_{i(t)}) \cdot B^{-10n + s_1}$. Also, $v_{t'}(y_{s_1}) \ge v_{t'}(x_{i(t')}) \cdot B^{-10n + s_1}$. Hence, we have:

        \[
            \frac{v_t(y_{s_1})}{v_{t'}(y_{s_1})}
            \le
            \frac{v_t(x_{i(t)})\cdot B^{-10n+s_1}}{v_{t'}(x_{i(t')})\cdot B^{-10n+s_1}}
            =
            \frac{v_t(x_{i(t)})}{v_{t'}(x_{i(t')})}.
        \]
        Since $t$ receives two $Y$-goods, we have:
        \[
            v_t(A_t\setminus\{y_{s_1}\})
            \ge
            v_t(x_{i(t)})+v_t(y_{s_2})
            \ge
            v_t(x_{i(t)})\cdot\bigl(1+B^{-10n}\bigr),
        \]
        where we used $v_t(y_{s_2})\ge v_t(x_{i(t)})\cdot B^{-10n+s_2}\ge v_t(x_{i(t)})\cdot B^{-10n}$.
        For $t'$, since it receives no $Y$-good,
        \[
            v_{t'}(A_{t'})
            \le
            v_{t'}(x_{i(t')}) + v_{t'}(Z)
        \]
        
        Now, since the valuations are lexicographic and $B=10,C=3$, 
        \begin{equation}
        v_{t'}(Z) < 2v_{t'}(z_n) < 2C\,B^{-19n}\cdot v_{t'}(\{x_{i(t')},y_{j(t')}\})< B^{-18n}\cdot v_{t'}(x_{i(t')})
        \label{eq:Z_is_small}
        \end{equation}

        Thus, $v_{t'}(A_{t'}) \leq v_{t'}(x_{i(t')}) \cdot (1 + B^{-18n})$, and we get
        
        \[
            \frac{v_t(A_t \setminus y_{s_1})}{v_{t'}(A_{t'})} \geq \frac{v_t(x_{i(t)})}{v_{t'}(x_{i(t')})} \, \frac{1+B^{-10n}}{1+B^{-18n}} > \frac{v_t(x_{i(t)})}{v_{t'}(x_{i(t')})} \geq \frac{v_t(y_{s_1})}{v_{t'}(y_{s_1})}.
        \]
        Hence, the condition of the transfer lemma (\Cref{lem:transfer_lemma}) holds, contradicting the optimality of $A$.
    \end{proof}

    Next, we argue that each agent in $R$ must receive its signature good from $Y$.
    \begin{claim}
        Each agent in \(R\) receives its signature good from \(Y\).
    \end{claim}

    \begin{proof}
        Let $y_{s(t)}$ denote the unique good from $Y$ that agent $t \in R$ receives (this exists and is unique by \Cref{clm:at_most_one_y}). We show that $s(t)=j(t)$ for all $t\in R$. From~\eqref{eq:Z_is_small}, we know that for all $t \in R$,
            
            \[
                v_t(Z) < v_t(x_{i(t)}) \cdot B^{-18n} < v_t(\{x_{i(t)}, y_{s(t)}\})\cdot B^{-18n}.
            \]

            Hence for every $t\in R$,
            \[
            v_t(A_t) \le v_t(\{x_{i(t)}, y_{s(t)}\}) + v_t(Z) < v_t(\{x_{i(t)}, y_{s(t)}\})\cdot(1 + B^{-18n}).
            \]

            Now, 
            \begin{equation}
                \begin{aligned}
                    \displaystyle \prod_{t \in R} v_t(A_t) &< \left[\prod_{t \in R} (v_t(\{x_{i(t)},y_{s(t)}\})\right] \cdot (1 + B^{-18n})^n\\ 
                    &< \left[\prod_{t\in R}(v_t(\{x_{i(t)}, y_{s(t)}\})\right] \cdot (1 + B^{-17n})\\
                    &= \left[\prod_{t \in R} \left(v_t(x_{i(t)}) (1 + \frac{v_t(y_{s(t)})}{v_t(x_{i(t)})})\right)\right] \cdot (1 + B^{-17n})\\
                    &= \left[\prod_{t \in R} v_t(x_{i(t)}) \prod_{t \in R} \left(1 + \frac{v_t(y_{s(t)})}{v_t(x_{i(t)})}\right)\right] \cdot (1 + B^{-17n}) 
                    \label{eq:prod_bound_maxZ}
                \end{aligned}
            \end{equation}
            where the second inequality follows by substituting $x = B^{-17n}$ in $(1 + x)^n < 1 + 2^n x$ for all $x \in (0, 1]$, and using $B = 10 > 2$. Now, $\frac{v_t(y_{s(t)})}{v_t(x_{i(t)})} = B^{-10n+s(t)}$ if $s(t) \neq j(t)$ and $ C \cdot B^{-10n+s(t)}$ if $s(t) = j(t)$. Hence, every good $r$ contributes $1 + CB^{-10n + r}$ to the second product if it is assigned to some agent as its signature good, and $1 + B^{-10n + r}$ otherwise. Suppose for the sake of contradiction that there exists some agent $t' \in R$ for whom $s(t') \neq j(t')$. Then, we have:
            
            \[
            \prod_{t\in R}\left(1+\frac{v_t(y_{s(t)})}{v_t(x_{i(t)})}\right)
            \;\le\;
            \left(\prod_{r=1}^{n}\bigl(1+C\,B^{-10n+r}\bigr)\right)\cdot
            \frac{1+B^{-10n+s(t')}}{1+C\,B^{-10n+s(t')}}.
            \]
            
            Substituting this back into \eqref{eq:prod_bound_maxZ} gives us
            \[
            \prod_{t\in R} v_t(A_t)
            \;<\;
            \left(\prod_{t\in R} v_t(x_{i(t)})\right)\cdot
            \left(\prod_{r=1}^{n}\bigl(1+C\,B^{-10n+r}\bigr)\right)\cdot
            \frac{1+B^{-10n+s(t')}}{1+C\,B^{-10n+s(t')}}\cdot (1+B^{-17n}).
            \]

            Now, note that $\frac{1 + B^{-10n + r}}{1 + CB^{-10n + r}} = 1 - \frac{(C-1)B^{-10n + r}}{1 + CB^{-10n + r}} < 1 - B^{-11n}$ for all $r \ge 0, n \ge 1$ (recall that $B = 10, C = 3$). Also, $(1 - B^{-11n})(1 + B^{-17n}) < 1$. Thus, we get:

            \[
            \prod_{t\in R} v_t(A_t)<
            \left(\prod_{t\in R} v_t(x_{i(t)})\right)\cdot
            \left(\prod_{r=1}^{n}\bigl(1+C\,B^{-10n+r}\bigr)\right).
            \]
            Since each \(t\in R\) receives its signature $X$-good we have $\prod_{t\in R} v_t(x_{i(t)}) = \prod_{r=1}^{n} (C\cdot B^{r})$. Therefore,
            \[
            \prod_{t\in R} v_t(A_t)
            \;<\;
            \left(\prod_{r=1}^{n} C\cdot B^{r}\right)\left(\prod_{r=1}^{n}\bigl(1+C\,B^{-10n+r}\bigr)\right),
            \]
            and multiplying by the (already fixed) dummy-good contribution \(\prod_{r=1}^{|T|-n} B^{10n+r}\) implies that the Nash product of \(A\) is strictly smaller than \(\theta^{|T|}\), a contradiction. Hence \(s(t)=j(t)\) for all \(t\in R\), i.e., each agent in \(R\) receives its signature good from \(Y\).
        \end{proof}
    
    Now, we prove that each agent in $R$ must receive at most one (and hence exactly one) good from $Z$.

    \begin{claim}
        Each agent in $R$ receives at most one good from $Z$.
    \end{claim}

    \begin{proof}
        Suppose not. Since $|R| = |Z| = n$, there must exist some agent $t \in R$ who receives at least two goods from $Z$, say $z_{s_1}$ and $z_{s_2}$, and another agent $t' \in R$ who receives no good from $Z$. Since every agent has exactly one signature $Z$-good, at least one of $z_{s_1}, z_{s_2}$ must not be a signature good for agent $t$. Without loss of generality, let it be $z_{s_1}$. Thus, $v_t(z_{s_1}) = v_t(\{x_{i(t)}, y_{j(t)}\})\cdot B^{-20n+s_1}$. Also, $v_{t'}(z_{s_1}) \geq v_{t'}(\{x_{i(t')}, y_{j(t')}\}\}) \cdot B^{-20n+s_1}$. Thus,

        \[
        \frac{v_t(z_{s_1})}{v_{t'}(z_{s_1})} \leq \frac{v_t(\{x_{i(t)}, y_{j(t)}\})}{v_{t'}(\{x_{i(t')}, y_{j(t')}\})}
        \]

        Since agent $t$ contains two $Z$-goods,

        \[
        v_t(A_t \setminus \{z_{s_1}\}) \geq v_t(\{x_{i(t)}, y_{j(t)}\}) + v_t(z_{s_2}) > v_t(\{x_{i(t)}, y_{j(t)}\})
        \]

        On the other hand, agent $t'$ has no $Z$-goods. Thus,

        \[v_{t'}(A_{t'}) = v_{t'}(\{x_{i(t')}, y_{j(t')}\})\]

        Clearly,

        \[
        \frac{v_t(A_t \setminus \{z_{s_1}\})}{v_{t'}(A_{t'})} > \frac{v_t(\{x_{i(t)}, y_{j(t)}\})}{v_{t'}(\{x_{i(t')}, y_{j(t')}\})} \geq \frac{v_t(z_{s_1})}{v_{t'}(z_{s_1})}
        \]
        
        Hence, the condition of the transfer lemma(~\Cref{lem:transfer_lemma}) is satisfied, contradicting the optimality of $A$.
    \end{proof}

    Finally, we argue that each agent in $R$ must receive its signature good from $Z$.
    \begin{claim}
        Each agent in $R$ receives its signature good from $Z$.
    \end{claim}

    \begin{proof}
        Let $z_{s(t)}$ denote the unique good from $Z$ that agent $t \in R$ receives. We show that $s(t) = k(t)$ for all $t \in R$. Since each agent $t \in R$ now holds $\{x_{i(t)}, y_{j(t)}, z_{s(t)}\}$ (plus no other goods), we have
        \[
            v_t(A_t) = v_t(\{x_{i(t)}, y_{j(t)}\}) \cdot \left(1 + \frac{v_t(z_{s(t)})}{v_t(\{x_{i(t)}, y_{j(t)}\})}\right).
        \]
        Now, $\frac{v_t(z_{s(t)})}{v_t(\{x_{i(t)}, y_{j(t)}\})} = C \cdot B^{-20n + s(t)}$ if $s(t) = k(t)$, and $B^{-20n + s(t)}$ otherwise. Since $s$ is a permutation of $[n]$, if some agent $t' \in R$ has $s(t') \neq k(t')$, then
        \[
            \prod_{t \in R} \left(1 + \frac{v_t(z_{s(t)})}{v_t(\{x_{i(t)}, y_{j(t)}\})}\right)
            \le
            \left(\prod_{r=1}^{n} (1 + C \cdot B^{-20n+r})\right) \cdot \frac{1 + B^{-20n+s(t')}}{1 + C \cdot B^{-20n+s(t')}}.
        \]
        Clearly, $\frac{1 + B^{-20n+r}}{1 + CB^{-20n+r}} < 1$ and thus,
        \[
            \prod_{t \in R} v_t(A_t) < \left(\prod_{t \in R} v_t(\{x_{i(t)}, y_{j(t)}\})\right) \cdot \left(\prod_{r=1}^{n} (1 + C \cdot B^{-20n+r})\right).
        \]
        Since each $t \in R$ receives its signature goods from $X$ and $Y$, we have $\prod_{t \in R} v_t(\{x_{i(t)}, y_{j(t)}\}) = \prod_{r=1}^{n} (C \cdot B^r) \cdot \prod_{r=1}^{n} (1 + C \cdot B^{-10n+r})$. Multiplying by the dummy-good contribution shows the Nash product of $A$ is strictly less than $\theta^{|T|}$, a contradiction. Hence $s(t) = k(t)$ for all $t \in R$.
    \end{proof}
    Since each agent in $R$ receives its signature goods from $X$, $Y$, and $Z$, and the signature goods of the $n$ agents in $R$ cover $X \cup Y \cup Z$ exactly once, $R$ corresponds to a perfect 3D matching.
    \end{proof}
    
\subsection{Proof of Theorem~\ref{thm:APX-hard_Nash_General_Lexicographic}} \label{sec:apxhard_appendix}

\NashLexAPXHard*
    We give a reduction from 4D-matching to Nash Welfare under lexicographic preferences.
    \paragraph{4D-Matching.} The 4D-matching problem is defined as follows: Given four disjoint sets \(W, X, Y, Z\), each of size \(n\), and a set of quadruples \(Q \subseteq W \times X \times Y \times Z\). Any subset $Q' \subseteq Q$ is called a 4D-matching if no two quadruples in \(Q'\) share an element in any of the four dimensions. The goal is to find a 4D-matching of maximum size.

    We will utilize the following result from \citet*{10.1007/978-3-540-45198-3_8}, who prove hardness of distinguishing between the instances where at least $1-\epsilon$ fraction of the vertices can be matched and the instances where at most \(\frac{53}{54} + \epsilon\) fraction of the vertices can be matched, for any constant \(\epsilon > 0\). They give a reduction from gap-MAX-3-LIN-2[$\nicefrac{1}{2} + \epsilon, 1 - \epsilon$] to gap-4D-matching[$n(\frac{53}{54} + \epsilon), n(1 - \epsilon)$]. In the instances of 4D-matching that they use in their reduction, the number of quadruples is exactly $3$ times the number of elements in each set.

    \begin{theorem}[\citealp*{10.1007/978-3-540-45198-3_8}]
        For any constant $\epsilon > 0$, it is NP-hard to distinguish between the case when the size of the maximum 4D-matching is at least \(n(1 - \epsilon)\) and the case when it is at most \(n(\frac{53}{54} + \epsilon)\), even when the number of quadruples is exactly $3n$.
        \label{4d-matching-hardness}
    \end{theorem}
    \paragraph{The reduction.} Given an instance of 4D-matching with sets \(W, X, Y, Z\) each of size \(n\) and a set of quadruples \(Q\) of size $|Q| = m = 3n$, we construct an instance as follows. Let $W = \{w_1, w_2, \ldots, w_n\}$, $X = \{x_1, x_2, \ldots, x_n\}$, $Y = \{y_1, y_2, \ldots, y_n\}$, $Z = \{z_1, z_2, \ldots, z_n\}$ and $Q = \{q_1, \ldots q_m\}$. We create an agent $a_t$ for each quadruple $q_t \in Q$. Each element of $W \cup X \cup Y \cup Z$ corresponds to a good. Additionally, we create a set of $m - n$ dummy goods $D = \{d_1, d_2, \ldots, d_{m-n}\}$, and a set of $m$ supporting goods $S = \{s_1, s_2, \ldots, s_m\}$. Thus, the total number of goods is $4n + (m - n) + m = 4n + (3n - n) + 3n = 9n$, and the total number of agents is $m = 3n$. The dummy goods are meant to be the most valuable goods for each agent by a huge margin. For an agent $a_t$ corresponding to the quadruple $q_t = (w_i, x_j, y_k, z_\ell)$, we call the goods $w_i, x_j, y_k, z_\ell$ its \emph{signature goods}, and the good $s_t$ its \emph{supporting good}. We will call the remaining goods as \emph{insignificant goods} for agent $a_t$. Now, fix a sufficiently small constant $0 < q \ll 1$. The preference order and the values of the goods for an agent $a_t$ corresponding to the quadruple $q_t = (w_i, x_j, y_k, z_\ell)$ are as follows:

    $$ \underbrace{d_{m-n}}_{\frac{2^{19n + m}}{q}} \succ \ldots \underbrace{d_1}_{\frac{2^{20n + 1}}{q}} \succ \underbrace{w_i}_{16} \succ \underbrace{x_j}_{8} \succ \underbrace{y_k}_{4} \succ \underbrace{z_\ell}_{2} \succ \underbrace{s_t}_{1} \succ \underbrace{\text{insignificant goods}}_{< q}$$

    We choose an arbitrary order and assign values $q/2, q/4, q/8, \ldots$ to the insignificant goods in that order, so that the instance has lexicographic (in fact, also doubling) valuations, and the total value of insignificant goods is less than $q$. Note that this instance is not ordered, since different agents have agent-specific supporting and signature goods.

    If there exists a 4D-matching of size at least $n(1 - \epsilon)$ in the original instance, let us allocate each agent its supporting good. Then, for each quadruple in the matching, we additionally allocate the corresponding agent its signature goods. Some of the remaining agents are then allocated dummy goods (each dummy good to a different agent). Thus, the Nash welfare of this allocation is at least:

    \[\theta = \left( \left(\prod_{r=1}^{m-n} \frac{2^{20n + r}}{q}\right) \cdot (16+8+4+2+1)^{(1-\epsilon)n} \cdot 1\right)^{\nicefrac{1}{3n}}.\]

    Let $H$ denote $(\prod_{r=1}^{m-n} \frac{2^{20n + r}}{q})^{\nicefrac{1}{3n}}$. Thus, $\theta = H \cdot 31^{\frac{1 - \epsilon}{3}}$.

    Let $c = \frac{53}{54}$ be the constant from the hardness of 4D-matching~\citep*{10.1007/978-3-540-45198-3_8}. Now, suppose that in the original instance, any 4D-matching has size at most $(c + \epsilon) n$. The following lemma then proves~\Cref{thm:APX-hard_Nash_General_Lexicographic}. 

\begin{restatable}{lemma}{APXReduction}
    Assume $\epsilon, q$ to be sufficiently small positive constants. If the maximum size of a 4D-matching in the original instance is at most $(c + \epsilon)n$, then the maximum Nash welfare in the constructed instance is strictly smaller than $0.9996 \times \theta$.
    \label{lem:4d-matching-to-nash}
\end{restatable}

\begin{proof}
    Let $A$ be a Nash optimal allocation. We first claim that in $A$, each agent must receive at most one dummy good. Moreover, the agents receiving dummy goods must not receive any other goods. For this, note that if an agent receives two dummy goods, transferring the less valuable dummy good to an agent who does not receive any dummy good would strictly increase the Nash welfare, contradicting the optimality of $A$. Now, suppose an agent $t$ who receives a dummy good $d$ also receives another good $g$. Let $t'$ be an agent who does not receive any dummy good. Note that even after removing $g$ from $A_t$, agent $t$ has a utility of at least $\alpha_t \geq \frac{2^{20n + 1}}{q}$. On the other hand, agent $t'$ has a utility of at most $\alpha_{t'} < 16 + 8 + 4 + 2 + 1 + q < 32$. Also, $v_t(g) \leq 16$ since $g$ is not a dummy good, and $v_{t'}(g) > \frac{q}{2^{7n}}$ since there are fewer than $7n$ insignificant goods. Thus, we can apply the transfer lemma (\Cref{lem:transfer_lemma}) since:

    $$ \alpha_t \geq \frac{2^{20n + 1}}{q} > \frac{2^{7n} \cdot 32 \cdot 16}{q} > \alpha_{t'} \cdot \frac{v_t(g)}{v_{t'}(g)}$$

    Hence, no agent receiving a dummy good can receive any other good. The contribution of the agents who receive a dummy good to the Nash welfare is therefore exactly $H$. Now, consider the set (say $R$) of agents who do not receive any dummy good. Let $R_1$ be the set of agents who receive all their signature goods and $R_2$ be $R \setminus R_1$. Note that $|R| = n$. Clearly, $|R_1| \leq (c + \epsilon)n$. For each agent in $R_1$, its utility is less than $16 + 8 + 4 + 2 + 1 + q = 31 + q$. For each agent in $R_2$, its utility is less than $16 + 8 + 4 + 1 + q = 29 + q$, since it does not receive at least one of its signature goods. Therefore, the Nash welfare of $A$ is at most:

    \[
    \theta' = H \cdot (31 + q)^{(c + \epsilon)/3} \cdot (29 + q)^{(1 - c - \epsilon)/3}
    \]

    Note that

    \[\lim_{q, \epsilon \to 0^+} \dfrac{\theta'}{\theta} = \left(\dfrac{29}{31}\right)^{(1 - c) / 3} = \left(\dfrac{29}{31}\right)^{1 / 162} < 0.99959. \]
\end{proof}

\end{document}